%% file: bethjournal.tex
\documentclass{lmcs}
\pdfoutput=1

\usepackage{lastpage}
\lmcsdoi{20}{3}{7}
\lmcsheading{}{\pageref{LastPage}}{}{}%
{Dec.~16,~2022}{Jul.~22,~2024}{}

\pdfoutput=1 %
\usepackage[utf8]{inputenc} %


\usepackage{xspace}
\usepackage{color}
\usepackage{mathpartir}
\usepackage{booktabs}
\usepackage{cellspace}
\usepackage{tikz}
\usepackage{plabels}
\usepackage{float}
\usepackage{tabularx}
\usepackage{bussproofs}  \EnableBpAbbreviations
\usepackage{balance}
\usepackage{thm-restate}
\usepackage{upgreek}
\usepackage{color}
\usepackage[linesnumbered,vlined,ruled]{algorithm2e}
\SetAlFnt{\fontsize{9.5}{11}\selectfont} %
\usepackage{url}
\usepackage{hyperref}
\usepackage{cleveref}
\makeatletter
\newcommand{\leqnomode}{\tagsleft@true\let\veqno\@@leqno}
\newcommand{\reqnomode}{\tagsleft@false\let\veqno\@@eqno}
\makeatother
\input{macros}

\newenvironment{highlight}
{\begin{equation*}\begin{minipage}{0.85\textwidth}\it}
{\end{minipage}\end{equation*}}

\begin{document}

\title[Synthesizing Nested Relational Queries]
      {Synthesizing Nested Relational Queries from Implicit Specifications: Via Model Theory and via
        Proof Theory}

\author[M.~Benedikt]{Michael Benedikt\lmcsorcid{0000-0003-2964-0880}}[a]
\author[C.~Pradic]{C\'ecilia Pradic\lmcsorcid{0000-0002-1600-8846}}[b]
\author[C.~ Wernhard]{Christoph Wernhard\lmcsorcid{0000-0002-0438-8829}}[c]

\address{University of Oxford, UK}
\email{michael.benedikt@cs.ox.ac.uk}

\address{University of Swansea, UK}
\email{c.pradic@swansea.ac.uk}

\address{University of Potsdam, Germany}
\email{christoph.wernhard@uni-potsdam.de}

\begin{abstract}
\noindent Derived datasets can be defined implicitly or explicitly. An
\emph{implicit definition} (of dataset $O$ in terms of datasets $\vec I$) is a
logical specification involving two distinguished sets of relational symbols.
One set of relations is for the ``source data'' $\vec I$, and the other is for the
``interface data'' $O$. Such a specification is a valid definition of $O$ in terms of $\vec I$, if any two
models of the specification agreeing on $\vec I$ agree on $O$. In contrast, an
\emph{explicit definition} is a transformation (or ``query'' below) that produces $O$ from $\vec I$.
Variants of \emph{Beth's theorem} \cite{beth} state that one can convert
implicit definitions to explicit ones. Further, this conversion can be done
effectively given a proof witnessing implicit definability in a suitable proof
system.

We prove the analogous implicit-to-explicit result for nested relations:
implicit definitions, given in the natural logic for nested relations, can be
converted to explicit definitions in the nested relational calculus ($\nrc$).
We first provide a model-theoretic argument for this result, which makes some
additional connections that may be of independent interest, between $\nrc$
queries, \emph{interpretations}, a standard mechanism for defining
structure-to-structure translation in logic, and between interpretations and
implicit to definability ``up to unique isomorphism''. The latter connection
uses a variation of a result of Gaifman concerning ``relatively categorical'' theories. We also provide a
proof-theoretic result that provides an effective argument: from a proof
witnessing implicit definability, we can efficiently produce an $\nrc$
definition. This will involve introducing the appropriate proof system for
reasoning with nested sets, along with some auxiliary Beth-type results for
this system.
As a consequence, we can effectively extract rewritings of $\nrc$ queries in
terms of $\nrc$ views, given a proof witnessing that the query is determined
by the views.
\end{abstract}

\maketitle

\section{Introduction}

One way of describing a virtual datasource is via \emph{implicit definition}:
a  specification $\Sigma$ -- e.g. in logic --  involving
 symbols for the  ``virtual'' object $O$ and the stored ``input'' data $\vec I$. The specification
may mention other data objects (e.g. auxiliary views). But to be an implicit definition,
any two models of $\Sigma$ that agree on $\vec I$ must agree on $O$.
In the case where $\Sigma$ is in first-order logic,
this hypothesis can be expressed as a first-order entailment, using two copies of the vocabulary, primed
and unprimed, representing the two models:

\begin{equation}
  \tag{$\star$}
\Sigma \wedge \Sigma' \wedge \bigwedge_{I_i \in \vec I} \forall \vec x_i ~ [I_i(\vec x_i) \leftrightarrow I'_i(\vec x_i)] \models
\forall \vec x ~ [O(\vec x) \leftrightarrow O'(\vec x)]
\end{equation}

Above $\Sigma'$ is a copy of $\Sigma$ with primed versions of each predicate.

A fundamental result in logic states that we can replace an implicit definition
with an \emph{explicit definition}: a first-order query $Q$ such that whenever
$\Sigma(\vec I, O, \ldots)$ holds, $O=Q(\vec I)$. The original result of this kind
is \emph{Beth's theorem} \cite{beth}, which deals with classical first-order logic.
Segoufin and Vianu's \cite{svconf} looks at the case where
$\Sigma$ is in \emph{active-domain first-order logic}, or equivalently a Boolean
relational algebra expression. Their conclusion is that one can produce an
explicit definition of $O$ over $\vec I$ in  relational
algebra. \cite{svconf} focused on the special case where $\Sigma(I_1 \ldots I_j,
\vec B, O)$
specifies each $I_i$ as  a view defined
by an active-domain first-order formula $\phi_{V_i}$ over base
data $\vec B$, and also defines $O$ as an active-domain first-order query $\phi_Q$ over
$\vec B$. In this case, $\Sigma$ implicitly defining $O$ in terms of $\vec I$
is called ``determinacy of the query  by the views''.
Segoufin and Vianu's result implies that \emph{whenever  a relational algebra  query
$Q$ is determined by relational algebra views $\vec V$,
then  $Q$ is rewritable over the views by a relational algebra query}.

Prior Beth-style results like \cite{beth,svconf}  are effective. From a proof of the entailment $(\star)$
 in a suitable
proof system, one can extract an explicit definition effectively,
even in polynomial time.
While in early proofs of Beth's theorem, the proof systems were custom-designed
for  the task of proving implicit definitions, and the bounds were not stated,
later on
standard proof systems such as tableaux
\cite{smullyan} or resolution \cite{huang} were employed, and the polynomial claim was explicit.
It is important that in our definition of implicit definability, we require the
existence of a proof witness. By the completeness theorem for first-order logic,
requiring such a proof witness is equivalent to demanding that implicit definability
of $O$ over $\vec I$ holds for all instances, not just finite ones.

\emph{Nested relations} are a natural data model for hierarchical data.
Nested relations are objects within a type system built up from basic types via tupling
and a set-former.
In the 1980's  and 90's, a number of algebraic languages were proposed for defining
transformations on 
nested collections. 
Eventually a standard
language emerged, the
 \emph{nested relational calculus} ($\nrc$). 
The language is strongly-typed and  functional, with transformations 
built up via tuple manipulation
operations  as well as operators for lifting 
transformations over a type $T$ to transformations taking 
as input a set of objects of type $T$, such as singleton
constructors and a mapping operator.
One common formulation of these uses variables and
a ``comprehension'' operator for forming new objects from old ones
\cite{natj}, while  an alternative
algebraic formalism
 presents the language as a set of operators that can be freely composed. It was
shown that each $\nrc$ expression can be evaluated in  polynomial time in the size
of a finite data input, and that 
when the input and output is ``flat'' (i.e. only one level of nesting), $\nrc$ expresses exactly the 
transformations in the standard relational database  language relational algebra.
Wong's thesis \cite{limsoonthesis} summarizes the argument made by this line of work
 ``$\nrc$ can be profitably regarded
as the `right' core for nested relational languages''. 
$\nrc$ has been the basis
for most work on transforming nested relations. It is the basis for a number of 
commercial   tools \cite{dremel}, including those embedding nested data transformations  in
programming languages \cite{linq}, in addition to having influence in the effective implementation of data transformations in 
functional programming languages   \cite{gibbonshenglein,ringads}.

Although $\nrc$ can be applied to other collection types, such as bags and lists,
we will focus here on just nested sets.  We will  show a new connection between
$\nrc$ and first-order logic.

There is a natural logic for describing properties of nested relations, the well-known
\emph{$\deltazero$ formulas}, built up from equalities using
quantifications $\exists x \in \tau$ and $\forall y \in \tau$ where $\tau$ is a term.
For example, formula $\forall x \in c ~ \pi_1(x) \in \pi_2(x)$ might describe a property
of a nested relation $c$ that
is a set of pairs, where the first component of a pair is of
some type $T$  and the second component is a set containing elements
of type $T$.
A $\deltazero$ formula $\Sigma(\inobj^1 \ldots \inobj^k, \outobj)$ over 
variables $\inobj^1 \ldots \inobj^k$ and variable $\outobj$  thus defines
a relationship between 
$\inobj^1 \ldots \inobj^k$ and $\outobj$. For such a  formula to define a transformation it must
be \emph{functional}: it must enforce that $\outobj$ is determined by the values of 
$\inobj^1 \ldots \inobj^k$.
More generally, if we have a formula $\Sigma(\inobj^1 \ldots \inobj^k,  \outobj, \vec a)$, we say that 
$\Sigma$
\emph{implicitly defines $\outobj$ as a function of $\inobj^1 \ldots \inobj^k$} if:
\begin{equation}
  \label{eqtext-defines}
  \begin{minipage}{0.85\textwidth}\it
    For each two bindings $\sigma_1$ and $\sigma_2$
    of the variables $\inobj^1 \ldots \inobj^k, \vec a, \outobj$ 
    to nested relations satisfying $\Sigma$, if $\sigma_1$ and $\sigma_2$
    agree on each  $\inobj^i$, then  they agree on $\outobj$.
  \end{minipage}
\end{equation}

That is, $\Sigma$ entails that the value of $\outobj$ is a partial function of  the value
of $\inobj^1 \ldots \inobj^k$.

Note that when we say ``for each binding of variables to  nested relations''  in
the definitions above, we include infinite nested relations as well
as finite ones. 
An alternative characterization of $\Sigma$ being an implicit definition, which will be more relevant 
to us in the second part of the paper,  is that there is a proof that $\Sigma$ defines  a functional relationship.
Note that (\ref{eqtext-defines}) could be expressed as a first-order entailment: 
$\Sigma(\inobj^1 \ldots \inobj^k, \outobj, \vec a) \wedge \Sigma(\inobj^1 \ldots \inobj^k, \outobj', \vec a')
\models \outobj=\outobj'$
where in the entailment we omit some first-order ``sanity axioms''
about tuples and sets.
We refer to a proof of (\ref{eqtext-defines}) for a given $\Sigma$ and subset of
the input variables $\inobj^1 \ldots \inobj^k$, as a \emph{proof that $\Sigma$ implicitly defines
 $\outobj$ as a function
of  $\inobj^1 \ldots \inobj^k$}, or simply a \emph{proof of functionality}
dropping $\Sigma$, $\outobj$, and $\inobj^1 \ldots \inobj^k$ when they are clear
from context.
By the completeness theorem of first-order logic,  whenever $\Sigma$
defines $\outobj$ as a function of
$\inobj^1 \ldots \inobj^k$
according to the semantic definition above, this could be witnessed by a
proof in any of the standard complete proof calculi for classical first-order logic
(e.g. tableaux, resolution).
Such a proof will use  the sanity axioms referred to above,
which capture extensionality of sets, the compatibility
of the membership relation with the type hierarchy, and properties of projections and tupling. This notion of proof is only
presented as an illustration. In the second half of the paper, rather than using
a general-purpose first-order system, we will present more restrictive proof 
calculi that are tailored to reasoning about equivalence of  nested sets
relative to  $\deltazero$ theories.

\begin{exa} \label{ex:invertible}
We consider a specification in logic involving two nested collections, $F$ and $G$.
The collection $F$ is  of type $\sett(\ur \times \ur)$, where $\ur$
refers to the basic set of elements, the ``Ur-elements'' in the sequel.
That is, $F$ is a set of pairs. The collection $G$ is of
type $\sett(\ur \times \sett(\ur))$, a set whose
members are pairs, the first component an element
and the second a set.

Our specification $\Sigma$ will state that for each element $g$ in $G$
there is an element $f_1$ appearing as the first component of a pair in $F$, such that
$g$ represents $f_1$, in the sense that its first component is $f_1$ and
its second component accumulates all elements paired with $f_1$ in $F$.
 This can be specified easily by a $\deltazero$ formula:
\begin{align*}
\forall g \in G ~\exists f \in F \quad&  \pi_1(g)=\pi_1(f)~
\wedge~ \forall x \in \pi_2(g) ~  \<\pi_1(f), x \>  \in F \\
\wedge~~&\forall f' \in F ~~ \left[\pi_1(f')=\pi_1(f) \rightarrow \pi_2(f')  \in \pi_2(g)\right]
\end{align*}

$\Sigma$ also states that for each element $f_1$ lying  within 
a pair in $F$  there is a corresponding element $g$ of $G$ that pairs $f_1$
with all of the elements  linked  with $f$ in $F$. 
\begin{align*}
\forall f  \in F ~ \exists g \in G \quad& \pi_1(g)=\pi_1(f) ~\wedge~ \forall x \in \pi_2(g) ~ \<\pi_1(f), x \> \in F\\
\wedge~~& \forall f' \in F ~~
\left[\pi_1(f')=\pi_1(f) \rightarrow ~ \pi_2(f') \in \pi_2(g)\right]
\end{align*}

We can argue that in a nested relation satisfying
$\Sigma$, $G$ is a function of $F$. Thus $\Sigma$
implicitly defines a function from $F$ to $G$. 

We give the argument
informally here. Fixing $F,G$ and $F,G'$ satisfying $\Sigma$, we will
prove that if $g \in G$ then $g \in G'$. The proof begins by using the conjunct in
the first item to obtain an $f \in F$. We can then use the  
second  item on $G'$ to obtain 
a $g' \in G'$. 
We now need to prove that $g' = g$. Since $g$ and $g'$ are pairs, it suffices to show that their two projections are the same. We can easily see that $\pi_1(g)=\pi_1(f)=\pi_1(g')$, so it suffices
to prove $\pi_2(g')=\pi_2(g)$. Here we will make use of extensionality, arguing for 
containments between $\pi_2(g')$ and $\pi_2(g)$ in both directions.
In one direction we consider an $x \in \pi_2(g')$, and we need  to show
$x$ is in $\pi_2(g)$. By the second conjunct in the second item we have $\<\pi_1(f), x\> \in F$.
Now using the first item we can argue that $x \in \pi_2(g)$. 
In the other direction we consider $x \in \pi_2(g)$, we can apply the first
item to claim $\<\pi_1(f), x\> \in F$ and then employ the second item
to derive $x \in \pi_2(g')$.

Now let us consider $G$ as the input and $F$ as the output.
We cannot say that $\Sigma$ describes $F$ as a \emph{total} function of $G$, since 
$\Sigma$  enforces constraints on $G$:
that the second component of a pair in $G$ cannot be empty, and
that any two pairs in $G$ that agree on the first component must agree on the second.
But we can prove from $\Sigma$ that $F$ is a partial function of $G$: fixing
$F,G$ and $F',G$ satisfying $\Sigma$, we can prove that $F=F'$. 
\end{exa}

Our first contribution is to show
that whenever a $\deltazero$ formula
$\Sigma$ implicitly defines a function $\trans$, that 
function can be expressed in a slight variant of $\nrc$. 
The result can be seen as an analog of the well-known Beth definability
theorem for first-order logic \cite{beth}.

The argument that we employ in our first contribution will go through some
connections that give further insight. We first note that $\nrc$-expressible functions have the
same expressiveness as \emph{interpretations}, a standard way of defining structure-to-structure
transformations using logic. We will need a special notion of interpretation appropriate for nested
sets. We will then establish an equivalence between interpretations and transformations
that are \emph{implicitly definable up to unique isomorphism}. This equivalence will hold in a much more
general setting  of multi-sorted first-order logic.
The argument will be  model-theoretic and non-constructive,
relying on a variation of \emph{Gaifman's coordinatisation theorem} \cite{hodgesbook}. Putting these
two connections together will establish our result.

Our second contribution is an effective version. We will start
by  providing a proof  system 
 which we show is complete for entailments involving $\deltazero$
formulas over nested relations.
An advantage of our system compared to a
 classical first-order proof system mentioned above, is that
we do not require special axioms about sets, like extensionality.
In particular, we never need to reason about formulas that are not $\deltazero$.
We give two variations of the system, one that
is lower-level (and less succinct).
Our effective variant is then:
\begin{equation}
  \begin{minipage}{0.85\textwidth}\it
  From a proof $p$ that $\Sigma$ implicitly defines $o$
  in terms of $\vec i$, we can obtain, in $\ptime$,  an $\nrc$ expression $E$ that explicitly
  defines $o$ from $\vec i$, relative to $\Sigma$.
  \end{minipage}
\end{equation}

\medskip

The $\ptime$ claim refers to the lower-level calculus.
The proof of the effective Beth result uses completely distinct techniques from the model-theoretic 
argument. We synthesize $\nrc$ expressions directly, rather than going through interpretations.
The key to our proof-theoretic analysis is an auxiliary result about equalities between $\deltazero$
expressions, the $\nrc$ Parameter Collection Theorem, Theorem~\ref{thm:mainflatcasebounded}.
It is a kind of interpolation theorem for parameterized definability. Roughly speaking, it
says that when the conjunction of two $\deltazero$  formulas proves that two variables are equal, then
there must be an expression definable with parameters that are common to both formulas that sits between them.
We also make use of the interplay between reasoning about equivalence of nested sets ``up to extensionality'',
without extensionality axioms, and the equality of nested sets in the presence of extensionality axioms.

A special case of our results will be in the case of $\nrc$ views and queries.

\begin{exa} \label{ex:views}
We consider  a variation of Example \ref{ex:invertible},  where our  specification $\Sigma(Q,V,B)$ 
describes  a view  $V$, a query $Q$, as well as some constraints on
the base data $B$.
Our base data $B$ is
of type $\sett(\ur \times \sett(\ur))$, 
where $\ur$
refers to the basic set of elements, the ``Ur-elements''.
That is, $B$ is a set of pairs, where the first item is a data item
and the second is a set of data items.  
View  $V$ is of type $\sett(\ur \times \ur)$, a set of pairs, given
by the query that is the usual ``flattening'' of $B$: in $\nrc$ this
can be expressed as $\{ \< \pi_1(b) , c \> \mid c \in \pi_2(b) \mid b \in B\}$.
The view  definition can be converted to a  specification in our logic.

A query $Q$ might ask for a selection of the pairs in $B$, those whose first
component is contained in the second:
$\{ b \in B ~ | ~ \pi_1(b) \in \pi_2(b) \}$.
The definition of $Q$ can also be incorporated
into our specification.

View $V$ is not sufficient to answer $Q$ in general. This is
the case if 
we assume  as part of $\Sigma$ an integrity constraint stating that the 
first component of $B$ is a key.
We can argue that 
$\Sigma(Q,V,B)$ implicitly defines $Q$ in terms of $V$.

From first our main contribution, it follows that there is an $\nrc$ rewriting of $Q$ in terms of $V$.
From our second contribution, it follows that there is a polynomial time function
that takes as input a proof in a certain proof system formalizing
the implicit definability of $Q$ in terms of $V$,
which produces as output an $\nrc$  rewriting of $Q$ in terms of $V$.
\end{exa}

Overall our results show a close connection between logical specifications
of transformations on nested collections and the functional
transformation language $\nrc$, a result which is not anticipated
by the prior theory.

\myparagraph{Organization}
After a discussion of related work in Section \ref{sec:related},
we provide preliminaries in Section \ref{sec:prelims}. 
Section \ref{sec:bethmodeltheoretic} presents our first main result, the expressive equivalence
of implicit and explicit definitions, proven using model-theoretic techniques. As mentioned above, this
requires an excursion into the relationship between $\nrc$ and a notion of first-order interpretation.
Section \ref{sec:effective} presents the effective version, which relies on a proof system
for reasoning with nested relations.
We close with discussion in Section \ref{sec:conc}.
Some proofs of a routine nature, as well as some auxiliary results,  are deferred to the appendix.

\section{Related work} \label{sec:related}
In addition to the theorems of  Beth and Segoufin-Vianu mentioned in the introduction,
there are numerous works on effective Beth-style results for other logics.
Some concern fragments of classical first-order logic, such as the guarded fragment
\cite{HMO,guardedinterpj};
others deal with non-classical logics such as description logics \cite{balderbethdl}.  The Segoufin-Vianu result is
closely related to variations of Beth's theorem and Craig interpolation  for relativized quantification, such as Otto's
interpolation theorem
\cite{otto}.
There are also effective interpolation and definability results for logics richer
than or incomparable to
first-order logic,
such as fragments of fixpoint logics \cite{interpolationmucalc,
lmcsusinterpolfixedpoint}. There are even Beth-style results for full infinitary logic
\cite{lopezescobar}, but there one can not hope for effectivity.
The connection between Beth-style results and view rewriting originates in \cite{svconf,NSV}.
The idea of using effective Beth results to generate view rewritings from proofs
appears in \cite{franconisafe}, and is explored in more detail first in
\cite{tomanweddell} and later in \cite{interpbook}.

Our first main result relates to Beth theorems  ``up-to-isomorphism''.  Our implicit
definability hypothesis is that two models that satisfy a  specification and
agree on the inputs must agree on the output nested relations, where ``agree on the output''
means up to extensional equivalence
of sets, which is a special (definable) kind of isomorphism.
Beth-like theorems up to isomorphism originate in Gaifman's \cite{gaifman74}.  are studied
extensively by Hodges and his collaborators (e.g. \cite{hodgesnormal,hodgesdugald,hodgesbook}).
  The focus of these works is model-theoretic, with emphasis
on  connections with categoricity and classification in classical model theory.
More specifically, 
\cite{hodgesbook} defines
the notion of \emph{rigidly relatively categorical} which is the single-sorted
analog of the notion of  implicitly interpretable which we will introduce in the model-theoretic part of this work. 
\cite{hodgesbook} does not prove
any  connection of this notion to explicit interpretability, although he proves the equivalence
with a related notion called ``coordinatisability''.
Most  of the ingredients in our main model-theoretic arguments are present in his exposition.
The later unpublished draft \cite{madarasz} extends these ideas to a multi-sorted setting, 
but without full proofs.

\myparagraph{Conference versions}
Our first contribution comes from the conference paper \cite{benediktpradicpopl}. 
Our second result derives from the conference paper \cite{bethpods}.

\section{Preliminaries} \label{sec:prelims}

\subsection{Nested relations}
We deal with schemas that describe objects of various
 \emph{types} given by the following grammar.
$$T, \; U \bnfeq \ur \bnfalt T \times U \bnfalt \unit \bnfalt \sett(T)$$
For simplicity throughout the remainder
we will assume only two basic types. There is the one-element type $\unit$, which will
be used to construct Booleans.
And there is  $\ur$, the ``scalars'' or \emph{Ur-elements} whose inhabitants are not specified further.
From the Ur-elements and a unit type we can build up the set of types via
product and the power set operation.  
We use standard conventions for abbreviating types, with the $n$-ary product abbreviating
an iteration of binary products.
A \emph{nested relational schema} consists of declarations of 
variable names associated to objects of given types. 

\begin{exa} \label{ex:aschema}
An example nested relational schema declares two objects
$R: \sett(\ur \times \ur)$ and $S: \sett( \ur \times \sett(\ur))$.
That is, $R$ is a set of pairs of Ur-elements: a standard ``flat'' binary relation.
$S$ is a collection of pairs whose first elements are Ur-elements
and whose second elements are sets of Ur-elements.
\end{exa}

The types have a natural interpretation.
The unit type has a unique member and the members of $\sett(T)$
are the sets of members of $T$.
An \emph{instance} of such a schema is defined in the obvious way.

For the schema in Example \ref{ex:aschema} above, assuming that $\ur = \mathbb{N}$, one possible instance has
$R = \{ \<4,  6\>, \<7,  3\>\}$ and 
$S = \{ \<4 , \{  6,  9 \} \> \}$.

\input{deltazero}

\input{nrc}

\section{Synthesizing via model theory: the expressive equivalence\texorpdfstring{\\}{} of $\nrc$, interpretations, and implicit definitions} \label{sec:bethmodeltheoretic}

Our first result will show the expressive equivalence of several specification languages
for transformations. We will show that $\nrc$ expressions are equivalent to implicit definitions,
but in the process we will show that both transformation languages are equivalent
to transformations given in a natural logical language which we call $\deltazero$ interpretations.
While the equivalence between interpretations and $\nrc$ will be effective, the key direction from
implicit definitions to interpretations will be a model-theoretic argument.

\subsection{Statement of the first result: equivalence between implicit and explicit} \label{subsec:nestedbethstatement}

We consider an input  schema $\inschema$ with one
input object $\inobj$ and an output schema
with one output object $\outobj$. Using product objects, we can
easily model any nested relational transformation in this way.
We deal with a $\deltazero$ formula $\phi(\inobj, \outobj, \vec a)$
with distinguished variables $\inobj, \outobj$.
Recall from the introduction that such  a formula
\emph{implicitly defines $\outobj$ as a function of $\inobj$} if for each nested relation
$\inobj$ there is at most
one $\outobj$ such that $\phi(\inobj, \outobj, \vec a)$ holds for some
$\vec a$.
A formula $\phi(\inobj, \outobj, \vec a)$ \emph{implicitly defines a function $\trans$ from $\inobj$ to $\outobj$}
if  for each $\inobj$, $\phi(\inobj, \outobj, \vec a)$ holds for some
$\vec a$ if and only $\trans(\inobj)=\outobj$.

\begin{exa}\label{ex:impnrc}
Consider the transformation $\trans_{\groupq}$ from Example \ref{ex:nrc}.
It has a simple implicit $\deltazero$ definition, %
which we can restate as follows. First, define the auxiliary formula
$\chi(x, p,R)$ stating that $\pi_1(p)$ is $x$ and $\pi_2(p)$ is the set of the $y$s such
that $\<x, y\>$ is in $R$ (the ''fiber of $R$ above $x$''):
$$\chi(x,p,R) ~~~\eqdef~~~ \pi_1(p) = x~~ \wedge~~ \left( \forall t' \in R~
\left[\pi_1(t') = x ~\imp~ \pi_2(t') \in \pi_2(p)  \right]\right) ~~\wedge~~
\forall z\in \pi_2(p)~ \< x,z \>  \in R$$
Then the transformation $T_{\groupq}$ from $Q$ to $R$ is implicitly defined
  by \[\forall t \in R~~\exists p \in Q~~\chi(\pi_1(t),p,R)) \wedge \forall p
  \in Q~~ \chi(\pi_1(p),p,R)\]
\end{exa}

We can now state  our first main result:

\begin{thm}[Implicit-to-explicit for nested relations via model theory] \label{thm:bethinterp}    
For any
$\deltazero$ formula $\Sigma(\inobj, \outobj, \vec a)$ which
implicitly defines  $\outobj$ as a function of $\inobj$, there is an $\nrc$ expression 
$E$ such that whenever $\Sigma(\inobj, \outobj, \vec a)$ holds,
then $E(\inobj)= \outobj$. 
\end{thm}

In particular, suppose that, in addition, for each $\inobj$ there is some $\outobj$ and $\vec a$ such that
$\Sigma(\inobj, \outobj, \vec a)$ holds:  then the expression
and the formula define the same transformation.

Recall that our notion of  implicit definitions allows extra parameters $\vec a$.
Sometimes these are called ``projective'' implicit definitions in the literature.
From Theorem \ref{thm:bethinterp} we easily see that no additional expressiveness
is gained by allowing parameters:

\begin{cor} \label{cor:nrcbeth}
The following are equivalent for a transformation $\trans$:

\begin{enumerate}
\item $\trans$ is implicitly definable by a $\deltazero$ formula
\item $\trans$ is implicitly definable by a $\deltazero$ formula $\phi(\inobj, \outobj)$
\item $\trans$ is $\nrcwget$ definable
\end{enumerate}
\end{cor}

\myparagraph{Finite instances versus all instances}
In Theorem \ref{thm:bethinterp} and Corollary \ref{cor:nrcbeth} we emphasize that our 
results concern the
class $\funall$ of transformations $\trans$ such  that there is a $\deltazero$ formula 
$\Sigma$  which
defines a functional relationship between  $\inobj$ and $\outobj$
on all instances, finite and infinite, and where the function
agrees with $\trans$. We can consider $\funall$ as a class of transformations on all instances or on all
finite instances, but the class  is defined by reference to all instances
for $\inobj$.  Expressed semantically
\[
\Sigma(\inobj, \outobj, \vec a) \wedge  \Sigma(\inobj,\outobj', \vec a') \models \outobj'= \outobj
\]
An  equivalent characterization of $\funall$  is \emph{proof-theoretic}: these are the 
transformations such that there is a  classical proof of functionality in a complete first-order proof system.
There are various choices for the system. The most straightforward choice would be a system
using some basic axioms about Ur-elements, products and projection functions, and the extensionality axiom
for the membership relation. We will see a slightly different approach to proof systems in Section \ref{sec:effective}.

Whether one thinks of $\funall$ semantically or proof-theoretically, our results say that 
$\funall$ is identical with the set of transformations  given by $\nrc$ expressions.
But the proof-theoretic perspective is crucial in order to even talk about an  effective synthesis procedure.

It is natural to  ask about the analogous class  $\funfin$ 
of transformations $\trans$ over \emph{finite inputs} for which there is a $\deltazero$ 
$\Sigma_\trans$ which
is functional, when only finite inputs are considered,  and where the corresponding function  agrees with  $\trans$.
It is well-known that $\funfin$ is not identical to $\nrc$ and is not so well-behaved. 
The transformation returning the powerset of a given input relation $\inobj$ is in $\funfin$: the powerset of
a finite input $\inobj$ is the unique collection $\outobj$ of subsets of $\inobj$
  that contains the empty set and such that for each element $e$ of $\inobj$,
if a set $s$ is in $\outobj$ then $s-\{e\}$ and $s \cup \{e\}$ are in $\outobj$.
From this we can see that $\funfin$
contains transformations of high complexity. Indeed, even when considering transformations from flat relations
to flat relations, $\funfin$ contains transformations whose membership in polynomial time  would imply that $\kw{UP} \cap \kw{coUP}$, the class of problems such that both the problem
and its complement  can be solved by an unambiguous non-deterministic polynomial time  machine,
is identical to $\ptime$ \cite{kolaitisimpdef}.
Most importantly for our goals, membership in $\funfin$ is \emph{not} witnessed by proofs in any effective proof system,
since this set is not computably enumerable.

\myparagraph{Total versus partial functions}
When we have a proof  that
$\Sigma(\inobj, \outobj, \vec a)$ defines $\outobj$ as a function of $\inobj$,
the corresponding function may still be partial. 
Our procedure will synthesize an expression $E$ defining a total function
that agrees with the  partial function defined by $\Sigma$.
If $\vec a$ is empty, we can also synthesize a Boolean $\nrc$ expression $\verify_\kw{InDomain}$
that verifies whether a given $\inobj$ is in the domain of the function: that is whether
there is $\outobj$ such that $\Sigma(\inobj, \outobj)$ holds.
$\verify_\kw{InDomain}$ can be taken as:
$$\bigcup \{ \verify_{\Sigma}(\inobj, e) \mid e \in \{ E(\inobj)\}\}$$
where $\verify_\Sigma$ is from Proposition \ref{prop:verify}.

When $\vec a$ is not empty we can not generate a domain check $\verify_\kw{InDomain}$, since
the auxiliary parameters might enforce some second-order property of
$\vec i$: for example $\Sigma(i_0, i_1, a, o)$ might state that $a$ is a bijection
from $i_0$ to $i_1$ and $o=\<i_1, i_2\>$. This clearly defines
a functional relationship between $i_1, i_2$ and $o$, but
the domain consists of $i_1, i_2$ that have the same cardinality, which cannot
be expressed in first-order logic.

\myparagraph{Organization of the proof of the theorem}
Our proof of Theorem \ref{thm:bethinterp} will go through a notion of $\deltazero$ interpretation, 
which we introduce in Subsection \ref{subsec:fointerp}. We will show that $\deltazero$ interpretations
define the same transformations as $\nrc$, which will allow us to restate the main result
as moving from implicit definitions to interpretations.
We then proceed first by some reductions
(Subsection \ref{subsec:reduction}),
showing that it suffices to prove a general result about implicit
definability and definability by interpretations
in  multi-sorted first-order logic, rather than
dealing with higher-order logic  and $\deltazero$ formulas.
In Subsection \ref{subsec:proofmultisorted} we give the argument for this
multi-sorted logic theorem.

\myparagraph{Interpolation for $\deltazero$ formulas} Often a key ingredient in 
moving from implicit to explicit definition is  an \emph{interpolation theorem},
stating that for each entailment between formulas $\phi_L$ and $\phi_R$ there is an intermediate formula
(an \emph{interpolant} for the entailment), 
which is entailed by $\phi_L$ and entails $\phi_R$ while using only symbols common to $\phi_L$ and
$\phi_R$.
We can show using any of the standard  approaches to interpolation (e.g. \cite{fittingbook}) that 
$\deltazero$ formulas admit
interpolation.

\begin{prop}
\label{prop:interpolationmodeltheoretic}
Let $\Gamma_L$, $\Gamma_R$, and $\psi$  be $\deltazero$ formulas
and call $C = \freevars(\Gamma_L) \cap (\freevars(\Gamma_R) \cup \freevars(\psi))$ the set of common free variables of
$\Gamma_L$ on the one hand and $\Gamma_R$ or $\psi$ on the other hand.
If we have an entailment $$\Gamma_L, \; \Gamma_R \models \psi$$
then there exists a $\deltazero$ formula $\theta$ with $\freevars(\theta) \subseteq C$ such that the following holds
$$\Gamma_L \models \theta \qquad \qquad \text{and} \qquad \qquad \Gamma_R, \theta \models \psi$$
\end{prop}

A stronger effective statement -- stating that the interpolant can be found efficiently in the size
of a proof of the entailment --  will be proven later in the paper: see Theorem
\ref{thm:interpolationdeltafocus}.

The interpolation result above should be thought of  as giving us the result we want for 
implicit definitions of \emph{Boolean} variables.
From it we can derive that whenever we have 
a $\deltazero$ $\Sigma(\vec i \ldots o)$ which implicitly defines
a Boolean variable
$o$ in terms of input
variables $\vec i$, there  must be a $\deltazero$ $\phi(\vec i, o)$ 
that defines $o$  in terms of $\vec i$. Setting $o$ to be true we get
a formula $\phi'(\vec i)$ that defines the inputs that correspond to true.
By Proposition \ref{prop:verify}, there is an
 $\nrc$ expression outputting a Boolean that explicitly defines $o$.

\subsection{Interpretations and nested relations} \label{subsec:fointerp}
Our first goal will be to show that for any $\deltazero$ implicit definitions there is an
$\nrc$ query that realizes it.
For this result, it will be useful to have  another characterization of $\nrc$,
an equivalence with transformations defined by \emph{interpretations}.

We first review the notion of an interpretation, which has become
a common way of defining transformations using logical expressions \cite{interpmikolaj,interpcolcombet}.
Let $\inschema$ and $\outschema$ be  multi-sorted vocabularies.
A first-order interpretation with input signature $\inschema$ and output
signature $\outschema$ consists of:

\begin{itemize}
\item for each output sort $\sort'$, a sequence of input sorts $\interpsort(\sort') = \vec{\sort}$,
\item a formula $\phi^{\sort'}_\equiv(\vec x_1, \vec x_2)$ for each output sort $\sort'$ in $\outschema$ (where both tuples of variables $\vec x_1$ and $\vec x_2$ have types $\interpsort(\sort')$),
\item a formula $\phi^{\sort'}_\domainof(\vec x_1)$ for each output sort $\sort'$ in $\outschema$ (where the variables $\vec x_1$ have types $\interpsort(\sort')$),
\item a formula $\phi_R(\vec x_1 , \ldots \vec x_n)$ for every relation $R$
of arity $n$ in  $\outschema$ (where the variables $\vec x_i$ have types $\interpsort(\sort'_i)$, provided the $i^{th}$ argument of $R$ has sort $\sort'_i$),
\item for every function symbol $f(x_1, \ldots, x_k)$ of $\outschema$ with output sort $\sort'$ and input $x_i$ of sort $\sort_i$, a sequence of terms
  $\overline{f}_1(\vec x_1, \ldots, \vec x_k), \ldots, \overline{f}_m(\vec x_1, \ldots, \vec x_k)$ with sorts enumerating $\interpsort(\sort_{out})$ (so in
    particular $m$ corresponds to the length of $\interpsort(\sort_{out})$) and $\vec x_i$ of sorts $\interpsort(\sort_i)$.
\end{itemize}
subject to the following constraints

\begin{itemize}
\item $\phi_\equiv^\sort(\vec x, \vec y)$ should define a partial equivalence relation, i.e. be symmetric
and transitive,
\item $\phi_\domainof^\sort(\vec x)$ should be equivalent to $\phi_\equiv^\sort(\vec x, \vec x)$,
\item $\phi_R(\vec x_1, \ldots, \vec x_n)$ and
$\phi_\equiv^{\sort_i}(\vec x_i, \vec y_i)$ for $1 \le 1 \le n$
(where $\sort_i$ is the output sort associated with position $i$
of the relation $R$) should jointly imply
$\phi_R(\vec y_1, \ldots, \vec y_n)$.
\item for every output function symbol $f(x_1, \ldots, x_k)$ represented by terms $\vec{\overline{f}}(\vec{\vec{x}})$, we have
$$\forall \vec{\vec x} \; \vec{\vec y}~~ \left(\bigwedge_i \varphi^{\sort_i}_\equiv(\vec x_i, \vec y_i) \rightarrow
    \varphi^{\sort'}_\equiv(\vec{\overline{f}}(\vec{\vec x}), \vec{\overline{f}}(\vec{\vec y})) \right)$$
where $\sort'$ is the sort of the output of $f$ and the $\sort_i$ correspond to the arities.
\end{itemize}

In $\phi^{\sort}_\equiv$ and $\phi^{\sort}_\domainof$,  each $\vec x_1, \vec x_2$ is a tuple
containing
variables of sorts agreeing with the prescribed sequence of input sorts for $\sort'$.
Given a structure $M$ for the input sorts and a sort $\sort$ we call
 a binding of these variables to input elements of the appropriate
input sorts an \emph{$M,\sort$ input match}.
If in output relation $R$
 position $i$ is of sort $\sort_i$, then  in $\phi_R(\vec t_1 , \ldots \vec t_n)$ we
require
$\vec t_i$ to be a tuple of variables
of sorts agreeing with the prescribed sequence of input sorts for $\sort_i$.
Each of the above formulas
is over the vocabulary of $\inschema$.

An interpretation
$\interp$ defines a function from structures over vocabulary $\inschema$ to structures
over vocabulary $\outschema$ as follows:

\begin{itemize}
\item The domain of sort $\sort'$ is the set of equivalence classes of the partial equivalence
relation defined by $\phi_\equiv^{\sort'}$ over the $M,\sort'$ input matches.
\item A relation $R$ in the output schema is interpreted by the set of those tuples $\vec a$ such
that
$\phi_R(\vec t_1 , \ldots \vec t_n)$ holds for some $\vec t_1 \dots \vec t_n$ with each $\vec t_i$ a 
representative of $a_i$.
\end{itemize}

An interpretation $\interp$ also defines a map $\phi \mapsto \phi^*$ from formulas over $\outschema$ to
formulas over $\inschema$ in the obvious way. This map commutes with all logical connectives and thus
preserves logical consequence.

In the sequel, we are concerned with interpretations preserving certain theories
consisting of sentences in first-order logic. Recall that a \emph{theory}
in first-order logic is a deductively closed set of sentences.
A sentence belonging to a given theory is called one of its \emph{theorems}.
Given a theory $\Sigma$ over $\inschema$ and a theory $\Sigma'$ over $\outschema$, we say that
$\interp$ is an interpretation of $\Sigma'$ within $\Sigma$ if $\interp$ is an interpretation
such that for every theorem $\phi$ of $\Sigma'$, $\phi^*$ is a theorem of $\Sigma$.
Since $\phi \mapsto \phi^*$ preserves logical consequence, if $\Sigma'$ is generated by a set
of axioms $A$, it suffices to check that $\Sigma'$ proves $\phi$ for $\phi \in A$.

Finally, we are also interested in interpretations restricting to the identity on part of the input.
Suppose that $\outschema$ and $\inschema$ share a sort $\sort$. An interpretation $\interp$ of $\outschema$
within $\inschema$ is said to preserve $\sort$ if the output sort associated to $\sort$ is $\sort$ itself
and the induced map of structures is the identity over $\sort$. Up to equivalence, that
means we  fix $\phi_\domainof^T(x)$ to be, up to equivalence, $\top$,
$\phi_\equiv^\sort(x,y)$ to be the equality $x = y$
and map constants of type $\sort$ to themselves.

\myparagraph{Interpretations defining nested relational transformations}
We now consider how to define  nested relational transformations via interpretations.
The main idea will be to restrict all the constituent formulas to be
$\deltazero$ and to relativize the notion of interpretation to a background theory
that corresponds to our sanity axioms about tupling and sets.

We define the notion of \emph{subtype} of a type $T$ inductively as follows.
\begin{itemize}
\item $T$ is a subtype of $\sett(T')$ if $T=\sett(T')$ or if it is a subtype of $T'$.
\item $T$ is a subtype of $T_1 \times T_2$ if $T = T_1 \times T_2$ or if it is a subtype of either $T_1$ or $T_2$.
\item The only subtypes of $\ur$ and $\unit$ are themselves.
\end{itemize}

For every type $T$, we build a multi-sorted vocabulary $\aschema_T$ as follows.
\begin{itemize}
\item The sorts are all subtypes of $T$, $\unit$ and $\booltype = \sett(\unit)$.
\item The function symbols are the projections, tupling,
the unique element of type $\unit$, the constants $\boolff, \booltt$ of sort $\booltype$
representing $\emptyset, \{\unitobj\}$
and a special constant $\obj$ of sort $T$.
\item The relation symbols are the equalities at every sort and the membership
predicates $\in_T$.
\end{itemize}
Let $T_\oneobj$ be a type  which will represent  the type of
a complex object $\oneobj$.
We build a theory $\Sigma(T_\oneobj)$ on top of $\aschema_{T_{\oneobj}}$ from the following axioms:
\begin{itemize}
\item Equality should satisfy the congruence axioms for every formula $\phi$
$$\forall x y~~ (x = y ~\wedge~ \phi ~~\imp~~ \phi[y/x])$$
Note that it is sufficient to require this for atomic formulas to infer it for
all formulas.
\item We require that projection and tupling obey the usual laws for every type of $\aschema_{T_{\oneobj}}$.
$$
\forall x^{T_1} ~ y^{T_2}~ \pi_1(\< x, y \>) = x \qquad
\forall x^{T_1} ~ y^{T_2}~ \pi_2(\< x, y \>) = y \qquad
\forall x^{T_1 \times T_2} ~ \< \pi_1(x), \pi_2(x) \> = x$$
\item We require that $\unit$ be a singleton and every $\sett(T)$ in $\aschema_{T_\oneobj}$
$$\forall x^\unit ~ \unitobj = x$$
\item Lastly our theory imposes set extensionality
$$\forall x^{\sett(T)} ~ y^{\sett(T)} ~~ \left([\forall z^T~ (z \in_T x \equi z \in_T y)] \imp x =_T y\right)$$
\end{itemize}

Note that in interpretations we associate the input to
a structure that includes a distinguished constant.
For example,
an input of type $\sett(\ur)$ will be coded by a structure
with an element relation, an Ur-element  sort, and a constant whose  sort is the
type  $\sett(\ur)$.
In other contexts, like $\nrc$ expressions and implicit definitions
of transformations, we considered inputs to be \emph{free variables}.
This is only a change in terminology, but it reflects the fact that in evaluating the
interpretation on any input $i_0$ we will keep the interpretation of the associated
constant fixed, while we need to look at multiple bindings of the variables in each formula in order to form
the output structure.

We will show that $\nrcwget$ expressions defining
transformations from a nested relation of type $T_1$ to 
a nested relation of type $T_2$  correspond to a subset of interpretations
of $\Sigma(T_2)$ within $\Sigma(T_1)$ that preserve $\ursort$.
The only additional restriction we impose is that all
formulas $\phi_\domainof^T$ and $\phi_\equiv^T$ in the definition of such an interpretation
must be $\deltazero$. This forbids, for instance, universal quantification over the whole
set of Ur-elements. We thus call a first-order interpretation of
$\Sigma(T_2)$ within $\Sigma(T_1)$ consisting of $\deltazero$ formulas
a \emph{$\deltazero$ interpretation of $\Sigma(T_2)$ within $\Sigma(T_1)$}.

We now describe what it means for such an interpretation to define
a transformation from an instance of one nested relational schema
to another; that is, to map one  object to another.
We will denote the distinguished constant lying in the input sort by
$\inobj$ and the distinguished constant in the output sort by
$\outobj$.
Given any object $o$ of type $T$, define $M_o$, a structure satisfying $\Sigma(T)$,
as the least structure such that
\begin{itemize}
\item every subobject of $o$ is part of $M_o$
\item when $T_1 \times T_2$ is a subtype of $T$ and $a_1, a_2$ are objects of sort $T_1, T_2$ of $M_o$, then $\<a_1,a_2\>$ is an object of $M_o$
\item a copy of $\emptyset$ is part of $M_o$ for every sort $\sett(T)$ in $\aschema_T$
\item $\unitobj$ and $\{\unitobj\}$ are in $M_o$ at sorts $\unit$ and $\bool$.
\end{itemize}
The map $o \mapsto M_o$ shows how to translate an object to a logical
structure that is appropriate as the input of an interpretation.
Note that the further constraints ensure that every sort has at least one element
in $M_o$ and that there is one sort, $\ursort$, which contains at least two elements;
these technicality are important to ensure that interpretations are expressive enough.

We now discuss how the output of an interpretation is mapped back to an object.
The output of an interpretation is a multi-sorted structure
with a distinguished constant $\outobj$ encoding the output nested relational schema,
but it is not technically a nested relational instance
as required by our semantics for nested relational transformations.
We can convert the output to a semantically appropriate entity
via a modification of the well-known Mostowski collapse \cite{mostowski}.
We define $\collapse(e,M)$ on elements $e$ of the domain of a structure $M$ for the multi-sorted encoding of
a schema, by structural induction on the type of $e$:
\begin{itemize}
\item If $e$ has sort $T_1 \times T_2$ then we set
$\collapse(e,M)= \<\collapse(\pi_1(e), M), \collapse(\pi_2(e), M)\>$
\item If $e$ has sort $\sett(T)$, then we set
$\collapse(e,M)=\{\collapse(t,M) \mid t \in e\}$
\item Otherwise, if $e$ has sort $\unit$ or $\ursort$, we set $\collapse(e,M) = e$
\end{itemize}

We now formally describe how $\deltazero$ interpretations define functions between objects
in the nested relational data model.

\begin{defi}
We say that a nested relational transformation $\trans$ from  $T_1$ to $T_2$
is defined by a $\deltazero$ interpretation $\interp$ if, 
for every object $\inobj$ of type $T_1$,
the structure $M$ associated with $\inobj$ is mapped to $M'$ where $\trans(\inobj)$ is equal to $\collapse(\outobj,M')$.
\end{defi}
We will often identify a $\deltazero$ interpretation with the corresponding transformation, speaking
of its input and output as  a nested relation (rather than the corresponding structure).
For such an interpretation $\interp$ and an input object $\inobj$ we write
$\interp(\inobj)$ for the output of the transformation   defined by $\interp$ on $\inobj$.

\begin{exa} \label{ex:nrcinterp}
Consider an input schema consisting of a single binary relation $R: \sett( \ur \times \sett(\ur ))$,
so an input object is a set of pairs, with each pair consisting of an Ur-element and a set
of Ur-elements. The corresponding theory is $\Sigma(\sett(\ur \times \sett(\ur)))$,
which has sorts $T_{in} = \sett(\ur \times \sett(\ur))$, $\ur \times \sett(\ur)$ and $\ur$ and
relation symbols $\in_\ur$ and $\in_{\ur \times \sett(\ur)}$ on top of equalities.

If we consider the following instance of the nested relational schema
\[
R_0= \{ \<a, \{a,b\}\>, \<a, \{a, c\}\>, \<b, \{a, c\}\> \}
\]
Then the corresponding encoded structure $M$ consists of:
\begin{itemize}
\item $M^{T_{in}}$ containing only the constant $R_0$
\item $M^{\ur \times \sett(\ur)}$ consisting of the elements of $R_0$,
\item $M^\ur$ consisting of  $\{a, b, c\}$ 
\item  $M^{\sett(\ur)}$ consisting of the sets $\{a,b\}$, $\{a, c\}$,
\item $M^\unit = \{\unitobj\}$ and $M^\booltype = \{\emptyset, \{\unitobj\}\}$
\item the element relations interpreted in the natural way
\end{itemize}

Consider the transformation that groups on the first component, returning
an output object of type $O=\sett(\ur \times \sett(\sett(\ur)))$.
This is a variation of the grouping transformation from 
Example \ref{ex:nrc}.
On the example input $R_0$ the transformation would return
\[
O_0= \{ \< a, \{ \{a,b\}, \{a, c\} \}\> , \<b, \{\{a, c\}\}\> \}
\]

The output  would be represented by
a structure having  sorts $T_{out} = \sett(\ur \times \sett(\sett(\ur)))$, $\ur \times \sett(\sett(\ur)$, $\ur$, $\sett(\sett(\ur))$ and $\sett(\ur)$ in addition to $\unit$ and $\bool$.
It is easy to capture this transformation with
a  $\deltazero$ interpretation. For example,
the interpretation  could code
the output sort $\sett(\ur \times \sett(\sett(\ur))$ by  the sort $\sett(\ur \times
\sett(\ur))$, representing
each group by the corresponding Ur-element. 
\end{exa}

We will often make use of the following observation about interpretations:

\begin{prop}
\label{prop:fointerp-comp}
$\deltazero$ interpretations can be composed, and their composition correspond to
the underlying composition of transformations.
\end{prop}

The composition of nested relational interpretations amounts to the usual composition
of FO-interpretations (see e.g. \cite{xqueryinterp}) and an easy check that
the additional requirements we impose on nested relational interpretations are 
preserved.

We can now state the equivalence of $\nrc$ and interpretations formally:

\begin{thm} \label{thm:coddnrc}
Every transformation in $\nrcwget$ can be translated effectively to a
$\deltazero$ interpretation.
Conversely, for every $\deltazero$ interpretation, one can effectively form an equivalent
$\nrcwget$ expression. The translation from $\nrcwget$ to  interpretations
can be done in $\exptime$ while the converse translation can be performed in $\ptime$.
\end{thm}
This characterization holds when equivalence is
over finite nested relational inputs and also when arbitrary nested relations are allowed
as inputs to the transformations.

Note that very similar results occur in the literature, going back at least
to \cite{simulation}. Thus we defer
the proof to Appendix~\ref{app:fointerpproof}.
The direction from interpretations to  $\nrc$ will be the one that is directly relevant to us in the sequel, and its proof
is given by a  simple translation.

\subsection{Reduction to a characterization theorem in multi-sorted logic} \label{subsec:reduction}
The first step in the proof of Theorem \ref{thm:bethinterp} is to reduce to a more general
statement relating implicit definitions in multi-sorted logic to
 interpretations.

The first part of the reduction is to argue that we 
can suppress auxiliary parameters $\vec a$ in implicit definitions, proving the equivalence of the first two items in~\ref{cor:nrcbeth}:

\begin{lem}
\label{thm:dropprojective}
For any $\deltazero$ formula $\Sigma(\inobj, \outobj, \vec a)$  that 
implicitly defines $\outobj$ as a function of
$\inobj$, there is another $\deltazero$ formula $\Sigma'(\inobj, \outobj)$ which 
implicitly  defines $\outobj$ as a function of $\inobj$,  such that
$\Sigma(\inobj, \outobj, \vec a) \imp \Sigma'(\inobj, \outobj)$.
\end{lem}

Note that from this lemma  we get the equivalence of 1 and 2 in Corollary \ref{cor:nrcbeth}.

\begin{proof}
The assumption that $\Sigma$ implicitly defines $\outobj$ as a function
of $\inobj$  means that we have an entailment
$$\Sigma(\inobj, \outobj, \vec a) \models \Sigma(\inobj, \outobj', \vec{a'}) \imp \outobj = \outobj'$$
Applying $\deltazero$ interpolation, Proposition \ref{prop:interpolationmodeltheoretic},
we may obtain a formula $\theta(\inobj,\outobj)$ such that
$$\Sigma(\inobj, \outobj, \vec a) \models \theta(\inobj, \outobj) \qquad \text{and} \qquad \theta(\inobj, \outobj) \wedge \Sigma(\inobj, \outobj', \vec {a'}) \models \outobj = \outobj'$$
Now we can derive the following entailment
$$\Sigma(\inobj, \outobj, \vec a) \models [\theta(\inobj,\outobj') \wedge \theta(\inobj,\outobj'')] \imp \outobj' = \outobj''$$
This entailment is obtained from the second property of $\theta$, since
we can infer that $\outobj'=\outobj$ and $\outobj''=\outobj$.

Now we can  apply interpolation again to obtain a formula $D(\inobj)$ such that
$$\Sigma(\inobj, \outobj, \vec a) \models D(\inobj) \qquad \text{and} \qquad D(\inobj) \wedge \theta(\inobj, \outobj') \wedge \theta(\inobj, \outobj'') \models \outobj' = \outobj''$$
We now claim that $\Sigma'(\inobj, \outobj) \eqdef D(\inobj) \wedge \theta(\inobj, \outobj)$ is an implicit definition
extending $\Sigma$.
Functionality of $\Sigma'$ is a consequence of the second entailment witnessing that $D$ is an interpolant.
Finally, the implication $\exists \vec a~~ \Sigma(\inobj,\outobj, \vec a) \models \Sigma'(\inobj, \outobj)$ is given by the
combination of the first entailments witnessing that $\theta$ and $D$ are interpolants.
\end{proof}

With the above result in hand, from this point on we assume
that we do not have auxiliary parameters $\vec a$ in our implicit definitions.

\myparagraph{Reduction to Monadic schemas}
A \emph{monadic type} is a type built only using
the atomic type $\ursort$ and the type constructor
$\sett$. To simplify notation we define
$\ursort_0 \eqdef \ursort, \; \ursort_1 \eqdef \sett(\ursort_0), \ldots, \ursort_{n+1} \eqdef \sett(\ursort_n)$.
A monadic type is thus a $\ursort_n$ for some $n \in \bbN$.
A nested relational schema is monadic if it contains only monadic types,
and a $\deltazero$ formula is said to be monadic if all of its variables have monadic types.

Restricting to monadic formulas simplifies our arguments considerably.
It turns out that by the usual ``Kuratowski encoding'' of pairs by
sets, we can reduce all of our questions about implicit versus explicit
definability to the case of monadic schemas. The following proposition implies
that we can derive all of our main results for arbitrary schemas from their
restriction to monadic formulas. The proof is routine but tedious, so we defer
it to Appendix~\ref{app:reducemonadicnrc}--\ref{app:reducemonadicinterp}.

\begin{prop} \label{prop:reducemonadic}
For any nested relational schema $\aschema$, there is
a monadic nested relational schema $\aschema'$,
 an injection $\convert$ from instances of $\aschema$ to instances
of $\aschema'$ that is definable in $\nrc$, and an $\nrc[\nrcget]$ expression
$\convert^{-1}$
such that $\convert^{-1} \circ \convert$ is the identity transformation from
$\aschema \to \aschema$.

Furthermore, there is a $\deltazero$ formula $\image_\convert$ from $\aschema'$ to $\bool$
such that $\image_\convert(i')$ holds if and only if $i' = \convert(i)$
for some instance $i$ of $\aschema$.

These translations can also be given in terms of $\deltazero$ interpretations
rather than $\nrc$ expressions.
\end{prop}

As we now explain, Proposition \ref{prop:reducemonadic}  allows
us to reduce Theorem   \ref{thm:bethinterp} 
to the special case where we have only monadic nested relational schemas.
Given a $\deltazero$  implicit definition
$\Sigma(\inobj, \outobj)$ we can form a new
definition that computes the composition of the following
transformations:  $\convert^{-1}_{\inschema}$, a
projection onto the first component, the transformation defined by $\Sigma$,
and $\convert_{\outschema}$. Our new definition  captures this composition by
 a formula $\Sigma'(\inobj', \outobj')$ that
defines $\outobj'$ as a function of $\inobj$, where the formula
is  over a monadic schema.
Assuming that we have proven the theorem in the monadic case,
we would get an $\nrc$ expression $E'$ from $\inschema'$ to $\outschema'$ agreeing with  this formula
on its domain.
Now we can compose $\convert_{\inschema}$, $E'$, $\convert^{-1}_{\outschema}$,  and the
projection to
get an $\nrc$ expression agreeing with the partial function
defined by $\Sigma(\inobj, \outobj)$ on its domain, as required.

\myparagraph{Reduction to a result in multi-sorted logic}
Now we are ready to give our last reduction,
relating  Theorem \ref{thm:bethinterp} to a  general result
concerning multi-sorted logic.

Let $\sig$ be any multi-sorted signature, $\bigsorts$ be its
sorts and $\smallsorts$ be a subset of $\bigsorts$.
We say that a relation
$R$ is \emph{over $\smallsorts$} if all of its arguments are in
$\smallsorts$. Let $\Sigma$
be a set of sentences in $\sig$.
Given a model
$M$ for $\sig$, let $\smallsorts(M)$ be the union
of the domains of relations over $\smallsorts$,
and let $\bigsorts(M)$ be defined similarly.

We say that $\bigsorts$ is \emph{implicitly interpretable}
over $\smallsorts$ relative to $\Sigma$ if:
\begin{highlight}
Fix any models $M_1$  and $M_2$ of $\Sigma$. Suppose
$m$  is an isomorphism
from $\smallsorts(M_1)$ to $\smallsorts(M_2)$: that is a bijection from the domain of each sort,
that  preserves all relations over $\smallsorts$ in both directions.
Then $m$ extends to a unique mapping
from $\bigsorts(M_1)$ to $\bigsorts(M_2)$ which
preserves all relations over $\bigsorts$.
\end{highlight}

Informally, implicit interpretability states  that
the sorts in $\bigsorts$ are semantically determined by
the sorts in $\smallsorts$. The property implies in particular  that
if $M_1$ and $M_2$ agree on the interpretation of sorts
in $\smallsorts$, then the identity mapping on sorts in $\smallsorts$
extends to a mapping that preserves sorts in $\bigsorts$.

We relate this semantic property to a syntactic one.
We say that $\bigsorts$ is \emph{explicitly interpretable}
over $\smallsorts$ relative to $\Sigma$ if
for all $\sort$ in $\bigsorts$ there is a formula $\psi_\sort(\vec x, y)$
where $\vec x$ are variables with sorts in $\smallsorts$, $y$ a variable of
sort $\sort$, such that:

\begin{itemize}
\item In any model $M$ of $\Sigma$, $\psi_\sort$ defines a partial function 
$F_\sort$
mapping  $\smallsorts$ tuples surjectively on to $\sort$.

\item For every relation $R$ of arity $n$ over $\bigsorts$, there is
a formula $\psi_R(\vec x_1, \ldots \vec x_n)$ using
only relations of $\smallsorts$ and
only quantification over $\smallsorts$ such that
in any model $M$ of $\Sigma$, the  pre-image of $R$
under the mappings $F_\sort$ for the different arguments
of $R$ is  defined by $\psi_R(\vec x_1, \ldots \vec x_n)$.
\end{itemize}

Explicit interpretability states that there is an  interpretation in
the sense of the previous section that produces the structure in $\bigsorts$
from the structure in $\smallsorts$, and in addition  there is a definable
 relationship between
an element $e$ of a sort in $\bigsorts$ and the tuple that codes $e$ in the interpretation.
Note that $\psi_\sort$, the mapping between the elements $y$ in $\sort$  and
the tuples in $\smallsorts$ that interpret them, can use arbitrary relations.
The key property is that when we pull a relation $R$ over $\bigsorts$
back using the mappings $\psi_\sort$, then we obtain something definable
using $\smallsorts$.

With these definitions in hand,
we are  ready to state a result in  multi-sorted logic which
allows us to generate interpretations from classical
proofs of functionality:

\begin{thm} \label{thm:general}
For any $\Sigma, \smallsorts, \bigsorts$ such that $\Sigma$ entails that a sort
of $\smallsorts$ has at least two elements, $\bigsorts$ is explicitly
interpretable over $\smallsorts$ if and only if it is implicitly interpretable
over $\smallsorts$.
\end{thm}

This can be thought of as an analog of Beth's theorem \cite{beth,craig57beth} for multi-sorted logic.
The proof is given in the next subsection.
For now we explain how it implies Theorem \ref{thm:bethinterp}.
In this explanation we assume a monadic schema for both input and output. Thus every element $e$
in an instance  has sort $\ursort_n$ for some $n \in \bbN$.

Consider a $\deltazero$ formula $\Sigma(\inobj, \outobj)$ over a monadic schema
that implicitly defines
$\outobj$ as a function of  $\inobj$.
$\Sigma$ can be considered as a multi-sorted first-order formula with  sorts for
every subtype occurrence of the input as well as distinct sorts for every subtype 
occurrence of the output other than
$\ursort$.
Because we are dealing with  monadic input and output schema, every sort other than
$\ursort$ will be of the form $\sett(T)$, and these sorts
have only the element relations $\in_T$ connecting them.  We refer
to these as \emph{input sorts} and \emph{output sorts}. We modify $\Sigma$ by asserting
that all elements of  the input sorts lie underneath $\inobj$, and
all elements of the output sorts lie underneath $\outobj$. 
Since $\Sigma$
was $\deltazero$, this does not change the semantics.
We also conjoin to $\Sigma$ the sanity axioms for the schema, including
 the extensionality axiom at the sorts corresponding
to each object type. Let $\Sigma^*$ be the resulting
formula. In this transformation, as was the case with interpretations,
we  change our perspective on inputs and outputs, considering
them as constants rather than as free variables. We do this only 
to match our result in multi-sorted logic, which  deals with
a set of \emph{sentences} in multi-sorted first-order logic, rather than formulas with
free variables.

Given models $M$ and $M'$ of $\Sigma^*$, we define relations $\equiv_i$ connecting
elements of $M$ of depth $i$ with elements of $M'$ of depth $i$. For $i=0$,
$\equiv_i$ is the identity: that is, it connects elements of $\ursort$ if and only if they
are identical. For $i=j+1$, $\equiv_i(x,x')$ holds exactly when for every $y \in x$
there is $y' \in x'$ such that $y \equiv_j y'$, and vice versa.

The fact that $\Sigma$ implicitly defines $\outobj$ as a  function of $\inobj$ tells us that:
\begin{highlight}
Suppose $M \models \Sigma^*$, $M' \models \Sigma^*$ and $M$ and $M'$ are identical on the input sorts.
Then the mapping taking a $y \in M$ of depth $i$ to a $y' \in M'$ such that $y' \equiv_i y$ is
 an isomorphism of the output sorts that is the identity on $\ursort$. 
Further, any isomorphism  of $\bigsorts(M)$ on to $\bigsorts(M')$ that is the identity 
on $\ursort$ must be equal to $M$: one can show this by induction on the depth $i$ using
the fact that $\Sigma^*$ includes the extensionality axiom.
\end{highlight}

From this, we see that the output sorts are implicitly interpretable over the input sorts
relative to $\Sigma^*$.
Using Theorem \ref{thm:general}, we conclude that the output sorts are explicitly interpretable
in the input sorts relative to $\Sigma^*$.
Applying the conclusion to the formula $x=x$, where $x$ is a
variable of a sort corresponding
to object type $T$ of the
output, we obtain  a first-order formula $\phi^T_{\domainof}(\vec x)$ over the input sorts.
Applying the conclusion to the formula $x=y$ for $x, y$ variables
corresponding to the object type $T$ we get a formula
$\phi_{\equiv_T}(\vec x, \vec x')$ over the input sorts. Finally applying
the conclusion
to the element relation $\epsilon_T$ at every level of the output,  we get a first-order
formula $\phi_{\epsilon_T}(\vec x, \vec x')$ over the input sorts. 
Because $\Sigma^*$ asserts that each element of the input sorts lies beneath a constant for $\inobj$, we can
convert all quantifiers to bind only beneath $\inobj$, giving us $\deltazero$ formulas.
It is easy to verify that these formulas give us the desired interpretation.
This completes the proof of Theorem \ref{thm:bethinterp}, assuming Theorem \ref{thm:general}.

\subsection{Proof of the multi-sorted logic result} \label{subsec:proofmultisorted}
In the previous subsection we reduced our goal result about generating interpretations
from proofs to  a result in multi-sorted first-order logic, Theorem \ref{thm:general}.
We will now present the proof of Theorem \ref{thm:general}.  
The direction from explicit interpretability to
implicit interpretability is straightforward, so we will be
interested only in the  direction from implicit to explicit.
Although the theorem appears to be new as stated, 
each of the components is a variant of arguments that already appear in the model theory literature.
The core of our presentation here is  a variation of the proof of Gaifman's
coordinatisation theorem
as presented in \cite{hodgesbook}.

We make use of  only quite basic results from model theory:
\begin{itemize}
\item the \emph{compactness theorem} for first-order logic, which states
that for any  theory $\Gamma$, if every finite subcollection of $\Gamma$
is satisfiable, then $\Gamma$ is satisfiable;
\item
the \emph{downward Löwenheim-Skolem theorem}, which
states that if $\Gamma$ is countable and has a model, then
it has a countable model;
\item the \emph{omitting types theorem} for first-order logic.
A first-order theory $\Sigma$ is
said to be \emph{complete} if for
every other first-order sentence $\phi$ in the vocabulary of $\Sigma$, either $\phi$ or $\neg \phi$
is entailed by $\Sigma$. 
Given a set of constants $B$, a \emph{type} over
$B$ is an infinite collection $\tau(\vec x)$ of formulas using
variables $\vec x$ and constants $B$.
A type is \emph{complete} with respect to a theory
$\Sigma$ if every first-order formula  with
variables in $\vec x$ and constants from $B$ is either entailed
or contradicted by $\tau(\vec x)$ and $\Sigma$.
 A type $\tau$ is said to be \emph{realized} in 
a model $M$ if there is a $\vec x_0$ in $M$ satisfying all formulas in
$\tau$. $\tau$ is \emph{non-principal} (with respect to
a first-order theory $\Sigma$)
 if
there is no formula $\gamma_0(\vec x)$ such that $\Sigma  \wedge \gamma_0(\vec x)$
entails all of $\tau(\vec x)$. 
The version of the omitting types theorem that we will
use states  that
\begin{highlight}
if we have a countable set $\Gamma$ of complete types that are all non-principal relative to a complete
theory $\Sigma$, there
is some model $M$ of $\Sigma$ in which  none of the types in $\Gamma$ are realized.
\end{highlight}

\end{itemize}
Each of these results follows from a standard
model construction technique \cite{hodgesbook}. 

We 
 can easily show that to prove the multi-sorted
result, it suffices to consider only those $\Sigma$ which are
complete theories.
\begin{prop} \label{prop:assumecomplete} Theorem \ref{thm:general}
follows from its restriction to $\Sigma$ a complete theory.
\end{prop}

Recall that we are proving the direction from implicit interpretability to explicit interpretability.
Our first step will be to show that 
each element in the output sort is definable from
the input sorts, if we allow ourselves to guess  some parameters. For example, consider
the grouping transformation mentioned 
in Example \ref{ex:nrc}. Each output is obtained from
grouping input relation $F$ over some Ur-element $a$. So each member of the output is definable
from the input constant $F$ and a ``guessed'' input element $a$.
We will show that this is true in general.

For the next results, we have a \emph{blanket assumption that our underlying language is countable}, which
will be necessary in some applications of the compactness theorem.

Given a model $M$ of $\Sigma$ and $\vec x_0 \in \bigsorts$ within $M$, the
\emph{type of $\vec x_0$ with parameters
from $\smallsorts$} is the set of all formulas satisfied by 
$\vec x_0$, using any sorts and relations but only constants from $\smallsorts$.

A type $p$ is  \emph{isolated over $\smallsorts$} if
there is a formula $\phi(\vec x, \vec a)$ with parameters $\vec a$
from $\smallsorts$ such that  $M \models \phi(\vec x, \vec a) \rightarrow \gamma(\vec x)$
for each $\gamma \in p$.
The following is a step towards showing that elements in the output are well-behaved:

\begin{lem} \label{lem:isolation}
Suppose $\bigsorts$ is implicitly interpretable over $\smallsort$ with respect to
$\Sigma$.
Then in any  countable model $M$ of $\Sigma$ the type of any $\vec b$ over $\bigsorts$ with
parameters from $\smallsorts$ is isolated over $\smallsorts$.
\end{lem}
\begin{proof}
Fix a counterexample $\vec b$, and let $\Gamma$ be the set of formulas in $\bigsorts$ with constants
from $\smallsorts$ satisfied by $\vec b$ in $M$.
We claim that there is a model $M'$ with $\smallsorts(M')$ identical to $\smallsorts(M)$
where there is no tuple  satisfying  $\Gamma$. This follows from the failure of
isolation and 
the omitting types theorem.

Now we have a contradiction of implicit interpretability, since the identity
mapping on $\smallsorts$ can not extend to an isomorphism of relations over
$\bigsorts$ from $M$ to $M'$.
\end{proof}

The next step is to argue that every element of $\bigsorts$ is definable by a formula using
parameters from
$\smallsorts$. 

\begin{lem} \label{lem:nonuniformdefinability}
Assume implicit interpretability of $\bigsorts$ over $\smallsorts$ relative to $\Sigma$.
In any model $M$ of $\Sigma$, for  every element $e$ of a sort $\bigsort$ in $\bigsorts$,
there is a first-order formula $\psi_e(\vec y, x)$  with variables $\vec y$
having sort in $\smallsorts$ and $x$ a variable of sort $\bigsort$,
along with a tuple $\vec a$ in $\smallsorts(M)$
such that $\psi_e(\vec a, x)$ is satisfied only by $e$ in $M$.
\end{lem}

In a single-sorted setting, this can be found in \cite{hodgesbook} Theorem 12.5.8 where
it is referred to as ``Gaifman's coordinatisation theorem'', credited independently to unpublished work
of Haim Gaifman and Dale Myers.
The multi-sorted version is also a variant of Remark 1.2, part 4 in \cite{groupoids}, which points to a proof
in the appendix of \cite{modeltheorydiff};
the remark assumes that $\bigsorts$ is the set of all sorts. Another variation is Theorem 3.3.4 of
\cite{madarasz}.

\begin{proof}
Since a counterexample involves only formulas in a countable
language, by the Löwenheim-Skolem theorem mentioned
above, it is enough to consider the case where $M$ is countable.
By Lemma \ref{lem:isolation}, the type of every $e$ is isolated by
a formula $\phi(\vec x, \vec a)$ with parameters from $\smallsorts$
and relations from $\bigsorts$. We claim that $\phi$  defines $e$:
that is, $e$ is the only satisfier. If not, then there is $e' \neq e$ that satisfies
$\phi$.   Consider the relation $\vec j \equiv \vec j'$ holding
if $\vec j$ and $\vec j'$ satisfy all the same formulas using
relations and variables from $\bigsorts$ and
parameters from $\smallsorts$. 
Isolation implies  that $e \equiv e'$. Further, isolation of types
shows that $\equiv$ has the ``back-and-forth property''
given $\vec  d \equiv \vec d'$, and $\vec c$ we can obtain $\vec c'$
with $\vec d \vec c \equiv \vec d' \vec c'$.
To see this,  fix $\vec d \equiv \vec d'$ and
consider $\vec c$. We have $\gamma(\vec x, \vec y, \vec a)$ isolating the type
of $\vec d,\vec c$, and further $\vec d$ satisfies $\exists \vec y ~ \gamma(\vec x, \vec y, \vec a)$
and thus so does $\vec d'$ with witness $\vec c'$. But then using $\vec d \equiv \vec d'$
again we see that $\vec d, \vec c \equiv \vec d', \vec c'$. Using countability of
$M$ and this property we can inductively create a mapping on $M$ fixing $\smallsorts$ pointwise,  preserving
all relations in $\bigsorts$, and taking $e$ to $e'$. But this contradicts
implicit interpretability.
\end{proof}

\begin{lem} \label{lem:uniformdefinability}
The formula in Lemma \ref{lem:nonuniformdefinability} can be taken to depend only
on the sort $\bigsort$.
\end{lem}

\begin{proof}
Consider the type over the single variable $x$ in $\bigsort$
consisting of the formulas $\neg \isnonuniformdefiner_\varphi(x)$ where $\isnonuniformdefiner_\varphi(x)$
is the following formula
$$\exists \vec b ~ [\varphi(\vec b, x) \wedge \forall x'~ (\varphi(\vec b,
  x') \imp x' = x)]$$
where the tuple $\vec b$ ranges over $\smallsorts$.
By Lemma~\ref{lem:nonuniformdefinability}, this type cannot be satisfied in a model of $\Sigma$.
Since it is unsatisfiable, by compactness, there are finitely many formulas $\varphi_1(\vec b,x) , \dots, \varphi_n(\vec b, x)$ such that $\forall x ~~ \bigvee_{i = 1}^n \isnonuniformdefiner_{\varphi_i}(x)$ is satisfied. Therefore, each $\varphi_i(\vec b, x)$ defines a partial function from tuples of $\smallsort$ to $\bigsort$
 and every element of $\bigsort$ is covered by one of the $\phi_i$.
Recall that we
 assumed  that $\Sigma$ enforces that $\smallsorts$ has a sort
with at least two elements.
Thus we can combine the $\phi_i(\vec b, x)$ into a single formula $\psi(\vec b, \vec c, x)$
defining a surjective partial function from $\smallsort$ to $\bigsort$
where $\vec c$ is an additional parameter in $\smallsorts$ selecting some $i \le n$.
\end{proof}

We now need to go from the ``sub-definability'' or ``element-wise definability''
result above to an interpretation.
Consider the formulas $\psi_\sort$ produced
by Lemma \ref{lem:uniformdefinability}. For a relation $R$ of arity $n$ over
$\bigsorts$,  where the $i^{th}$ argument has sort $\sort_i$, consider the formula
\begin{align*}
\psi_R(\vec y_1 \ldots \vec y_n) = \exists x_1 \ldots x_n ~[R(x_1 \ldots x_n) \wedge 
\bigwedge_i \psi_{\sort_{i}}(\vec y_i, x_i)]
\end{align*}
where $\vec x_i$ is a tuple of variables of sorts in $\smallsorts$.

The formulas $\psi_\sort$ for each sort $\sort$ and the formulas $\psi_R$ for each
relation $R$ are as required by the definition of
explicitly interpretable, except that they may use quantified variables and relations
of $\bigsorts$, while we only
want to use variables and relations from $\smallsorts$. We take care of this in the following lemma,
which says that  formulas over $\bigsorts$
do not allow us to define any more subsets of $\smallsorts$ than we can
with formulas over $\smallsorts$.

\begin{lem} \label{lem:uniformreduction} 
Under the assumption of implicit interpretability, 
for  every
formula $\phi(\vec y)$ over $\bigsorts$ with $\vec y$ variables
whose sorts are in $\smallsorts$ there is a formula $\phi^\circ(\vec y)$ over
$\smallsorts$  -- that is, containing only variables, constants, and relations
from $\smallsorts$ -- such that for every model $M$ of $\Sigma$,
\[
M \models \forall \vec y ~ [\phi(\vec y) \leftrightarrow \phi^\circ(\vec y)]
\]
\end{lem}
\begin{proof}

We give an argument assuming the existence of a saturated model  for the theory: that is, a model $M$
in which for every set of formulas $\Gamma(\vec x)$, with parameters from the model,
 of cardinality strictly smaller than the model, if $\Gamma$  is finitely satisfiable in $M$ then it is realized in $M$.
Such models exist for any theory under the generalized continuum hypothesis GCH.
The additional set-theoretic hypothesis can be removed by using weaker notions of saturation -- the modification is a standard
one in model theory, see \cite{ck,hodgesbook}.

Assume not, with $\phi(\vec y)$ as a counterexample.
By completeness and our assumption, we know that there is a
saturated model $M$ of $\Sigma$ containing $\vec c$, $\vec c'$ that agree on
all formulas in $\smallsorts$ but that disagree on $\phi$. Call $R$ the
reduct of $M$ to $\smallsorts$. The partial map that sends $\vec c$ to
$\vec c'$ is a partial automorphism of $R$. Since $M$ is saturated, so is $R$,
so in particular $R$ is strongly homogeneous and the aforementioned partial
map can be extended to an automorphism of $R$ that sends $\vec c$ to $\vec c'$.
But then by weak implicit interpretability, this should extend to an automorphism
  of $M$ that sends $\vec c$ to $\vec c'$. This implies that $M$ satisfies
$\phi(\vec c) \equi \phi(\vec c')$, a contradiction.
\end{proof}

Above we obtained the formulas $\psi_R$ for each relation symbol $R$
needed for an explicit interpretation.
We can obtain formulas defining the necessary equivalence relations
$\psi_{\equiv}$ and $\psi_{\domainof}$ easily from these.

Putting Lemmas \ref{lem:nonuniformdefinability}, \ref{lem:uniformdefinability},
and \ref{lem:uniformreduction} together yields  a proof
of Theorem \ref{thm:general}.

\subsection{Putting it all together to  complete the proof of Theorem \ref{thm:bethinterp}}
We summarize our results on extracting $\nrcwget$  expressions  from classical proofs of
functionality.
We have shown in Subsection \ref{subsec:reduction} how to convert the problem
to one with no extra variables other than input and output
and with only monadic schemas --
and thus no use of products or tupling. We also showed how to convert the resulting
formula  into a theory in multi-sorted first-order logic. That is, we no longer
need to talk about $\deltazero$ formulas.  

In Subsection \ref{subsec:proofmultisorted} we showed that from a theory in multi-sorted
first-order logic we can obtain an interpretation. This first-order
interpretation in a multi-sorted logic can then be converted
back to a $\deltazero$ interpretation, since the background theory forces
each of the input sorts in the multi-sorted structure to
correspond to a level of nesting below one of the constants corresponding
to an input object.
Finally, the results of Subsection \ref{subsec:fointerp} allow us to convert this interpretation
to an $\nrcwget$ expression.
With the exception of the result in multi-sorted logic, all of the constructions are effective.
Further, these effective conversions are all in
polynomial time except for the transformation from an
interpretation to an $\nrcwget$ expression, which is exponential time in the worst case.
Outside of the multi-sorted result, which makes use of infinitary methods, the conversions
are each sound when equivalence over finite input structures is considered as well as the
default case when arbitrary inputs are considered. As explained in Subsection \ref{subsec:reduction},
when equivalence over finite inputs is considered, we cannot hope to get a synthesis result
of this kind.

\section{The effective result: efficiently generating $\nrc$ expressions from proofs} \label{sec:effective}

We now turn to the question of \emph{effective} conversion from implicit definitions to explicit $\nrc$ transformations, leading up to our second main contribution.

\subsection{Moving effectively from implicit to explicit: statement of the main result} \label{subsec:statement}
Recall that previously we have phased implicit definability as an entailment in the presence of extensionality axioms.
We also recall that we can rephrase it without explicitly referring to extensionality.

A $\deltazero$ formula $\varphi(\vvi, \vva, o)$,  implicitly defines
variable $o$ in terms of variables $\vvi$  if we have
\begin{equation}
\varphi(\vvi, \vva, o) \; \wedge \; \varphi(\vvi, \vva', o')
~\modelssem~ o \equiv_T o'
\end{equation}

Here $\equiv_T$ is equivalence-modulo-extensionality, as defined in Section \ref{sec:prelims}.
 It is the same 
as equality if we add  extensionality axioms on the left of the entailment symbol. 
Thus \emph{this entailment is  the same as entailment with the sanity axioms from the previous section, using
native equality rather than $\equiv_T$}. And since these are $\deltazero$ formulas, \emph{it is also  the same
as entailment over ``general models'', where we assume only correct typing and projection commuting with tupling}.

\myparagraph{Proof systems for $\deltazero$ formulas}
Recall that our goal  is an effective version of Corollary \ref{cor:nrcbeth}.
For this we need to formalize our proof system for $\deltazero$ formulas, which
will allow us to talk about proof witnesses for implicit definability.

A finite \mset of primitive membership expressions $t \incontext u$  (i.e.
extended $\deltazero$ formulas)
will be called an $\in$\emph{-context}. These expressions arise naturally when breaking
down bounded quantifiers during a proof.

We introduce notation for instantiating a block of bounded quantifiers at a time in a
$\deltazero$ formula.

A term is a \emph{tuple-term} if it is built up from variables using pairing.

If we want to talk only about  \emph{effective generation} of $\nrc$ witnesses from
proofs, we can use a basic proof system for $\deltazero$ formulas, whose
inference rules
are shown in Figure
\ref{fig:dzcalc-2sided}.

\begin{figure} 
\begin{mathpar}
\inferrule*
[left={Ax}%
]
{ }{\Theta; \; \Gamma, \phi \seq \phi, \Delta}%
\and
\inferrule*
  [left={$\bot$\textsc{-L}}]
  {  }{\Theta; \; \Gamma, \bot \seq \Delta}
\\

\inferrule*
  [left={$\neg$}\textsc{-L}]
{
\Theta; \; \Gamma \seq \neg \phi, \Delta
}{\Theta; \; \Gamma, \phi \seq \Delta}
\and
  \inferrule*
  [left={$\neg$}\textsc{-R}]
{
\Theta; \; \Gamma, \phi \seq \Delta
}{\Theta; \; \Gamma \seq \neg \phi, \Delta}
\\

\inferrule*
  [left={$\wedge$\textsc{-R}}]
{
\Theta; \; \Gamma \seq \phi_1, \Delta
\and
\Theta; \; \Gamma \seq \phi_2, \Delta
}{\Theta; \; \Gamma \seq \Delta, \phi_1 \wedge \phi_2}
\and
\inferrule*
  [left={$\vee$}\textsc{-R}]
{
\Theta; \; \Gamma \seq \phi_1, \phi_2,\Delta
}{\Theta; \; \Gamma \seq \phi_1 \vee \phi_2, \Delta}
  \\

\inferrule*
  [left={$\forall$\textsc{-R}}, right={$y$ fresh}]
{\Theta, y \in b; \; \Gamma \seq \phi[y/x], \Delta}{\Theta; \; \Gamma \seq \forall x \in b \;\phi, \Delta}
\and
\inferrule*
  [left={$\exists$\textsc{-R}}]
{\Theta, t \in b; \; \Gamma \seq \phi[t/x], \exists x \in b\;\phi,\Delta}{\Theta, t \in b; \; \Gamma \seq \exists x \in b\;\phi, \Delta}\\
\inferrule*
[left={Refl}]
{\Theta; \; \Gamma, t =_\ur t \seq \Delta}
{\Theta; \; \Gamma \seq \Delta}
\and
\inferrule*
[left={Repl}%
]
{\Theta; \; \Gamma, t=_\ur u, \phi[u/x], \phi[t/x] \seq \Delta}
{\Theta; \; \Gamma, t =_\ur u, \phi[t/x] \seq \Delta}\\
\inferrule*
[left={$\times_\eta$}, right={$x_1,x_2$ fresh}]
{\Theta[\tuple{x_1,x_2}/x]; \; \Gamma[\tuple{x_1,x_2}/x] \seq \Delta[\tuple{x_1,x_2}/x]}
{\Theta; \; \Gamma \seq \Delta}
\and
\inferrule*
[left={$\times_\beta$}, right={$i {\in} \{1,2\}$}]
{\Theta[t_i/x]; \; \Gamma[t_i/x] \seq \Delta[t_i/x]}
{\Theta[\pi_i(\tuple{t_1,t_2})/x]; \; \Gamma[\pi_i(\tuple{t_1,t_2})/x] \seq \Delta[\pi_i(\tuple{t_1,t_2})/x]}
\end{mathpar}
\caption{Proof rules for a $\deltazero$ calculus, without restrictions for
  efficient generation of witnesses. The left side of $;$ specifies the
  $\in$-context. The negation rules $\neg$\textsc{-L} and $\neg$\textsc{-R}
  permit to exchange formulas between both sides of $\seq$ such that it
  suffices to have the rules for the connectives only for one side, where we
  choose the right side.}

\label{fig:dzcalc-2sided}
\end{figure}

The node labels are a variation of  the traditional rules for first-order logic, with a couple
of quirks related to  the specifics of $\deltazero$  formulas.
Each node label has shape
$\Theta; \Gamma \vdash \Delta$ where
\begin{itemize}
\item $\Theta$ is an 
$\in$-context.
Recall that these are \msets of membership atoms --- the only formulas in our proof
system that are extended $\deltazero$ but not $\deltazero$.
They will emerge during proofs involving $\deltazero$ formulas when
we start breaking down bounded-quantifier formulas.  
\item $\Gamma$ and $\Delta$ are finite \msets of $\deltazero$
  formulas.\footnote{Much of our machinery also works if sequents are built
  from finite sets instead of finite multisets. However, the specification of
  certain proof transformations (e.g. Appendix~\ref{app:focused}) is much simpler with
  multisets.}
\end{itemize}

\noindent
For example, \textsc{Repl} in the figure is a ``congruence rule'', capturing
that terms that are equal are interchangeable. Informally, it says that to prove
conclusion $\Delta$ from a hypothesis that includes a formula $\phi$ including
variable $t$ and an equality $t =_{\ur} u$, it suffices to 
add to the hypotheses
a copy of $\phi$ with $u$ replacing some occurrences of $t$.

A proof tree whose root is labelled by $\Theta; \; \Gamma \vdash \Delta$ witnesses that,
for any given meaning of the free variables, if all the membership relations
in $\Theta$ and all formulas in $\Gamma$ are satisfied, then there is a formula in $\Delta$
which is true.
We say that we have a proof of a single formula $\phi$ when we have a proof of $\emptyset ; \emptyset \vdash \phi$.

The proof system is easily seen to be sound:
if $\Theta; \; \Gamma \vdash \Delta$, then $\Theta; \; \Gamma \models \Delta$, where we remind the reader that
$\models$ considers all models, not just extensional ones.
It can be shown to be complete by a standard technique (a ``Henkin construction'', see Appendix~\ref{app:completeness}).

To generate $\nrc$ definitions \emph{efficiently} from proof witnesses will require a more
restrictive proof system, in which we enforce some ordering on how proof rules can be applied,
depending on the shape of the hypotheses.
We refer to proofs in this system as \emph{\focused proofs},
To this end we consider a \mset of formulas \emph{existential-leading}
($\pospolarity$) if it contains only atomic formulas (i.e. formulas
of the form $t =_\ursort t'$ or $t \neq_\ursort t'$),
formulas with existential quantification as top-level connective and the
truth-value constant $\bot$.

Our restricted proof system -- \focused proofs --  is shown in Figure \ref{fig:dzcalc-1sided}.
A superficial difference from Figure \ref{fig:dzcalc-2sided} is that
the restricted system is ``almost $1$-sided'': $\deltazero$ formulas only
occur on the right, with only $\in$-contexts on the left. In particular,
a top-level goal $\Theta;\Gamma \proves \Delta$ in the  higher-level system
would be expressed as $\Theta \proves \neg \Gamma, \Delta$ in this system.
We will often abuse notation by referring to {\focused} proofs
of a $2$-sided sequent $\Theta; \Gamma \proves \Delta$, considering them as
``macros'' for the corresponding $1$-sided sequent.
For example, the hypothesis of 
 the $\neq$ rule  
could be written in  $2$-sided notation as
$\Theta,  t =_\ur u \seq \alpha[u/x], \alpha[t/x], \Delta^{\pospolaritysuper}$
while the conclusion could be written as
$\Theta , t =_\ur u \seq  \alpha[t/x], \Delta^{\pospolaritysuper}$.  As with
\textsc{Repl} in the prior system, this rule is about duplicating a hypothesis with
some occurrences of  $t$ replaced by $u$.

A major aspect of the restriction is baked into the shape of the $\exists$ rule.
It enforces that  existentials are instantiated only in a 
context which is $\pospolarity$, and that the term being substituted is of a simple shape (tupling of variables).

\begin{figure} 
\begin{mathpar}
\inferrule*
[left={$=$}]
{  }
{\Theta \seq x =_\ur x,\Delta}%
\and
\inferrule*
[left={$\top$}]
{  }
{\Theta \seq \top, \Delta}\\

\inferrule*
[left={$\neq$},right={$\alpha$ atomic}]
{\Theta \seq t \neq_\ur u, \alpha[u/x], \alpha[t/x],
  \Delta^{\pospolaritysuper}}
{\Theta \seq t \neq_\ur u, \alpha[t/x], \Delta^{\pospolaritysuper}}
\\

\inferrule*
  [left={$\wedge$}]
{
\Theta \seq \phi_1, \Delta
\and
\Theta \seq \phi_2, \Delta
}{\Theta \seq \phi_1 \wedge \phi_2, \Delta}
\and
\inferrule*
  [left={$\vee$}]
{
\Theta \seq \phi_1, \phi_2,\Delta
}{\Theta \seq \phi_1 \vee \phi_2, \Delta}
\\

\inferrule*
  [left={$\forall$},right={$y$ fresh}]
{\Theta, y \in b \seq \phi[y/x], \Delta}
{\Theta  \seq \forall x \in b. \;\phi, \Delta}
\and
\inferrule*
  [left={$\exists$},right=\text{$t$ tuple-term}]
{\Theta, t \in b; \; \Gamma \seq \phi[t/x], \exists x \in b\;\phi,\Delta^{\pospolaritysuper}}
{\Theta, t \in b;  \Gamma \seq \exists x \in b\;\phi, \Delta^{\pospolaritysuper}} \\

\inferrule*
[left={$\times_\eta$},right={$x_1,x_2$ fresh}]
{\Theta[\tuple{x_1,x_2}/x] \seq \Delta^{\pospolaritysuper}[\tuple{x_1,x_2}/x]}
{\Theta  \seq \Delta^{\pospolaritysuper}}
\and
\inferrule*
[left={$\times_\beta$},right={$i \in \{1,2\}$}]
{\Theta[t_i/x] \seq \Delta^{\pospolaritysuper}[t_i/x]}
{\Theta[\pi_i(\tuple{t_1,t_2})/x] \seq \Delta^{\pospolaritysuper}[\pi_i(\tuple{t_1,t_2})/x]}

\end{mathpar}

\caption{Our \focused calculus for efficient generation of witnesses. The left
  side of $\seq$ specifies the $\in$-context. The right side a finite \mset of
  $\deltazero$ formulas. Recall that an atomic formula is an equality or  inequality
  for $\ursort$. Formula \msets $\Delta^\pospolaritysuper$ are existential-leading, that is, they contain only atomic formulas
  and formulas with existential quantification as top-level connective.}

\label{fig:dzcalc-1sided}
\end{figure}

Soundness is evident, since it is a special case of the proof system
above. Completeness is not as obvious, since we are restricting the proof rules.
But we can translate proofs in the more general system of Figure \ref{fig:dzcalc-2sided} into
a {\focused} proof, but with an exponential blow-up: see Appendix~\ref{app:focused} for details.

Furthermore, since for $\deltazero$ formulas
equivalence over all structures is the same as equivalence over nested relations, 
a $\deltazero$ formula $\phi$ is provable exactly when $\modelssem \phi$.

\begin{exa} \label{ex:simplenesting}
Let us look at  how to formalize a variation of 
Example~\ref{ex:views}.
  The specification $\Sigma(\basenest,\vflat)$ includes two conjuncts 
$C_1(\basenest,\vflat)$ and $C_2(\basenest, \vflat)$.
  $C_1(\basenest, \vflat)$ states that every pair $\tuple{k,e}$ of $\vflat$ corresponds
  to a $\tuple{k, S}$ in $\basenest$ with $e \in S$:
  \[
\forall v \in \vflat ~\exists b \in \basenest. ~
  \pi_1(v)=_\ur \pi_1(b) \wedge \pi_2(v) \memmac \pi_2(b)\]

$C_2(\basenest,\vflat)$ is:
\[
 \forall b \in \basenest ~ \forall e \in \pi_2(b)  ~
  \exists v \in \vflat. ~ \pi_1(v)=_\ur \pi_1(b) \wedge \pi_2(v) =_\ur e \\
\]

Let us assume a stronger constraint, $\Sigma_{\lossless}(\basenest)$,
saying that the first component is a key and second is non-empty:
  \[
    \begin{array}{l@{~}l} &
    \forall b \in \basenest  ~ \forall b' \in \basenest.  \pi_1(b) =_{\ur} \pi_1(b')
  \rightarrow b \equiv b'\\
    \wedge & \forall b \in \basenest ~ \exists e \in \pi_2(b). \top
  \end{array}
  \]

With $\Sigma_{\lossless}$ we can show something stronger 
than in Example \ref{ex:views}:
$\Sigma \wedge \Sigma_{\lossless}$ implicitly defines 
$\basenest$ in terms of $\vflat$. That is, the view determines the identity query,
which is witnessed by a proof of
\[
  \Sigma(\basenest,\vflat) \wedge \Sigma_{\lossless}(\basenest) 
\wedge \Sigma(\basenest',\vflat)  \wedge \Sigma_{\lossless}(\basenest')
  \imp \basenest \equiv \basenest'\]
Let's prove this informally. Assuming the premise, it is sufficient
to prove $\basenest \subseteq \basenest'$ by symmetry. So fix $\tuple{k, S} \in \basenest$.
  By the second conjunct of $\Sigma_{\lossless}(\basenest)$, we know there is $e \in S$.
  Thus by $C_2(\basenest, \vflat)$,  $\vflat$ contains the pair $\tuple{k, e}$.
  Then, by $C_1(\basenest',\vflat)$, there is a $S'$ such that $\tuple{k, S'} \in \basenest'$.
To conclude it suffices to show that $S \equiv S'$. There are two similar directions, let us
detail the inclusion $S \subseteq S'$; so fix $s \in S$.
      By $C_2(\basenest,\vflat)$, we have $\tuple{k,s} \in \vflat$.
      By $C_1(\basenest',\vflat)$ 
      there exists $S''$ such that $\tuple{k, S''} \in \basenest'$
      with $s \memmac S''$. But since we also have 
$\tuple{k, S'} \in \basenest'$, the constraint
      $\Sigma_{\lossless}(\basenest)$ implies that $S' \equiv S''$, so $s \in S'$ as desired.
\end{exa}

\myparagraph{Effective Beth result}
A derivation of implicit definability %
in our proof system
will be referred to as a \emph{witness} to the implicit definability
of $o$ in terms of $\vvi$ relative to $\varphi$.
Formally, this is a derivation witnessing the judgement:
\[
\varphi(\vvi, \vva, o) \; \wedge \; \varphi(\vvi, \vva', o') ~\proves o \equiv_T o'
\]

With these definitions, we now state formally our main result,
the  effective version of Corollary \ref{cor:nrcbeth}.

\begin{restatable}[Effective implicit to explicit for nested data]{thm}{thmmainset}
\label{thm:mainset}
Given a 
witness to the  implicit definition of $o$ in terms of
$\vvi$ relative to $\deltazero$
$\varphi(\vvi, \vva, o)$,
one can compute  $\nrc$ expression $E$ such that for any $\vvi$, $\vva$ and $o$,
if  $\varphi(\vec i, \vva , o)$ then $E(\vvi)=o$.
Furthermore, if the witness is {\focused}, this can be done in polynomial time.
\end{restatable}

\paragraph{Application to views and queries}
We now state the consequence for effective rewriting queries over views mentioned in
the introduction.
Consider a query given by
 $\nrc$ expression $E_Q$ over inputs $\vec B$ and
$\nrc$ expressions $E_{V_1} \ldots E_{V_n}$ over $\vec B$.
$E_Q$ is \emph{determined} by 
$E_{V_1} \ldots E_{V_n}$, if
every two nested relations (finite or infinite) interpreting $\vec B$
that agree on the output of each $E_{V_i}$
agree on the output of $E_Q$. 
An $\nrc$ rewriting of $E_Q$ in terms of $E_{V_1} \ldots E_{V_n}$ is
an expression $R(V_1 \ldots V_n)$ such that for any nested relation
$\vec B$, if we evaluate each $E_{V_i}$ on $\vec B$ to obtain $V_i$ and
evaluate $R$ on the resulting $V_1 \ldots V_n$, we obtain $Q(\vec B)$.

Given $E_Q$ and $E_{V_1} \ldots E_{V_n}$, let
$\Sigma_{\vec V, Q}(\vec V, \vec B, Q, \ldots)$ conjoin the input-output specifications, as defined in Section \ref{sec:prelims},
for  $E_{V_1} \ldots E_{V_n}$ and $E_Q$. This formula has variables
$\vec B, V_1 \ldots V_n, Q$ along with auxiliary variables for subqueries.
 A proof witnessing determinacy of $E_Q$  by
$E_{V_1} \ldots E_{V_n}$, is a proof that $\Sigma_{V,Q}$ implicitly
defines $Q$ in terms of $\vec V$.

\begin{cor} \label{cor:effdeterminacy} From a witness that a set of $\nrc$ views $\vec V$ determines
an $\nrc$ query $Q$, we can produce an $\nrc$  rewriting
of $Q$ in terms of $\vec V$. If the witness is {\focused}, this can be done
in $\ptime$.
\end{cor}
The notion of determinacy of a query over views
relative  to a $\deltazero$ theory (e.g. the key constraint in Example \ref{ex:views}) is a straightforward generalization of the definitions above, and
Corollary \ref{cor:effdeterminacy} extends to this setting.

In the case where we are dealing with flat relations, the effective version
is well-known: see Toman and Weddell's \cite{tomanweddell}, and the discussion
in \cite{franconisafe,interpbook}.

We emphasize that the result involves equivalence up to extensionality,
which underlines the distinction from the classical Beth theorem. If
we wrote  out implicit definability up to extensionality as an entailment involving
two copies of the signature, we would  run into problems in applying the
standard proof of Beth's theorem.

\subsection{Tools for the effective Beth theorem}
  \label{subsec:tools}

\subsubsection{Interpolation}
The first tool for our effective Beth theorem,  Theorem \ref{thm:mainset},
will
be an effective version of interpolation Proposition \ref{prop:interpolationmodeltheoretic}. 
Recall that interpolation results state
that if we have an entailment involving
two formulas, a ``left'' formula $\phi_L$ and a ``right'' formula $\phi_R$, we can get an ``explanation''
for the entailment that factors through an expression only involving
non-logical symbols (in our case, variables) that
are common to $\phi_L$ and $\phi_R$.

\begin{thm}
\label{thm:interpolationdeltafocus}
Let $\Theta$ be an $\in$-context and $\Gamma, \Delta$ finite \msets of $\deltazero$ formulas.
Then from any proof of $\Theta; \; \Gamma \vdash \Delta$,
we can compute in linear time an extended $\deltazero$ formula $\theta$ with $\FV(\theta) \subseteq \FV(\Theta, \Gamma) \cap \FV(\Delta)$
such that $\Theta; \; \Gamma \vdash \theta$ and $\emptyset; \; \theta \vdash \Delta$.
\end{thm}

The $\theta$ produced by the theorem is a \emph{Craig interpolant}.
Craig's interpolation theorem \cite{craig57interp}
states that when $\Gamma \proves \Delta$ with $\Gamma, \Delta$  in first-order logic, such
 a $\theta$ exists in
first-order logic.
Our variant  states one can find $\theta$ in $\deltazero$ efficiently from
a proof of the entailment in either of our $\deltazero$ proof systems.
We have stated the result for the $2$-sided system. It holds also for the %
{\focused} system, where the partition of the formulas into left and right of the proof
symbol is arbitrary.
The argument
is induction on proof length, roughly following prior
interpolation algorithms \cite{smullyan}.
See Appendix~\ref{app:deltazerointerpolation}.

\label{page-plain-deltazero-due-to-equality}
We compare with the model-theoretic statement
Proposition~\ref{prop:interpolationmodeltheoretic}. There,  the interpolant is
$\deltazero$, while here in the effective version it is \textit{extended}
$\deltazero$. This is due to the linear time requirement, which leads to the
involvement of equalities $=_T$ in the interpolation algorithm. The
construction would also work in plain $\deltazero$ if $\equiv_T$ is used
instead. However, $\equiv_T$ is a shorthand for a formula whose size is
exponential in the term depth of the type $T$.

\subsubsection{Some admissible rules}

As we mentioned earlier, our {\focused} proof system is extremely low-level, and
so it is convenient to have higher-level proof rules as macros. We formalize
this below.

\begin{defi}
  \label{def:polyadm}
A rule with premise $\Theta' \vdash \Delta'$ and conclusion $\Theta \vdash \Delta$
  \[ \dfrac{\Theta' \vdash \Delta'}{\Theta \vdash \Delta} \]
is \emph{(polytime) admissible} in a given calculus if a proof of
the conclusion $\Theta \vdash \Delta$ in that calculus can be computed from a
proof of the premise $\Theta' \vdash \Delta'$ (in polynomial time).
\end{defi}

Up to rewriting the sequent to be one-sided, all the rules in
Figure~\ref{fig:dzcalc-2sided} are polytime admissible in the {\focused} calculus.
Our main theorem will rely on the polytime admissibility within
the {\focused} calculus of
additional rules that involve chains of existential quantifiers.
To state them, we need
to introduce a generalization of bounded quantification:
``quantifying over subobjects of a variable''. 
\begin{definition} \label{def:subtypeocc}
For every type $T$, define a subset of the non-empty words over the three-letter 
alphabet $\{1,2,\membersub \}$ of \emph{subtype occurrences of $T$} inductively as follows:
\begin{itemize}
\item If $p$ is a subtype occurrence of $T$ or the empty word  then the concatenation
$\membersub,  p$ is a subtype occurrence of $\sett(T)$.
\item If $i \in \{1,2\}$ and $p$ is a subtype occurrence of $T_i$,
$i, p$ is a subtype occurrence of $T_1 \times T_2$.
\end{itemize}
Given subtype occurrence $p$ and  quantifier symbol $\mathbf{Q} \in \{\forall, \exists\}$,
define the notation $\mathbf{Q} \; x \in_p t. \phi$ by induction on $p$:
\begin{itemize}
\item $\mathbf{Q} \; x \in_\membersub t. \phi$ is $\mathbf{Q} \; x \in t$
\item $\mathbf{Q} \; x \in_{\membersub, p} t. \phi$ is $\mathbf{Q} \; y \in t.
\mathbf{Q} \; x \in_p y. \phi$ with $y$ a fresh variable
\item $\mathbf{Q} \; x \in_{i, p} t. \phi$ is $\mathbf{Q} \; x \in_p \pi_i(t). \phi$
  when $i \in \{1,2\}$.
\end{itemize}
\end{definition}

Now we are ready to state the results we need on admissibility,
referring in each case to the {\focused} calculus. Some
further rules, that are easily seen to be admissible,  are used in the  appendices.
All proofs are found in Appendix \ref{app:derivations}.
The first states that if we have proven that there exists a
 subobject of $o'$ equivalent to object $r$, then we can prove
that for each element $z$ of $r$ there is a corresponding equivalent
subobject $z'$ within $o'$. Furthermore, this can be done effectively, and the
output proof has at most the same size: that is, there is not even a polynomial blow-up involved.

\begin{restatable}{lem}{effnestedproofdown}
  \label{lem:effnestedproof-down}
Assume $p$ is a subtype occurrence for the type of the term $o'$. The following is polytime admissible
\[\dfrac{\Theta \vdash \Delta,\; \exists {r' \inq_p o'}. \; r \equiv_{\sett({T'})} r'}{
\Theta, z \in r \vdash \Delta, \; \exists {z' \inq_{\membersub, p} o'}. \; z \equiv_{{T'}} z'
  }\]
Furthermore, the size of the output proof is at most the size of the input proof.
\end{restatable}

The second result states that we can move between an equivalence of $r, r'$
and a universally-quantified  biconditional between memberships in $r$
and $r'$. Because we are dealing with $\deltazero$ formulas, the universal quantification
has to be bounded by some additional variable $a$.

\begin{restatable}{lem}{equivsettoequivalence}
  \label{lem:equivsettoequivalence}
  The following is polytime admissible (where $p$ is a subtype occurrence of the type of $o'$)
  \[ 
\dfrac{ \Theta \seq \Delta, \exists r' \inq_p o'. \; r \equiv_{\sett({T'})} r'}{ \Theta \seq \Delta, \exists r' \inq_p o'. \forall z \inq a. \; z \memmac
  r \leftrightarrow z \memmac r'}
  \]
\end{restatable}

\subsubsection{The $\nrc$ Parameter Collection Theorem}
Our last tool is a kind of interpolation result
connecting $\deltazero$ formulas and $\nrc$:

\begin{restatable}[$\nrc$ Parameter Collection]{thm}{mainflatcasebounded}
\label{thm:mainflatcasebounded}
Let $L$, $R$ be sets of variables with $C = L \cap R$ and
\begin{itemize}
\item $\phi_{L}$ and $\lambda(z)$ $\deltazero$ formulas over $L$
\item $\phi_{R}$ and $\rho(z,y)$ $\deltazero$ formulas over $R$
\item $r$ a variable of $R$ and $c$ a variable of $C$.
\end{itemize}

Suppose that we have an EL-normalized proof of
\begin{align*}
\phi_L \wedge \phi_R ~  \imp ~ \exists y \in_p r  ~ \forall z \in c  ~
(\lambda(z) \equi \rho(z, y))
\end{align*}

Then one may compute in polynomial time an $\nrc$ expression $E$ with
free variables in $C$
such that
\begin{align*}
\phi_L 
\wedge \phi_R \imp
  \{ z \in c \mid \lambda(z)\} \in E
\end{align*}
\end{restatable}

If $\lambda$ was a ``common formula'' ---
one using only variables in $C$
---  then the nested
relation
$ \{ z \in c \mid \lambda(z)\}$ would be definable over $C$ in $\nrc$ via $\deltazero$-comprehension.
Unfortunately $\lambda$ is a ``left formula'', possibly with variables outside of $C$.
Our hypothesis is that it is equivalent to a 
``parameterized right formula'': a formula with variables in $R$ and parameters that lie below them.
Intuitively,
this can happen only if $\lambda$ can be rewritten to a formula $\rho'(z, x)$
with variables of $C$ and a distinguished $c_0 \in C$ such that
\begin{align*}
  \phi_L \wedge \phi_R ~  \imp ~ \exists x \in_p  c_0 ~ \forall z \in c  ~
(\lambda(z) \equi \rho'(z, x))
\end{align*}

And if this is true, we can use an $\nrc$ expression
over $C$ to define a set that will contain
the correct ``parameter'' value $x$ defining $\lambda$.
From this we can define a set containing the nested
relation $ \{ z \in c \mid \lambda(z)\}$.
A formalization of this rough intuition --- ``when left formulas are equivalent to parameterized right formulas, they are equivalent to parameterized
common formulas'' --- is given in~ Section \ref{app:paramfo}, where a similar
statement that does not mention $\nrc$ is proven for first-order logic
(sadly it does not seem strong enough to derive Theorem \ref{thm:mainflatcasebounded}).

We now give the proof of the theorem.

To get the desired conclusion, we need to prove a more general statement by induction
over proof trees. Besides making the obvious generalization to handle
two \msets of
formulas instead of the particular formulas  $\phi_L$ and $\phi_R$, as well as some corresponding left and right $\in$-contexts,
that may appear during the proof, we need to additionally generate a new formula
$\theta$ that only uses common variables, which can replace $\phi_R$ in the conclusion.
This is captured in the following lemma:

\begin{restatable}{lem}{lemnrcparamcoll}
  \label{lem:nrcparamcoll}
Let $L$, $R$ be sets of variables with $C = L \cap R$ and
\begin{itemize}
\item $\Delta_{L}, \lambda(z)$ a \mset of $\deltazero$ formulas over $L$
\item $\Delta_{R}, \rho(z,y)$ a \mset of $\deltazero$ formulas over $R$
\item $\Theta_L$ (respectively $\Theta_R$) an $\in$-context over $L$
(respectively over $R$)
\item $r_1, \ldots, r_k$ variables of $R$ and $c$ a variable of $C$.
\item we write $\cG(r_i)$ for $\exists y \in_{p_i} r_i ~ \forall z \in c ~ (\lambda(z) \equi \rho(z,y))$
\end{itemize}

Suppose that we have a
\normalized
proof tree with conclusion
\begin{align*}
\Theta_L, \Theta_R \vdash \Delta_L, \Delta_R, \cG(r_1), \ldots, \cG(r_k)
\end{align*}

Then one may compute in polynomial time an $\nrc$ expression $E$
and an extended $\deltazero$ formula $\theta$ using only variables from $C$
such that
\begin{align*}
\Theta_L ~~\modelssem~~ \Delta_L, \theta \vee
  \{ z \in c \mid \lambda(z) \} \in E
\quad \text{and} \quad
\Theta_R \modelssem \Delta_R, \neg \theta
\end{align*}
\end{restatable}

The theorem follows easily from this lemma, so we focus on proving the lemma,
by induction over the size of the proof of 
$\Theta_L, \Theta_R \vdash \Delta_L, \Delta_R, \cG(r_1), \ldots, \cG(r_k)$,
making a case distinction according to which
rule is applied last.
The way $\theta$ will be built will, perhaps unsurprisingly, be very reminiscent
of the way interpolants are normally constructed in standard proof systems
\cite{fittingbook,smullyan}.

For readability, we adopt the
following conventions:
\begin{itemize}
  \item We write $\goalList$ for the multiset of formulas $\cG(r_1), \ldots, \cG(r_k)$
and $\Lambda$ for the expression $\{ z \in c \mid \lambda(z) \}$.
\item For formulas and $\nrc$ expressions obtained by applying the
induction hypothesis, we use the names $\theta\IH$ and $E\IH$
    (or $\theta_1\IH, \theta_2\IH$ and $E_1\IH, E_2\IH$ when the induction hypothesis
is applied several times). In each subcase, our goal will be to build suitable
$\theta$ and $E$.
\item We will color pairs of terms, formulas and \msets of  formulas
according to whether they are %
part of either $\colL{\Theta_L; \Delta_L}$
or $\colR{\Theta_R; \Delta_R}$ either at the start of the case analysis or when
we want to apply the induction hypothesis.
In particular, the last sequent of the proof
under consideration will be depicted as
\[ \colL{\Theta_L}, \colR{\Theta_R} \vdash \colL{\Delta_L}, \colR{\Delta_R}, \cG \]
\item Unless it is non-trivial, we leave checking that the free variables in our
proposed definition for $E$ and $\Theta$ are taken among variables of $C$ 
to the reader.
\end{itemize}

In two of the cases below, we will make use of some syntactic sugar on top of bounded quantification. 
We introduce
\[\existsvars{x_1 \ldots x_n} t \in b\,.\, \phi\; \text{ and }\;
\forall x_1 \ldots x_n| t \in b\,.\, \phi\] as notation. This will be an extended $\deltazero$ formula, intuitively quantifying
over variables $x_1,\ldots,x_n$ that occur in tuple-term $t$, which is bounded by
set $b$. Here $\phi$ is an extended $\deltazero$ formula. %
The variables other than $x_1,\ldots,x_n$ occurring
in $t$ remain free in the resulting formula.

\label{page:existsvars}

To make this precise, let $\varpathref(t,v,u)$ be the set of the terms
expressed with projection applied to~$u$ that refer to an occurrence of
variable $v$ in tuple-term $t$, ordered according to the occurrences when $t$ is printed.
Formally, $\varpathref(t,t,u) \eqdef \{u\}$; $\varpathref(t,v,u) \eqdef \{\}$
if $t$ is a variable other than $v$; and $\varpathref(\la t_1, t_2\ra,v,u)
\eqdef \varpathref(t_1,v,\pi_1(u)) \cup \varpathref(t_2,v,\pi_2(u))$. We write
$\varpathref_i(t,v,u)$ for the $i$-th member of $\varpathref(t,v,u)$. We can
then define $\existsvars{x_1 \ldots x_n} t \in b\,.\, \phi$ as
{\renewcommand{\arraycolsep}{1pt}
\[
\begin{array}{ll}
\exists y \in b\, .\, &
\phi[\varpathref_1(t,x_1,y)/x_1,\ldots,\varpathref_1(t,x_n,y)/x_n]\; \land\\
& \bigwedge_{z \text{ is a variable occurring in } t \text{ other than } x_1,\ldots,x_n}
  z = \varpathref_1(t,z,y)\; \land\\
  & \bigwedge_{u \text{ is a variable occurring in } t \text{ and }
    \varpathref_j(t,y,u) \text{ for } j > 1 \text{ in } \varpathref(t,y,u)}
    \varpathref_1(t,u,y) = \varpathref_j(t,u,y)
\end{array}\]}
$\forall x_1 \ldots x_n| t \in b\,.\, \phi$ can be defined analogously.

Here equality $=_T$ (written above without the type decoration) of
\textit{extended} $\deltazero$ formulas comes into play, in contrast to
$\equiv_T$, to meet the polynomial time requirements of
Theorems~\ref{thm:interpolationdeltafocus} and~\ref{thm:mainflatcasebounded}
as indicated on p.~\pageref{page-plain-deltazero-due-to-equality}.

As an example for the notation consider
\[\existsvars{x_1 x_2}\la \la x_1, x_2 \ra, z\ra \in b\,.\, x_1 \equiv x_2,\]
which stands for
\[\exists y \in b\, .\, \pi_1(\pi_1(y)) \equiv \pi_2(\pi_1(y)) \land z = \pi_2(y).\]

With these conventions in mind, let us proceed.

\begin{itemize}
  \item If the last rule applied is the $\top$ rule, in both cases we are
    going to take $E \eqdef \emptyset$, but pick $\theta$ to be $\bot$ or $\top$
    according to whether $\top$ occurs in $\colL{\Delta_L}$ or $\colR{\Delta_R}$;
    we leave checking the details to the reader.
  \item If the last rule applied is the $\wedge$ rule, we have two cases according
    to the position of the principal formula $\phi_1 \wedge \phi_2$. In both
    cases, $E$ will be obtained by unioning $\nrc$ expressions obtained from the
    induction hypothesis, and $\theta$ will be either a disjunction or a conjunction.
    \begin{itemize}
      \item If we have $\colL{\Delta_L} = \colL{\phi_1 \wedge \phi_2, \Delta_L'}$,
        so that the proof has shape
        \begin{mathpar}
          \inferrule*{
            \colL{\Theta_L}, \colR{\Theta_R} \vdash \colL{\phi_1, \Delta_L'}, \colR{\Delta_R}, \goalList
            \and
            \colL{\Theta_L}, \colR{\Theta_R} \vdash \colL{\phi_2, \Delta_L'}, \colR{\Delta_R}, \goalList
          }
          {
            \colL{\Theta_L}, \colR{\Theta_R} \vdash \colL{\phi_1 \wedge \phi_2, \Delta_L'}, \colR{\Delta_R}, \goalList
          }
        \end{mathpar}
        by the induction hypothesis, we have $\nrc$ expressions
        $E_1\IH$, $E_2\IH$ and formulas $\theta_1\IH, \theta_2\IH$ such that
        \[
          \begin{array}{c !\qquad!{\text{and}}!\qquad c}
            \colL{\Theta_L} \modelssem \colL{\phi_1, \Delta_L'},
            \theta_1\IH \vee \Lambda \in E_1\IH
            &
            \colR{\Theta_R} \modelssem \colR{\Delta_R}, \neg \theta_1\IH
\\
            \colL{\Theta_L} \modelssem \colL{\phi_2, \Delta_L'},
            \theta_2\IH \vee \Lambda \in E_1\IH
            &
            \colR{\Theta_R} \modelssem \colR{\Delta_R}, \neg \theta_2\IH
          \end{array}
        \]
        In that case, we take $E = E_1\IH \cup E_2\IH$ and $\theta \eqdef \theta_1\IH \vee \theta_2\IH$.
        Weakening the properties on the left column, we have
        \[
            \colL{\Theta_L} \modelssem \colL{\phi_i, \Delta_L'},
            \theta \vee \Lambda \in E\]
        for both $i \in \{1,2\}$, so we have
\[
            \colL{\Theta_L} \modelssem \colL{\phi_1 \wedge \phi_2, \Delta_L'},
            \theta \vee \Lambda \in E
\]
as desired. Since $\neg \theta = \neg \theta_1\IH \wedge \neg\theta_2\IH$, we get
        \[
          \colR{\Theta_R} \modelssem \colR{\Delta_R}, \neg \theta
        \]
        by combining both properties from the right column.
      \item The dual case where $\colR{\Delta_R} = \colR{\phi_1 \wedge \phi_2, \Delta_R'}$ is handled similarly, except that we set
$\theta \eqdef \theta_1\IH \wedge \theta_2\IH$.
    \end{itemize}
  \item Suppose the last rule applied is $\vee$ with principal formula $\phi_1 \vee \phi_2$.
    Depending on whether
    $\colL{\Delta_L} = \colL{\phi_1 \vee \phi_2, \Delta_L'}$
    or
    $\colR{\Delta_R} = \colR{\phi_1 \vee \phi_2, \Delta_R'}$, the proof will
    end with one of the following steps
\[
\dfrac{
\colL{\Theta_L}, \colR{\Theta_R} \vdash
\colL{\phi_1, \phi_2, \Delta_L'},\colR{\Delta_R},\goalList}
{\colL{\Theta_L}, \colR{\Theta_R} \vdash
\colL{\phi_1 \vee \phi_2, \Delta_L'},\colR{\Delta_R},\goalList}
\qquad \text{or} \qquad
\dfrac{
\colL{\Theta_L}, \colR{\Theta_R} \vdash
\colL{\Delta_L'},\colR{\phi_1,\phi_2,\Delta_R},\goalList}
{\colL{\Theta_L}, \colR{\Theta_R} \vdash
\colL{\Delta_L'},\colR{\phi_1 \vee \phi_2, \Delta_R},\goalList}
\]
In both cases, we apply the inductive hypothesis according to the obvious
splitting of contexts and \msets of formulas, to get an $\nrc$ definition $E\IH$
along with  a formula $\theta\IH$ that satisfy the desired semantic property.
We set $E \eqdef E\IH$ and $\theta \eqdef \theta\IH$.
\item Suppose the last rule applied is $\forall$ with principal formula $\forall x \in b. \phi$. As in the previous
case,  
    depending on  whether
    $\colL{\Delta_L} = \colL{\forall x \in b. \phi, \Delta_L'}$
    or
    $\colR{\Delta_R} = \colR{\forall x \in b. \phi, \Delta_R'}$, the proof will
    end with one of the following steps (assuming $y$ is fresh below)
\[
\dfrac{
\colL{\Theta_L, y \in b}, \colR{\Theta_R} \vdash
\colL{\phi[y/x], \Delta_L'},\colR{\Delta_R},\goalList}
{\colL{\Theta_L}, \colR{\Theta_R} \vdash
\colL{\forall x \in b. \phi, \Delta_L'},\colR{\Delta_R},\goalList}
\qquad \text{or} \qquad
\dfrac{
\colL{\Theta_L}, \colR{\Theta_R, y \in b} \vdash
\colL{\Delta_L'},\colR{\phi[y/x],\Delta_R},\goalList}
{\colL{\Theta_L}, \colR{\Theta_R} \vdash
\colL{\Delta_L'},\colR{\forall x \in b. \phi, \Delta_R},\goalList}
\]
In both cases, we again apply the inductive hypothesis according to the obvious
splitting of contexts and \msets of formulas to get an $\nrc$ definition $E\IH$
and a formula $\theta\IH$ that satisfy the desired semantic property.
We set $E \eqdef E\IH$ and $\theta \eqdef \theta\IH$.
\item Now we consider  the case where the
last rule applied is $\exists$. Here we have two main subcases,
  according to whether the formula that is instantiated in the premise, that
  is, the principal formula, belongs
  to $\goalList$ or not. In the first case, we have two further subcases according to whether the instantiated formula
  still has a leading existential quantifier or not.

  \begin{itemize}
    \item If the principal formula is of the shape $\cG = \exists y \in r ~ \forall z \in c.  ~ (\lambda(z) \equi \rho(z, y))$
      (so in particular, $\cG$ has a single leading existential quantifier)
      and is
      a member of $\goalList$,
      the proof necessarily has shape 
      \[
        \scalebox{0.9}{
        \text{
          \AXC{$\colL{\Theta_L}, \colR{\Theta_R}, \colL{x \in c} \vdash \colL{\Delta_L}, \colR{\Delta_R},
      \colR{\neg \rho(x,w)}, \colL{\lambda(x)},
      \goalList$}
\myLeftLabel{$\vee$}
      \UIC{$\colL{\Theta_L}, \colR{\Theta_R}, \colL{x \in c} \vdash \colL{\Delta_L}, \colR{\Delta_R},
\rho(x, w) \imp \lambda(x),
      \goalList$}
      \AXC{$\colL{\Theta_L}, \colR{\Theta_R}, \colL{x \in c} \vdash \colL{\Delta_L}, \colR{\Delta_R},
      \colL{\neg \lambda(x)}, \colR{\rho(x, w)},
      \goalList$}
\myLeftLabel{$\vee$}
\UIC{$\colL{\Theta_L}, \colR{\Theta_R}, x \in c \vdash \colL{\Delta_L}, \colR{\Delta_R},
\lambda(x) \imp \rho(x, w),
      \goalList$}
\myLeftLabel{$\wedge$}
\BIC{$\colL{\Theta_L}, \colR{\Theta_R}, x \in c \vdash \colL{\Delta_L}, \colR{\Delta_R},
\lambda(x) \equi \rho(x, w),
      \goalList$}
\myLeftLabel{$\forall$}
\UIC{$\colL{\Theta_L}, \colR{\Theta_R} \vdash \colL{\Delta_L}, \colR{\Delta_R},
\forall z \in c.  ~ (\lambda(z) \equi \rho(z, w)),
      \goalList$}
\myLeftLabel{$\exists$}
\UIC{$\colL{\Theta_L}, \colR{\Theta_R} \vdash \colL{\Delta_L}, \colR{\Delta_R}, \goalList$}
\DisplayProof
        }}\]
where $x$ is a fresh variable
So in particular, we have two strict subproofs with respective conclusions
      \[
\begin{array}{llcl}
  &  \colL{\Theta_L, x \in c}, \colR{\Theta_R} &\vdash & \colL{\lambda(x), \Delta_L}, \colR{\neg \rho(x,w), \Delta_R},\goalList\\
\text{and}&
      \colL{\Theta_L, x \in c}, \colR{\Theta_R} &\vdash & \colL{\neg\lambda(x),\Delta_L}, \colR{\rho(x,w), \Delta_R},
      \goalList
\end{array}
\]
Applying the inductive hypothesis, we obtain $\nrc$ expressions $E_1\IH$, $E_2\IH$
and formulas $\theta_1\IH, \theta_2\IH$ which contain free variables in $C \cup \{x\}$ such that all of the following hold
\begin{align}
&
  \colL{\Theta_L, x \in c} \modelssem \phantom{\neg}\colL{\lambda(x), \Delta_L}, \theta_1\IH \vee \Lambda  \in E_1\IH
\label{align:1} \\
  \text{and} \qquad &
  \colL{\Theta_L, x \in c} \modelssem \colL{\neg \lambda(x), \Delta_L}, \theta_2\IH \vee \Lambda  \in E_2\IH
  \label{align:2} \\
  \text{and} \qquad &
  \phantom{x \in c,}\colR{\Theta_R} \modelssem \colR{\neg\rho(x,w),\Delta_R}, \neg \theta_1\IH
  \label{align:3} \\
  \text{and} \qquad &
  \phantom{x \in c,} \colR{\Theta_R} \modelssem \phantom{\neg}\colR{\rho(x,w), \Delta_R}, \neg \theta_2\IH
  \label{align:4}
\end{align}
With this in hand, we set
\[ \theta \eqdef \exists x \in c. ~ \theta_1\IH \wedge \theta_2\IH \qquad \text{and} \qquad
         E \eqdef \left\{\left\{ x \in c \mid \theta_2\IH\right\}\right\} ~\cup ~ \bigcup\left\{ E_1\IH \cup E_2\IH \mid x \in c \right\}
      \]
Note in particular that the free variables of $E$ and $\theta$ are contained in $C$, since we bind $x$.
The bindings of $x$ have radically different meaning across the two main components 
$E_1 \eqdef \left\{\left\{ x \in c \mid \theta_2\IH\right\}\right\}$ and $ E_2 \eqdef \left\{ E_1\IH \cup E_2\IH \mid x \in c \right\}$
of $E = E_1 \cup E_2$. $E_1$ consists of a single definition corresponding to the restriction of $c$ to $\theta_2\IH$, and there $x$
plays the role of an element being defined. On the other hand, $E_2$ corresponds to the joining of all the definitions obtained inductively,
which may contain an $x \in c$ as a parameter. So  we have two families of potential definitions for $\Lambda$ indexed by $x \in c$ that we join together.
Now let us show that we have the desired semantic properties. First we need to show that
$E$ contains a definition for $\Lambda$ under the right hypotheses, i.e.,
\begin{small}
\begin{align}
\colL{\Theta_L} \modelssem \colL{\Delta_L}, \exists x \in c. \; \theta_1\IH \wedge \theta_2\IH, \Lambda \in \left( 
\left\{\left\{ x \in c \mid \theta_2\IH\right\}\right\} ~\cup ~ \bigcup\left\{ E_1\IH \cup E_2\IH\mid x \in c \right\}\right)
\label{align:5}
\end{align}
\end{small}
which can be rephrased as
\begin{align*}
\colL{\Theta_L} \modelssem \colL{\Delta_L}, \exists x \in c. \; \theta_1\IH \wedge \theta_2\IH,
\Lambda = \left\{ x \in c \mid \theta_2\IH\right\}, \exists x \in c. \; \Lambda \in E_1\IH \cup E_2\IH
\end{align*}
Now concentrate on the statement
$\Lambda = \left\{x \in c \mid \theta_2\IH\right\}$.
It would follow from 
the two inclusions
$\Lambda \subseteq \left\{x \in c \mid \theta_2\IH\right\}$
and
$\left\{x \in c \mid \theta_2\IH\right\} \subseteq \Lambda$, so, recalling that
$\Lambda = \{ x \in c \mid \lambda(x)\}$, the overall conclusion would follow from
having
\begin{align*}
&
\colL{\Theta_L}, x \in c \modelssem \colL{\Delta_L}, \exists x \in c. \; \theta_1\IH \wedge \theta_2\IH,
  \lambda(x) \imp \theta_2\IH, \exists x \in c. \; \Lambda \in E_1\IH \cup E_2\IH
\\
  \text{and}\qquad
  &
\colL{\Theta_L}, x \in c \modelssem \colL{\Delta_L}, \exists x \in c. \; \theta_1\IH \wedge \theta_2\IH,
  \theta_2\IH \imp \lambda(x), \exists x \in c. \; \Lambda \in E_1\IH \cup E_2\IH
\end{align*}
Those in turn follow from the following two statements
\begin{align*}
&
\colL{\Theta_L}, x \in c \modelssem \colL{\Delta_L}, \theta_1\IH \wedge \theta_2\IH,
  \neg \lambda(x), \theta_2\IH, \Lambda \in E_1\IH \cup E_2\IH
\\
  \text{and}\qquad
  &
\colL{\Theta_L}, x \in c \modelssem \colL{\Delta_L}, \theta_1\IH \wedge \theta_2\IH,
  \neg\theta_2\IH, \lambda(x), \Lambda \in E_1\IH \cup E_2\IH
\end{align*}
which are straightforward consequences of~\ref{align:2} and~\ref{align:1} respectively.
This concludes the proof of~\ref{align:5}.

Now we only need to prove a final property, which is
\begin{align*}
\colR{\Theta_R} \modelssem \colR{\Delta_R}, \forall x \in c. \; \neg \theta_1\IH \vee \neg \theta_2\IH 
\end{align*}
which is equivalent to the validity of
\begin{align*}
\colR{\Theta_R}, x \in c \modelssem \colR{\Delta_R}, \neg \theta_1\IH \vee \neg \theta_2\IH
\end{align*}
which can be obtained by combining~\ref{align:3} and~\ref{align:4} with excluded
middle for $\rho(x,w)$.
\item If the principal formula is of the shape $\cG = \exists r' \in r ~ \cG'$ where
  $\cG'$ begins with another existential quantifier and $\cG$ is a member of $\goalList = \cG, \goalList'$,
      the proof necessarily has shape
\[
\dfrac{\colL{\Theta_L}, \colR{\Theta_R} \vdash \colL{\Delta_L}, \colR{\Delta_R}, \cG', \goalList}
{\colL{\Theta_L}, \colR{\Theta_R} \vdash \colL{\Delta_L}, \colR{\Delta_R}, \goalList}
\]
then we can apply the induction hypothesis where in lieu of $\goalList$ we have $\cG', \goalList$ and
obtain $\theta\IH$ and $E\IH$. It is then clear that we can simply set $\theta \eqdef \theta\IH$ and $E \eqdef E\IH$.
\item If the principal formula is not a member of $\goalList$, then we have two subcases corresponding to
  whether the principal formula under consideration occurs in $\colL{\Delta_L}$ or $\colR{\Delta_R}$,
  and whether the relevant membership statement that witnesses the instantiation
  is a member of $\colL{\Theta_L}$ or $\colL{\Theta_R}$. Let us list all of the different alternatives
    \begin{itemize}
      \item If the last step of the proof has shape
\[
  \dfrac{\colL{\Theta_L'}, \colL{w \in b}, \colR{\Theta_R} \vdash \colL{\Delta_L'}, \colL{\phi[w/x], \exists x \in b~ \phi}, \colR{\Delta_R},
        \goalList}{\colL{\Theta_L'}, \colL{w \in b}, \colR{\Theta_R} \vdash \colL{\Delta_L'}, \colL{\exists x \in b~ \phi}, \colR{\Delta_R}, \goalList}
  \]
      with $\colL{\Theta_L} = \colL{\Theta_L', w \in b}$ and $\colL{\Delta_L} = \colL{\Delta_L', \exists x \in b ~ \phi}$,
      we can conclude by setting $\theta \eqdef \theta\IH$ and $E \eqdef E\IH$, essentially because the set of free variables $L, R$ and $C$ can be taken
      to be the same in the premise.
      \item In the dual case where the last step has shape
\[
  \dfrac{\colL{\Theta_L}, \colR{\Theta_R', w \in b} \vdash \colL{\Delta_L}, \colR{\Delta_R', \phi[w/x], \exists x \in b~ \phi},\goalList}{\colL{\Theta_L},
        \colR{\Theta_R, w \in b} \vdash \colL{\Delta_L}, \colR{\Delta_R', \exists x \in b~ \phi},\goalList}
\]
      with $\colR{\Theta_R} = \colR{\Theta_R', w \in b}$ and $\colR{\Delta_R} = \colL{\Delta_R', \exists x \in b ~ \phi}$,
      we can also conclude immediately by setting $\theta \eqdef \theta\IH$ and $E \eqdef E\IH$.
      \item If the last step has shape 
\[
  \dfrac{\colL{\Theta_L'}, \colL{w \in b}, \colR{\Theta_R} \vdash \colL{\Delta_L}, \colR{\Delta_R', \phi[\colL{w}/x], \exists x \in b~ \phi},
        \goalList}{\colL{\Theta_L'}, \colL{w \in b}, \colR{\Theta_R} \vdash \colL{\Delta_L}, \colR{\Delta_R', \exists x \in b~ \phi}, \goalList}
\]
      with $\colL{\Theta_L} = \colL{\Theta_L', w \in b}$ and $\colR{\Delta_R} = \colR{\Delta_R', \exists x \in b ~ \phi}$,
      we need to do something non-trivial. We can still use the inductive hypothesis to obtain $\theta\IH$ and $E\IH$, but they may feature variables
      $x_1, \ldots, x_n$ of $w$ that are in $L$ as free variables. But we also have that the free variables of $b$ are included in $C$.
      With that in mind, we can set $\theta \eqdef \existsvars{x_1 \ldots x_n} w \in b\,.\, \theta\IH$ and $E \eqdef \bigcup
        \left\{ E\IH \mid w \in b \right\}$.
      \item If the last step has shape
\[
  \dfrac{\colL{\Theta_L}, \colR{\Theta_R, w \in b} \vdash \colL{\Delta_L', \phi[\colR{w}/x], \exists x \in b~ \phi}, \colR{\Delta_R},
        \goalList}{\colL{\Theta_L}, \colR{\Theta_R', w \in b} \vdash \colL{\Delta_L', \exists x \in b~ \phi}, \colR{\Delta_R}, \goalList}
\]
      with $\colR{\Theta_R} = \colR{\Theta_R', w \in b}$ and $\colL{\Delta_L} = \colL{\Delta_L', \exists x \in b ~ \phi}$,
      we proceed similarly by setting $\theta \eqdef \forall x_1 \ldots x_n| w \in b\,.\, \theta\IH$ and $E \eqdef \bigcup 
      \left\{ E\IH \mid w \in b \right\}$.
    \end{itemize}
\end{itemize}
\item The case of the $=$ rule can be handled exactly as the $\top$ rule.
\item For the $\neq$ rule, we distinguish several subcases:
  \begin{itemize}
    \item If we have $\colL{\Delta_L} = \colL{y \neq_\ur z, \alpha[y/x], \Delta_L'}$
      or $\colR{\Delta_R} = \colR{y \neq_\ur z, \alpha[y/x], \Delta_R'}$, so that the
      last step has one of the two following shapes
      \[
        \dfrac{
          \colL{\Theta_L}, \colR{\Theta_R} \vdash \colL{y \neq_\ur z, \alpha[y/x], \alpha[z/x], \Delta_L'
        }, \colR{\Delta_R'}, \goalList
      }{
          \colL{\Theta_L}, \colR{\Theta_R} \vdash \colL{y \neq_\ur z, \alpha[y/x], \Delta_L'
        }, \colR{\Delta_R'}, \goalList
      }
      \qquad
        \dfrac{
          \colL{\Theta_L}, \colR{\Theta_R} \vdash \colL{\Delta_L'}, \colR{y \neq_\ur z, \alpha[y/x], \alpha[z/x], \Delta_R'
        }, \goalList
      }{
        \colL{\Theta_L}, \colR{\Theta_R} \vdash \colL{\Delta_L'}, \colR{y \neq_\ur z, \alpha[y/x], \Delta_R'
        }, \goalList
      }
    \]
    we can apply the induction hypothesis to obtain some $\theta\IH$ and $E\IH$ such that setting $\theta \eqdef \theta\IH$ and $E \eqdef E\IH$ solves this
      subcase; we leave checking the additional properties to the reader.
      \item Otherwise, if we have
        $\colL{\Delta_L} = \colL{y \neq_\ur z, \Delta_L'}$,
        $\colR{\Delta_R} = \colR{\alpha[y/x], \Delta'_R}$ and a last step of
        shape
      \[
        \dfrac{
          \colL{\Theta_L}, \colR{\Theta_R} \vdash \colL{y \neq_\ur z, \Delta_L'
        }, \colR{\alpha[y/x], \alpha[z/x],\Delta_R'}, \goalList
      }{
          \colL{\Theta_L}, \colR{\Theta_R} \vdash \colL{y \neq_\ur z, \Delta_L'
        }, \colR{\alpha[y/x], \Delta_R'}, \goalList
      }
    \]
      In that case, the inductive hypothesis gives $\theta\IH$ and $E\IH$ with
      free variables in $C \cup \{z\}$ such that
      \[
        \colL{\Theta_L} \modelssem \colL{y \neq_\ur z, \Delta_L'}, \theta\IH, \Lambda \in E\IH
      \quad \text{and} \quad \colR{\Theta_R} \modelssem \colR{\alpha[y/x], \alpha[z/x], \Delta_R'}, \neg \theta\IH \]
    We then have two subcases according to whether $z \in C$ or not
    \begin{itemize}
        \item If $z \in C$, we can take $\theta \eqdef \theta\IH \wedge y =_\ur z$ and $E \eqdef E\IH$. Their free variables are in $C$
          and we only need to check
      \[
        \colL{\Theta_L} \modelssem \colL{y \neq_\ur z, \Delta_L'}, \theta\IH \wedge y =_\ur z, \Lambda \in E\IH
      \quad \text{and} \quad \colR{\Theta_R} \modelssem \colR{\alpha[y/x], \Delta_R'}, \neg \theta\IH, y \neq_\ur z \]
which follow easily from the induction hypothesis.
        \item Otherwise, we take $\theta \eqdef \theta\IH[y/z]$ and $E \eqdef E\IH[y/z]$.
          In that case, note that we have $\alpha[z/x][y/z] = \alpha[y/x]$ (which would not be necessarily the case if $z$ belonged to $C$).
          This allows to conclude that we have
      \[
        \colL{\Theta_L} \modelssem \colL{y \neq_\ur z, \Delta_L'}, \theta\IH[y/z], \Lambda \in E\IH[y/z]
        \quad \text{and} \quad \colR{\Theta_R} \modelssem \colR{\alpha[y/x], \Delta_R'}, \neg \theta\IH[y/z] \]
          directly from the induction hypothesis.
      \end{itemize}
      \item Otherwise, if we have
        $\colL{\Delta_L} = \colL{\alpha[y/x], \Delta_L'}$,
        $\colR{\Delta_R} = \colR{y \neq_\ur z, \Delta'_R}$ and a last step of
        shape
      \[
        \dfrac{
          \colL{\Theta_L}, \colR{\Theta_R} \vdash
          \colL{\alpha[y/x], \alpha[z/x],\Delta_L'},
          \colR{y \neq_\ur z, \Delta_R'},
\goalList
      }{
  \colL{\Theta_L}, \colR{\Theta_R} \vdash
          \colL{\alpha[y/x],\Delta_L'},
          \colR{y \neq_\ur z, \Delta_R'},
\goalList
        }
    \]
      In that case, the inductive hypothesis gives $\theta\IH$ and $E\IH$ with
      free variables in $C \cup \{z\}$ such that
      \[
        \colL{\Theta_L} \modelssem \colL{\alpha[y/x],\alpha[z/x], \Delta_L'}, \theta\IH, \Lambda \in E\IH
      \qquad \text{and} \qquad \colR{\Theta_R} \modelssem \colR{y \neq_\ur z, \Delta_R'}, \neg \theta\IH \]
    We then have two subcases according to whether $z \in C$ or not
    \begin{itemize}
        \item If $z \in C$, we can take $\theta \eqdef \theta\IH \vee y \neq_\ur z$ and $E \eqdef E\IH$. Their free variables are in $C$
          and we only need to check
      \[
        \colL{\Theta_L} \modelssem \colL{\alpha[y/x], \Delta_L'}, \theta\IH \vee y \neq_\ur z, \Lambda \in E\IH
      \hspace{1.6em} \text{and} \hspace{1.6em} \colR{\Theta_R} \modelssem \colR{y \neq_\ur z, \Delta_R'}, \neg \theta\IH \wedge y =_\ur z \]
which follow easily from the induction hypothesis.
        \item Otherwise, we take $\theta \eqdef \theta\IH[y/z]$ and $E \eqdef E\IH[y/z]$.
          In that case, note that we have $\alpha[z/x][y/z] = \alpha[y/x]$ (which would not be necessarily the case if $z$ belonged to $C$).
          This allows to conclude that we have
      \[
        \colL{\Theta_L} \modelssem \colL{\alpha[y/x], \Delta_L'}, \theta\IH[y/z], \Lambda \in E\IH[y/z]
        \qquad \text{and} \qquad \colR{\Theta_R} \modelssem \colR{y \neq_\ur z, \Delta_R'}, \neg \theta\IH[y/z] \]
          directly from the induction hypothesis.
      \end{itemize}
  \end{itemize}
\item If the last rule applied is $\times_\eta$, the proof has shape
\[
  \dfrac{\colL{\Theta_L}[\tuple{x_1,x_2}/x],
         \colR{\Theta_R}[\tuple{x_1,x_2}/x] \vdash 
         \colL{\Delta_L}[\tuple{x_1,x_2}/x],
         \colR{\Delta_R}[\tuple{x_1,x_2}/x], \goalList[\tuple{x_1,x_2}/x]}{\colL{\Theta_L}, \colR{\Theta_R} \vdash \colL{\Delta_L}, \colR{\Delta_R}, \goalList}
\]
and one applies the inductive hypothesis as expected to get $\theta\IH$ and
$E\IH$ such that
\[
\begin{array}{ll@{}l}
  & \colL{\Theta_L}[\tuple{x_1,x_2}/x] &\modelssem
         \colL{\Delta_L}[\tuple{x_1,x_2}/x], \theta\IH, \Lambda[\tuple{x_1,x_2}/x] \in E\\
\text{and}&
         \colR{\Theta_R}[\tuple{x_1,x_2}/x] &\modelssem 
         \colR{\Delta_R}[\tuple{x_1,x_2}/x], \neg\theta\IH
\end{array}
\]
and with free variables included in $C$ if $x \notin C$ or
$\left(C \cup \{x_1, x_2\}\right) \setminus \{x\}$ otherwise.
In both cases, it is straightforward to check that taking
$\theta \eqdef \theta\IH[\pi_1(x)/x_1,\pi_2(x)/x_2]$ and
$E \eqdef E\IH[\pi_1(x)/x_1,\pi_2(x)/x_2]$ will yield the desired result.
\item Finally, if the last rule applied is the $\times_\beta$ rule, it has shape
\[
  \dfrac{
    (\colL{\Theta_L},
         \colR{\Theta_R})[x_i/x] \vdash 
         (\colL{\Delta_L},
         \colR{\Delta_R},\goalList')[x_i/x]}{
         (\colL{\Theta_L},
         \colR{\Theta_R})[\pi_i(\tuple{x_1,x_2})/x] \vdash 
         (\colL{\Delta_L},
         \colR{\Delta_R},\goalList')[\pi_i(\tuple{x_1,x_2})/x]}
\]
and we can apply the induction hypothesis to get satisfactory
$\theta\IH$ and $E\IH$ (moving from $\goalList$ to $\goalList'[x_i/x]$ is unproblematic, as we can assume the
lemma works for $\goalList$ with arbitrary subformulas $\lambda$ and $\rho$); it is easy to see that we can set $\theta \eqdef \theta\IH$
and $E \eqdef E\IH$ and conclude.
\end{itemize}

This completes the proof of Lemma~\ref{lem:nrcparamcoll}.

\subsection{Proof of the main result} \label{subsec:proof:ain}
We now turn to the proof of our second main result, which we recall from the earlier
subsection:
\thmmainset*

We have as input a proof of
\[ \phi(\vvi, \vva, o) \wedge \phi(\vvi, \vva', o') \imp o \equiv_T o'\]
and we want an $\nrc$ expression $E(\vvi)$ such that
\[ \phi(\vvi, \vva, o) \modelssem E(\vvi) \equiv_T o\]

This will be a consequence of the following theorem.
\begin{restatable}{thm}{mainsetinter}
\label{thm:mainsetinter}
Given $\deltazero$ $\varphi(\vvi, \vva, o)$ and $\psi(\vvi, \vvb, o')$
together with a {\focused}
proof with conclusion
  \[\Theta(\vvi, \vva, r); \; \varphi(\vvi, \vva, r), \psi(\vvi, \vvb, o') \vdash \exists {r' \in_p o'}. \; r \equiv_T r'\]
we can compute in polynomial time an $\nrc$ expression $E(\vvi)$ such that
  \[ \Theta(\vvi, \vva, r); \; \varphi(\vvi, \vva, r), \psi(\vvi, \vvb, o') ~\modelssem~ r \in E(\vvi)\]
\end{restatable}

That is, we can find an $\nrc$ query that ``collects answers''.
Assuming Theorem \ref{thm:mainsetinter}, let's prove the main result.

\begin{proof}[Proof of Theorem \ref{thm:mainset}]
  We assume $o$ has a set type, deferring the simple product and Ur-element
  cases (the latter using $\nrcget$) to Appendix~\ref{app:mainthm}.
Fix  an implicit definition of $o$ up to extensionality relative
to $\varphi(\vvi, \vva, o)$ and a {\focused} proof of
  \[\varphi(\vvi, \vva, o) \; \wedge \; \varphi(\vvi, \vva', o') ~\proves~ o \equiv_{\sett(T)} o' \]
We can apply Lemma \ref{lem:effnestedproof-down}, in the simple case where
$p$ is the ``empty path'',
to obtain a {\focused} derivation of
  \begin{align}
r \in o ; \; \varphi(\vvi, \vva, o), \varphi(\vvi, \vva', o')  
    \vdash~ \exists r' \inq o' \; r \equiv_T r' \label{algn:ent1}
  \end{align}
Then applying Theorem \ref{thm:mainsetinter} gives
an $\nrc$ expression $E(\vvi)$ such that
\[
\varphi(\vvi, \vva, o) \wedge r \memmac o \wedge \; \varphi(\vvi, \vva', o') ~  \modelssem ~ r \in E(\vvi)
\]
Thus, the object determined by $\vvi$ is always contained in $E(\vvi)$.
Recall that by (\ref{algn:ent1}), we have a derivation of
  \[
    r \in o; \; \varphi(\vvi, \vva, o) \vdash \varphi(\vvi, \vva', o')  \rightarrow \exists r' 
  \inq o'\; r \equiv_T r'\]
and applying interpolation (Theorem \ref{thm:interpolationdeltafocus}) to that gives
  a $\deltazero$ formula $\kappa(\vec i,  r)$ such that the following are valid
\begin{align}
  r \in o \wedge \varphi(\vvi, \vva, o) \imp \kappa(\vvi, r) \label{algn:inter1}\\
  \kappa(\vvi, r) \wedge \varphi(\vvi, \vva', o') \imp \exists r' \inq o'. \; r \equiv_T r' \label{algn:inter2}
\end{align}

  We claim that $E_\kappa(\vvi) = \left\{ x \in E(\vvi) \mid \kappa(\vvi, x)\right\}$ is the
desired $\nrc$ expression. To show this, assume $\varphi(\vvi, \vva, o)$ holds.
  We know already that $o \subseteq E(\vvi)$ and, by~(\ref{algn:inter1}), every $r \in o$ satisfies
  $\kappa(\vvi, o)$, so $o \subseteq E_\kappa(\vvi)$.
  Conversely, if $x \in E_{\kappa}(\vvi)$, we have $\kappa(\vvi,x)$, so by~(\ref{algn:inter2}),
  we have that $x \in o$, so $E_\kappa(\vvi) \subseteq o$. So $E_\kappa(\vvi) = o$, which concludes the proof.
\end{proof}

We now turn to the proof of Theorem \ref{thm:mainsetinter}.

\myparagraph{Proof of Theorem \ref{thm:mainsetinter}}
We prove the theorem %
by induction over the type $T$. 
We only prove the inductive step
for set types: the inductive case for products is straightforward.

For $T = \ur$, the base case of the induction, it is clear that we can take for $E$ an expression computing the set of all $\ur$-elements
in the transitive closure of $\vvi$. This can clearly be done in $\nrc$.

So now, we assume $T = \sett({T'})$ and that Theorem~\ref{thm:mainsetinter} holds up to ${T'}$.
We have a {\focused} derivation of
\begin{align}
\Theta; \; \varphi(\vvi, r), \; \psi(\vvi, o') ~\vdash~ \exists {r' \in_p o'}. \; r \equiv_{\sett(T')} r'
  \label{algn:inputder}
\end{align}
omitting the additional variables for brevity. 

From our input derivation, we can easily see that each element
of $r$ must be equivalent to some element below $o'$.
This is reflected in Lemma~\ref{lem:effnestedproof-down}, which allows us to
efficiently compute a proof of
\begin{align}
  \Theta, z \in r; \; \varphi(\vvi, r), \; \psi(\vvi, o') ~\vdash~ \exists z' \inq_{mp} o'. \; z \equiv_{T'} z'
  \label{align:inputderonelvldown}
\end{align}
We can then apply the inductive hypothesis of our main theorem at sort ${T'}$,
which is strictly smaller than $\sett({T'})$, on that new proof.
This yields an $\nrc$ expression $E\IH(\vvi)$
of type $\sett({T'})$ such that 
\[\Theta, z \in r; \; \varphi(\vvi,r), \psi(\vvi, o') \modelssem z \in E\IH(\vvi)\]
That is, our original hypotheses
entail $r \subseteq  E\IH(\vvi)$.

Thus, we have used the inductive hypothesis to get a ``superset expression''. But
now we want an expression that has $r$ as an element. We will do this
by  unioning a collection of definable subsets of $E\IH(\vvi)$.
To get these, we come back to our input derivation (\ref{algn:inputder}).
By Lemma \ref{lem:equivsettoequivalence}, we can efficiently compute a derivation of
\[\Theta; \; \varphi(\vvi, r), \psi(\vvi, o') \vdash \exists r' \inq_p
o'\,  \forall z \inq a  \;
(z \memmac r \equi z \memmac r')\]
where we take $a$ to be a fresh variable of sort $\sett({T'})$. 
Now, applying our $\nrc$ Parameter Collection result (Theorem~\ref{thm:mainflatcasebounded})
we obtain an $\nrc$ expression $E^{\mathrm{coll}}(\vvi, a)$ satisfying
\begin{align*}
\Theta; \; \varphi(\vvi, r), \psi(\vvi, o') ~\modelssem 
  a \cap r  \in E^{\mathrm{coll}}(\vvi, a)
\end{align*}
Now, recalling that we have $r \subseteq E\IH(\vvi)$ and instantiating $a$ to be
$E\IH(\vvi)$, we can conclude that
\begin{align*}
\Theta; \; \varphi(\vvi, r), \psi(\vvi, o') ~\modelssem 
  r  \in E^{\mathrm{coll}}(\vvi, E\IH(\vvi))
\end{align*}
Thus we can take $E^{\mathrm{coll}}(\vvi, E\IH(\vvi))$ as an explicit definition.

\myparagraph{Complexity}
Now let us sketch the complexity analysis of the underlying transformation.
The main induction is the type $T$ of the object to be defined, and most of the
lemmas we use have a complexity that
depend on the size of the proofs, which is commensurate with the size of the proof
tree multiplied by the size of the input sequent, that we will write $n$. Let
us call $C(T,n)$ the time complexity of our procedure
and show it can be taken to be polynomial. For the base case and the product case,
we have that $C(\ur, n)$ is $\bigO(n^k)$ for $k \ge 1$.
For set types $\sett(T)$, we first have a polynomial-time procedure in $n$ to obtain
the new proof in Lemma~\ref{align:inputderonelvldown}, and we have a recursive call
on this proof. Note that Lemma~\ref{lem:effnestedproof-down} also tells us
that this proof has size at most $n$, so the recursive call has complexity
at most $C(T, n)$. Then the subsequent transformations are simply polynomial-time
on the input proof, so we have for some exponent $k$ large enough
\[ C(\sett(T), n) \le C(T, n) + n^k\]
A similar analysis for products yields that
\[C(T_1 \times T_2, n) \le C(T_1, n) + C(T_2,n) + n^k\]
All in all, if we call $s(T)$ the size of a type defined in the obvious way,
we have
\[C(T, n) = \bigO(s(T) n^k)\]
by induction on $T$, so since $s(T) \le n$, the
overall time complexity is indeed polynomial in the size of the input derivation.
\qed

\section{Discussion and future work} \label{sec:conc}

Our first contribution implies that whenever a set of $\nrc$ views determines
an $\nrc$ query, the query is rewritable over the views in $\nrc$.
By our second result, from a proof witnessing determinacy in our {\focused} proof system, we can efficiently generate the rewriting. Both results apply to  a setting
where we have determinacy with respect to constraints and views, as in
Example \ref{ex:views}, or to general $\deltazero$ implicit definitions that
may not stem from views.

In terms of impact on databases, a crucial limitation of our work is that we do not yet
know how to find the proofs. In the case of relational data,
we know of many ``islands of decidability'' where proofs of determinacy can be found effectively -- e.g. 
for views and queries in guarded logics \cite{gnfjsl}.
But it remains open to find similar decidability results for views/queries in fragments of $\nrc$.

It is possible to use our proof system without full automation -- simply search for a proof, and then
when one finds one, generate the rewriting. We have had some success with this approach in the relational setting,
where standard theorem proving technology can be applied \cite{usijcai17}. But for the proof systems proposed
here, we do not have either our own theorem prover or a reduction to a system that has been implemented in the past.
The need to find proofs automatically is pressing since
our system is so low-level that it is difficult to do proofs by
hand. Indeed, a formal proof of implicit definability for Example \ref{ex:views}, 
or even the simpler Example \ref{ex:simplenesting},
 would come  to several pages. 

In \cite{benediktpradicpopl}, we introduce an intuitionistic version of our proof system,
and give a specialized algorithm for generating $\nrc$ transformations from proofs of implicit
definability within this system.  The algorithm for the intuitionistic case is considerably simpler
than for the proof systems we present here, and would probably make a good starting point for
an implementation of the system.

The implicit-to-explicit methodology requires
a proof of implicit definability, which implies implicit definability over all instances, not just finite ones.
This requirement is necessary:
one cannot hope to convert implicit definitions over finite instances
to explicit  $\nrc$ queries, even ineffectively.  We do not believe that this
is a limitation in practice. See Appendix~\ref{app:finite} for details.

For the effective Beth result,  the key proof tool was the $\nrc$ Parameter Collection theorem,
Theorem \ref{thm:mainflatcasebounded}. There is an intuition behind this
theorem  that concerns a general  setting,
where  we have a first-order theory $\Sigma$ that factors into a conjunction
of two formulas $\Sigma_L \wedge \Sigma_R$, and from this we have
a notion of a ''left formula'' (with predicates from $\Sigma_L$),
a ''right formula'' (predicates from $\Sigma_R$), and a ``common formula''
(all predicates occur in both $\Sigma_L$ and $\Sigma_R$). Under the
hypothesis that a left formula $\lambda$ is 
definable from a right formula with parameters, we can conclude that
the left formula must actually be definable from a common formula
with parameters: see Appendix~\ref{app:paramfo} for a formal version and the
corresponding proof.

Our work contributes to the broader topic of proof-theoretic vs model-theoretic techniques for interpolation and definability theorems. For Beth's theorem, 
there are reasonably short model-theoretic \cite{Lyndon59, ck} and proof-theoretic arguments \cite{craig57beth,fittingbook}.
In database terms, you can argue semantically that relational
algebra
is complete
for  rewritings of queries determined by views, and producing
a rewriting from a proof of determinacy is not   that difficult.
But for a number of results  on definability proved in the 60's and 70's 
\cite{changbeth,makkai,kueker,gaifman74}, there
are short model-theoretic arguments, but no proof-theoretic ones.
For our $\nrc$ analog of Beth's theorem, the situation is more similar
to the latter case:
the model-theoretic proof of completeness 
is relatively short and elementary, but generating explicit
definitions from proofs is much more challenging.
We hope that our results and  tools represent a step towards
providing  effective versions, and towards understanding the relationship between
model-theoretic and proof-theoretic arguments.

\input{ack}

\bibliographystyle{alphaurl}
\bibliography{obj}

\appendix
\input{appendix}

\end{document}

%% file: macros.tex
\definecolor{darkred}{rgb}{0.7,0,0}
\definecolor{darkblue}{rgb}{0,0,0.7}
\definecolor{colch}{rgb}{0.1,0.5,0.2}

\newcommand{\membersub}{\ni}
\newcommand{\vv}[1]{\vec{#1}}
\newcommand{\vvi}{\vv{i}}
\newcommand{\vva}{\vv{a}}
\newcommand{\vvb}{\vv{b}}

\newcommand{\kw}[1]{{\mathsf{#1}}\xspace}

\newcommand{\bigO}{\mathcal{O}}

\newcommand{\trans}{{\mathcal T}}

\newcommand{\freevars}{\kw{FV}}
\newcommand{\funall}{\kw{Fun}_{\kw{All}}}
\newcommand{\funfin}{\kw{Fun}_{\kw{Fin}}}

\newcommand{\groupq}{\kw{Group}}
\newcommand{\projq}{\kw{Proj}}
\newcommand{\filterq}{\kw{Filter}}

\newcommand{\convert}{\kw{Convert}}

\newcommand{\obj}{\kw{o}}
\newcommand{\oneobjth}{\kw{O}}
\newcommand{\oneobj}{\kw{obj}}
\newcommand{\inobj}{\kw{o}_{in}}
\newcommand{\outobj}{\kw{o}_{out}}
\newcommand{\map}{\kw{Map}}
\newcommand{\flatten}{\kw{Flatten}}

\newcommand{\booltt}{\kw{tt}}
\newcommand{\boolff}{\kw{ff}}

\newcommand{\nrcwget}{\nrc}%

\newcommand{\collapse}{\kw{Collapse}}

\newcommand{\exptime}{\kw{EXPTIME}}

\newtheorem{definition}[thm]{Definition}       

\renewcommand{\phi}{\varphi}

\newcommand{\issing}{\kw{IsSing}}
\newcommand{\sing}{\kw{Sing}}
\newcommand{\istwo}{\kw{IsTwo}}

\newcommand{\image}{\kw{Im}}

\newcommand{\allpairs}{\kw{AllPairs}}
\newcommand{\pair}{\kw{Pair}}
\newcommand{\subtype}{\leq}

\newcommand{\interpsort}{\tau}
\newcommand{\isnonuniformdefiner}{\delta}

\newcommand{\vflat}{V}
\newcommand{\basenest}{B}
\newcommand{\lossless}{\kw{lossless}}
\newcommand{\modelssem}{\models_{\kw{nested}}}
\newcommand{\equi}{\leftrightarrow}
\newcommand{\imp}{\rightarrow}

\newcommand{\inq}{\in}

\newcommand{\memmac}{\mathrel{\hat{\in}}}

\definecolor{colch}{rgb}{0.1,0.5,0.2}

\newcommand{\toconsider}[1]{} %

\newcommand{\la}{\langle}
\newcommand{\<}{\la}
\newcommand{\ra}{\rangle}
\renewcommand{\>}{\ra}

\newcommand{\ipol}[2]{\; :\; \la #1, #2\ra}
\newcommand{\ipolnarrow}[2]{: \la #1, #2\ra}

\newcommand{\defname}[1]{\emph{#1}}
\newcommand{\name}[1]{\emph{#1}}
\newcommand{\valid}{\models}
\newcommand{\entails}{\models}
\newcommand{\sep}{;\;\;}
\newcommand{\D}{\mathcal{D}\xspace}
\newcommand{\vs}{\vv{v}}

\newcommand{\tabrulename}[1]{\raisebox{-2.0ex}[0pt][0pt]{\small #1}}
\newcommand{\tabrulecond}[1]{\raisebox{-2.0ex}[0pt][0pt]{$#1$}}
\newcommand{\rulename}[1]{\textsc{#1}}
\newcommand{\SL}{\Gamma_L}
\newcommand{\SR}{\Gamma_R}

\newcommand{\ADDVARSLAM}{\vv{l}}
\newcommand{\ADDVARSRHO}{\vv{r}}
\newcommand{\ADDPREDSLAM}{\pred(\lambda)}
\newcommand{\ADDPREDSRHO}{\pred(\rho)}

\newcommand{\pdepd}{\kw{PDEPD}}

\newcommand{\ADDLVARS}{\FV^{\mathsf{RHS}}}
\newcommand{\PREDD}{\pred^{\mathsf{RHS}}}
\newcommand{\exdef}{\exists^{\mathsf{RHS}}}
\newcommand{\substdef}[2]{[#2/#1]^{\mathsf{RHS}}}%

\newcommand{\pospolarity}{\kw{EL}}%
\newcommand{\pospolaritysuper}{\kw{EL}}%

\newcommand{\subst}[2]{[#2/#1]}%

\newcommand{\prooftreesymbol}
{{\footnotesize
\setlength{\arraycolsep}{0pt}
\begin{array}{rcl}
.\;\;\;\;\;&.&\;\;\;\;.\\[-1.8ex]  
.\;\;\;\;&.&\;\;\;.\\[-1.8ex]
.\;\;\;&.&\;\;.\\[-1.8ex]
.\;\;&.&\;.\\[-1.8ex]
.\;&.&.\\[-1.8ex]
&.&
\end{array}}}

\newcolumntype{Z}[1]{>{\raggedleft\let\newline\\\arraybackslash\hspace{0pt}$}p{#1}<$}
\newcolumntype{X}[1]{>{\raggedright\let\newline\\\arraybackslash\hspace{0pt}$}p{#1}<$}
\newcolumntype{Y}[1]{>{\centering\let\newline\\\arraybackslash\hspace{0pt}$}p{#1}<$}          

\newenvironment{arrayprf}
               {\begin{array}{Z{1.8em}@{\hspace{1em}}l@{\hspace{0.5em}}l}}
               {\end{array}}

\newcounter{subprop}[thm]

\newcommand{\mset}{multiset\xspace}
\newcommand{\msets}{multisets\xspace}

\newcommand{\seq}{\vdash}

\newcommand{\pred}{\mathit{PR}}

\newcommand\colL[1]{\textcolor{darkred}{#1}}
\newcommand\colR[1]{\textcolor{darkblue}{#1}}
\newcommand\IH{^{\mathsf{IH}}}

\newcommand{\proves}{\vdash}

\newcommand{\myparagraph}[1]{\paragraph{#1}}

\newcommand{\myeat}[1]{}
\newcommand{\deltazero}{\Delta_0}
\newcommand{\incontext}{\in}
\newcommand\bnfeq{\mathrel{::=}}
\newcommand\bnfalt{\; | \;}
\newcommand\sett{{\sf Set}}

\newcommand{\nrc}{\kw{NRC}}

\newenvironment{myexmp}{\refstepcounter{myexmp}\par\medskip
\noindent\textbf{Example~\themyexmp.}}{\null\hfill$\triangleleft$\medskip}
\newcounter{myexmp}[section]
\renewcommand{\themyexmp}{\thesection.\arabic{myexmp}}

\newcommand{\inschema}{\aschema_{in}}
\newcommand{\outschema}{\aschema_{out}}

\newcommand{\aschema}{{\mathcal SCH}}
\newcommand{\domainof}{\kw{Domain}}

\newcommand{\unit}{\kw{Unit}}
\newcommand{\unitobj}{\<\>}

\newcommand\lam\lambda
\newcommand\case{{\sf case}}

\newcommand\bbN{\mathbb{N}}

\newcommand\eqdef{\mathrel{:=}}

\newcommand{\interp}{{\mathcal I}}

\newcommand{\nrcget}{\textsc{get}}

\newcommand\ur{\mathfrak{U}}

\newcommand{\ptime}{\kw{PTIME}}
\newcommand{\booltype}{\kw{Bool}}

\newcommand{\sig}{{\mathcal SIG}}
\newcommand{\sorts}{\kw{Sorts}}
\newcommand{\bool}{\booltype} %

\newcommand{\smallsort}{\sort_0}
\newcommand{\bigsort}{\sort_1}
\newcommand{\smallsorts}{\sorts_0}
\newcommand{\bigsorts}{\sorts_1}
\newcommand{\sort}{\kw{S}}

\newcommand{\ursort}{\ur}
\newcommand{\true}{\kw{True}}
\newcommand{\verify}{\kw{Verify}}

\newcommand\FV{FV}

\newcommand\tuple[1]{\langle #1 \rangle}

\newcommand\myLeftLabel[1]{\LeftLabel{\scriptsize #1}}

\newcommand{\normalized}{$\EL$-normalized\xspace}
\newcommand{\normalization}{$\EL$-normalization\xspace}
\newcommand{\Normalization}{$\EL$-Normalization\xspace}
\newcommand{\EL}{\pospolaritysuper}

\newcommand{\rulepair}{\times_\eta}
\newcommand{\ruleproj}{\times_\beta}
\newcommand{\ruledominated}{rule-do\-mi\-na\-ted\xspace}
\newcommand{\ALruledominated}{${\lor}{\land}\forall$-rule-do\-mi\-na\-ted\xspace}
\newcommand{\fdot}{\,.\,}
\newcommand{\nequ}{\neq_{\ur}}

 \newcommand{\focused}{$\EL$-normalized\xspace}

\newcommand{\folfocused}{FO-normalised\xspace}

\ifdefined\cA
\else
\newcommand\cA{\mathcal{A}}
\fi
\ifdefined\cB
\else
\newcommand\cB{\mathcal{B}}
\fi
\ifdefined\cC
\else
\newcommand\cC{\mathcal{C}}
\fi
\ifdefined\cD
\else
\newcommand\cD{\mathcal{D}}
\fi
\ifdefined\cE
\else
\newcommand\cE{\mathcal{E}}
\fi
\ifdefined\cF
\else
\newcommand\cF{\mathcal{F}}
\fi
\ifdefined\cG
\else
\newcommand\cG{\mathcal{G}}
\fi
\ifdefined\cH
\else
\newcommand\cH{\mathcal{H}}
\fi
\ifdefined\cI
\else
\newcommand\cI{\mathcal{I}}
\fi
\ifdefined\cJ
\else
\newcommand\cJ{\mathcal{J}}
\fi
\ifdefined\cK
\else
\newcommand\cK{\mathcal{K}}
\fi
\ifdefined\cL
\else
\newcommand\cL{\mathcal{L}}
\fi
\ifdefined\cM
\else
\newcommand\cM{\mathcal{M}}
\fi
\ifdefined\cN
\else
\newcommand\cN{\mathcal{N}}
\fi
\ifdefined\cO
\else
\newcommand\cO{\mathcal{O}}
\fi
\ifdefined\cP
\else
\newcommand\cP{\mathcal{P}}
\fi
\ifdefined\cQ
\else
\newcommand\cQ{\mathcal{Q}}
\fi
\ifdefined\cR
\else
\newcommand\cR{\mathcal{R}}
\fi
\ifdefined\cS
\else
\newcommand\cS{\mathcal{S}}
\fi
\ifdefined\cT
\else
\newcommand\cT{\mathcal{T}}
\fi
\ifdefined\cU
\else
\newcommand\cU{\mathcal{U}}
\fi
\ifdefined\cV
\else
\newcommand\cV{\mathcal{V}}
\fi
\ifdefined\cW
\else
\newcommand\cW{\mathcal{W}}
\fi
\ifdefined\cX
\else
\newcommand\cX{\mathcal{X}}
\fi
\ifdefined\cY
\else
\newcommand\cY{\mathcal{Y}}
\fi
\ifdefined\cZ
\else
\newcommand\cZ{\mathcal{Z}}
\fi
\ifdefined\bA
\else
\newcommand\bA{\mathbb{A}}
\fi
\ifdefined\bB
\else
\newcommand\bB{\mathbb{B}}
\fi
\ifdefined\bC
\else
\newcommand\bC{\mathbb{C}}
\fi
\ifdefined\bD
\else
\newcommand\bD{\mathbb{D}}
\fi
\ifdefined\bE
\else
\newcommand\bE{\mathbb{E}}
\fi
\ifdefined\bF
\else
\newcommand\bF{\mathbb{F}}
\fi
\ifdefined\bG
\else
\newcommand\bG{\mathbb{G}}
\fi
\ifdefined\bH
\else
\newcommand\bH{\mathbb{H}}
\fi
\ifdefined\bI
\else
\newcommand\bI{\mathbb{I}}
\fi
\ifdefined\bJ
\else
\newcommand\bJ{\mathbb{J}}
\fi
\ifdefined\bK
\else
\newcommand\bK{\mathbb{K}}
\fi
\ifdefined\bL
\else
\newcommand\bL{\mathbb{L}}
\fi
\ifdefined\bM
\else
\newcommand\bM{\mathbb{M}}
\fi
\ifdefined\bN
\else
\newcommand\bN{\mathbb{N}}
\fi
\ifdefined\bO
\else
\newcommand\bO{\mathbb{O}}
\fi
\ifdefined\bP
\else
\newcommand\bP{\mathbb{P}}
\fi
\ifdefined\bQ
\else
\newcommand\bQ{\mathbb{Q}}
\fi
\ifdefined\bR
\else
\newcommand\bR{\mathbb{R}}
\fi
\ifdefined\bS
\else
\newcommand\bS{\mathbb{S}}
\fi
\ifdefined\bT
\else
\newcommand\bT{\mathbb{T}}
\fi
\ifdefined\bU
\else
\newcommand\bU{\mathbb{U}}
\fi
\ifdefined\bV
\else
\newcommand\bV{\mathbb{V}}
\fi
\ifdefined\bW
\else
\newcommand\bW{\mathbb{W}}
\fi
\ifdefined\bX
\else
\newcommand\bX{\mathbb{X}}
\fi
\ifdefined\bY
\else
\newcommand\bY{\mathbb{Y}}
\fi
\ifdefined\bZ
\else
\newcommand\bZ{\mathbb{Z}}
\fi

\newcommand{\varpathref}{\mathcal{R}\mathit{ef}}

\newcommand{\existsvars}[1]{\exists #1|}

\newcommand{\elimprojvar}{\mathit{Elim}\text{-}\pi\text{-}\mathit{var}}
\newcommand{\elimprojpair}{\mathit{Elim}\text{-}\pi\text{-}\mathit{pair}}

\newcommand{\goalList}{\tilde\cG}

%% file: deltazero.tex
\subsection{\texorpdfstring{$\deltazero$}{Δ₀} formulas}
We need a logic appropriate for talking about nested relations.
A natural and well-known subset of first-order logic formulas with
a set membership relation are the $\deltazero$ formulas.
They are built up from equality of Ur-elements via Boolean operators
as well as relativized existential and universal quantification. All terms involving
tupling and projections are allowed. 
Our definition of $\deltazero$ formula is a variation of a well-studied notion in set theory \cite{jech}.

Formally, we deal with multi-sorted first-order logic, with sorts corresponding
to each of our types. We use the following syntax for $\deltazero$ formulas and terms.
Terms are built from variables using tupling and projections.
All formulas and terms are assumed to be well-typed in the obvious way, with the
expected sort of $t$ and $u$ being $\ursort$ in expressions
$t =_\ursort u$ and $t \neq_\ursort u$, and
in $\exists t \in_T u ~ \phi$  the sort of $t$ is $T$ and the sort of $u$ is  $\sett(T)$.
$$
\begin{array}{lcl}
t, u &\bnfeq& x \bnfalt \unitobj \bnfalt \< t, u \>  \bnfalt \pi_1(t) \bnfalt \pi_2(t)
\\
\varphi, \psi &\bnfeq& t =_\ursort t' \bnfalt t \neq_\ursort t' \bnfalt \top \bnfalt \bot \bnfalt \varphi \vee \psi \bnfalt \phi \wedge \psi \bnfalt \\
& & \forall x \in_T t~ \varphi(x) \bnfalt \exists x \in_T t~ \varphi(x)
\end{array}
$$
We call a formula \emph{atomic} if it does not have a strict subformula.
In the syntax presented above, atomic formulas are either of the form $t =_\ursort t'$ or
$t \neq_\ursort t'$ (note that there are no equalities for sorts other
than $\ursort$).
Negation $\neg \phi$ will be  defined as a macro
by induction on $\phi$ by dualizing every connective and then the atomic
formulas recursively.
Other connectives can be derived in the usual way on top of negation:
 $\phi \rightarrow \psi$ by $\neg \phi \vee \psi$.

More crucial is the fact that 
membership is not itself a formula, it is only used in the relativized quantifiers.
An \emph{extended $\deltazero$ formula} allows also membership atomic formulas
$x \in_T y$, $x \notin_T y$ and equalities $x =_T y$, $x \neq_T y$ at every type $T$.

The notion of an extended $\deltazero$  formula $\phi$ \emph{entailing} another formula
$\psi$ is the standard one in first-order logic,
meaning that every model of $\phi$ is a model of $\psi$.
We emphasize here that by every model, we include models where membership is not
extensional. We will require only one ``sanity axiom'':    projection and 
tupling commute.
An important point is that there is no distinction between entailment in such ``general models'' and entailment over nested relations.

\begin{prop}
For $\phi$ and $\psi$ are $\deltazero$, rather than extended $\deltazero$, $\phi$ entails $\psi$
iff every nested relation that models $\phi$ is a model of $\psi$.
\end{prop}

The point above is due to two facts. First,
we have neither $\in$ or  equality at higher types as predicates.
This guarantees that any model 
can be modified, without changing the truth value of $\deltazero$ formulas,
into a model satisfying extensionality: if we have $x$ and $y$ with
$(\forall z \in_T x~~ z \in_T y) ~~\wedge~~ (\forall z \in_T y~~ z \in_T x)$
then  $x$ and $y$ must be the same.
Secondly, a well-typed extensional  model is isomorphic to a nested relation,
by the well-known
Mostowski collapse construction that iteratively identifies elements that
have the same members.
The lack of  primitive membership and equality relations in $\deltazero$ formulas
allows us to avoid having to consider extensionality axioms, which would require
special handling in our proof system.  

Equality, inclusion and membership predicates ``up to extensionality''
may be defined as macros by induction on the involved types, while staying within $\deltazero$
formulas. Formally we have:
\begin{definition} \label{def:macros}
$$
\begin{array}{rcl}
t \memmac_T u &\eqdef& \exists z \in_T u\; t \equiv_T z \\
t \subseteq_T u &\eqdef& \forall z \in_T t ~~ z \memmac_T u\\
t \equiv_{\sett(T)} u &\eqdef& {t \subseteq_T u} ~~\wedge ~~{u \subseteq_T t} \\
t \equiv_\unit u &\eqdef& \top\\
t \equiv_\ur u &\eqdef& t =_\ur u\\
t \equiv_{T_1 \times T_2} u &\eqdef& {\pi_1(t) \equiv_{T_1} \pi_1(u)}~~ \wedge~~ {\pi_2(t) \equiv_{T_2} \pi_2(u)
}\\
\end{array}
$$
\end{definition}
We will use small letters for variables in $\deltazero$
formulas, except in examples when we sometimes use capitals
to  emphasize that an object is of set type.
We drop the type subscripts $T$ in bounded quantifiers, primitive memberships,
and macros $\equiv_T$ when clear.
Of course  membership-up-to-equivalence $\memmac$ and  membership  $\in$ agree
on extensional models. But $\in$ and $\memmac$ are not interchangeable on general models,
and hence are not interchangeable in $\deltazero$ formulas. For example:
\[
x \in y ,  x \in y' \models \exists z \in y ~ z \in y'
\]
But we do not have
\[
x \memmac y, x \memmac y' \models  \exists z \in y ~ z \in y'
\]

%% file: nrc.tex
\subsection{Nested Relational Calculus}
We review the main language for declaratively transforming  nested relations, Nested
Relational Calculus ($\nrc$). Variables occurring in expressions
are  typed, and each expression is  associated with 
an \emph{output type}, both of these being in the type system described above.
We let $\booltype$ denote the type $\sett(\unit)$. Then $\booltype$ has exactly
two elements, and will be used to simulate Booleans.
The grammar for $\nrc$ expressions  is presented in
Figure \ref{fig:nrc_type}.

\begin{figure}[t]
  \[
    \begin{array}{rl@{}l@{~~}r}
      E, E' \bnfeq& x& \bnfalt \unitobj \bnfalt \tuple{E, E'} \bnfalt \pi_1(E) \bnfalt \pi_2(E) \bnfalt
      & {\footnotesize \text{(variable, (un)tupling)}} 
      \\
       && \{E\} \bnfalt \nrcget_T(E) \bnfalt \bigcup \{ E \mid x \in E' \} 
      & {\footnotesize \text{((un)nesting, binding union)}}
\\
       && \bnfalt \emptyset \bnfalt E \cup E' \bnfalt E \setminus E'
& {\footnotesize \text{(finite unions, difference)}}
  \end{array}\]
  \vspace{-1em}
  \caption{$\nrc$ syntax (typing rules omitted)}
\label{fig:nrc_type}
\end{figure}

The definition
of the free and bound variables of an expression is standard,
the union operator $\bigcup \{ E \mid x \in R \}$ binding the variable $x$.
The semantics of these expressions should be fairly evident, see
\cite{limsoonthesis}.
If $E$ has type $T$, and has input (i.e. free) variables $x_1 \ldots x_n$
of types $T_1 \ldots T_n$, respectively, then the semantics associates with $E$
a function that 
given a binding associating each free variable with a value of the appropriate type,
returns an object of type $T$.
For example, the expression $\unitobj$ always returns the empty tuple, while
$\emptyset_T$ returns the empty set of type $T$.

As explained in prior work (e.g. \cite{limsoonthesis}), on top of
the $\nrc$ syntax above we can support richer operations as  ``macros''.
For every type $T$ there is an $\nrc$ expression
$=_T$ of type $\booltype$ representing equality of elements of type $T$.
In particular, there is an expression $=_\ur$ representing equality between
Ur-elements. 
 For every type $T$ there is an $\nrc$ expression $\in_T$ of type
$\booltype$ representing membership between an element of type $T$
in an element of type $\sett(T)$. We can define conditional expressions, 
joins, projections on $k$-tuples, and $k$-tuple formers.
Arbitrary arity tupling and projection operations  $\<E_1,\ldots E_n\>$, $\pi_j(E)$ for $j >2$
can be seen as abbreviations for a composition of binary operations.
Further
\begin{itemize}
\item
If $B$ is an expression of type $\bool$ and $E_1, E_2$ expressions of
type $T$, then there is an expression $\case(B,E_1,E_2)$ of type $T$ that implements ``if $B$ then $E_1$ else $E_2$''.
\item If $E_1$ and $E_2$ are expressions of type $\sett(T)$, then there are  expressions
$E_1 \cap E_2$  and $E_1  \setminus E_2$ of type $\sett(T)$.
\end{itemize}
The derivations of these are not difficult.
For example, the conditional
required by the first item is given by:
\[
\bigcup \{\{E_1\} \mid ~ x \in B\} \cup \bigcup \{\{E_2\} \mid x \in (\neg ~ B)\} 
\]

Finally, we note that $\nrc$ is closed under \emph{$\deltazero$ comprehension}:
if $E$ is in $\nrc$, $\phi$ is an extended 
$\deltazero$ formula, then we can efficiently form an expression  
$\{z \in E \mid \phi \}$ which returns the subset of $E$ such
that $\phi$ holds.
We  make use of these macros freely in our examples of $\nrc$,
such as Example \ref{ex:views}.

\begin{exa} \label{ex:nrc}
Consider an input schema including a binary relation $F: \sett( \ur \times \ur ) $.
The query $\trans_{\projq}$ with input
$F$ returning the projection of $F$ on the first component can be expressed in
$\nrc$ as  $\bigcup \{ \{\pi_1(f) \} \mid f \in F \}$.
The query $\trans_{\filterq}$ with input $F$ and also $v$ of type
$\ur$ that filters $F$ down to those pairs which agree with $v$ on the first component can
be expressed in $\nrc$ as $\bigcup \left\{\case([\pi_1(f) =_\ur v], \{f\}, \emptyset) \mid f \in F \right\}$.
Consider now the query $\trans_{\groupq}$
that groups $F$ on the first component, returning an object of type
$\sett( \ur \times  \sett(\ur))$.
The query
 can be expressed in $\nrc$ as $\bigcup \left\{ \{\< v,  \bigcup \{ \{\pi_2(f)\} \mid f \in \trans_{\filterq} \} \>\} \mid v \in \trans_{\projq}\right\}$.
Finally, consider the %
query $\trans_{\flatten}$ 
that flattens an input $G$ of type $\sett(\ur \times \sett(\ur))$ .
This can be expressed in $\nrc$ as 
\begin{align*}
\bigcup \left\{   \bigcup \{ \{ \<\pi_1(g), x\>  \} \mid x \in \pi_2(g)  \}  \mid g \in G \right\}
\end{align*}
\end{exa}

The language $\nrc$ as originally defined cannot express certain natural
transformations whose output type is $\ur$.
To get a canonical language for such transformations, above we included in our
$\nrc$ syntax a family of operations
$\nrcget_T : \sett(T) \to T$ that extracts the unique element from a singleton.
$\nrcget$ was considered in \cite{limsoonthesis}.
The semantics are: if $E$ returns a singleton set $\{x\}$, then $\nrcget_T(E)$
returns $x$; otherwise it returns some default object of the appropriate type.
In \cite{suciuthesis}, it is shown that $\nrcget$ is not expressible
in $\nrc$ at sort $\ursort$.
However, $\nrcget_T$ for general $T$ is definable from $\nrcget_\ursort$
and the other $\nrc$ constructs. 

Since $\nrcget$ will be needed for our key results, in the remainder of the paper, \emph{we will  write simply $\nrc$ for $\nrc$ as defined as usual, augmented
with $\nrcget$}. The role of
$\nrcget$ will only be for transformations that return something of type $\ursort$.

\subsection{Connections between $\nrc$ queries using \texorpdfstring{$\deltazero$}{Δ₀} formulas}
Since we have a Boolean type in $\nrc$, one may ask about the expressiveness
of $\nrc$ for defining transformations of shape $T_1, \ldots, T_n \to \bool$.
It turns out that they are equivalent to $\deltazero$ formulas.
This gives one justification for focusing on $\deltazero$ formulas.

\begin{prop} \label{prop:verify} There is a polynomial time
algorithm taking an extended $\deltazero$ formula $\phi(\vec x)$ as input and producing
 an $\nrc$ expression $\verify_\phi(\vec x)$ of type $\bool$
such that, for  any valuation in any nested relation,  $\verify_\phi(\vec x)$ returns true if and only if $\phi(\vec x)$ holds.
\end{prop}

This useful result is proved by an easy induction over $\phi$: see
Appendix~\ref{app:fointerpbool} for details.

In the opposite direction, given an $\nrc$ expression $E$ with input relations $\vec i$,
we can create 
a $\deltazero$ formula $\Sigma_E(\vec i, o)$ that
is an \emph{input-output specification of $E$}:
a formula such that $\Sigma_E$ implies $o=E(\vec i)$ and whenever
$o=E(\vec i)$ holds there is a set of objects including $\vec i$ and $o$ satisfying  $\Sigma_E$.
For 
the ``composition-free'' fragment -- in
which comprehensions $\bigcup$ can only be over input
variables -- this conversion can be done in $\ptime$. But it cannot be done
efficiently for general $\nrc$, under  complexity-theoretic hypotheses
 \cite{koch}.

We  also write entailments that use $\nrc$ expressions. For example, we write:
\[
\phi(x, \vec c \ldots) \modelssem x \in E(\vec c)
\] 
for $\phi$ $\deltazero$
and $E \in \nrc$. An entailment
with $\modelssem$  involving $\nrc$ expressions
 means that in every \emph{nested relation} satisfying $\phi$, $x$
is in the output of $E$ on $\vec c$. Note that the semantics
of $\nrc$ expressions is only defined on nested relations.

%% file: ack.tex
\myparagraph{Acknowledgements}
We thank Szymon Toru\'nczyk and Ehud Hrushovski for pointing us towards the model-theoretic approach to these results.
This paper extends abstracts appearing in POPL 2021 \cite{benediktpradicpopl} and PODS \cite{bethpods}
We thank in particular
 the POPL conference reviewers for their detailed feedback. 
Most of all we are deeply grateful to the reviewers of LMCS for their detailed
comments on the submission.

This research was funded in whole or in part by EPSRC grant EP/T022124/1.
Funded by the Deutsche Forschungsgemeinschaft (DFG, German Research
Foundation) -- Project-ID~457292495.
For the purpose of Open Access, the authors have applied a CC BY
public copyright license to any Author Accepted Manuscript (AAM)
version arising from this submission.

%% file: appendix.tex
\input{app-finite}

\input{app-fointerpbool}

\input{app-reducemonadicnrc}

\input{app-reducemonadicinterp}

\input{app-fointerpproof}

\input{app-completeness}

\input{app-focused}

\input{app-deltazerointerpolation} %
\input{app-derivations}

\input{app-mainthm}

\input{app-paramfo}

%% file: app-finite.tex
\section{Comparison to the situation with finite instances}
\label{app:finite}
Our result concerns a specification $\Sigma(\vec I, O \ldots)$
such that $\vec I$ implicitly defines $O$. This can be defined
``syntactically'' -- via the existence of a proof (e.g. in our own proof system).
Thus, the class of queries that we deal with could be called
the ``provably implicitly definable queries''.
The same class of queries can also be defined semantically, and this is how implicitly
defined queries are often presented. But in order to be equivalent to the
proof-theoretic version, we need the implicit definability of the object
$O$ over $\vec I$ to holds considering
all nested relations $\vec I, O \ldots$, not just finite ones.
Of course, the fact that when you phrase the property
 semantically requires referencing unrestricted instance  does not mean that
our results depend on the existence of infinite nested relations.

Discussions of finite vs. unrestricted instances appear in many other papers
(e.g. \cite{benediktpradicpopl}). And the results in this submission do not raise
any new issues with regard to the topic. But we discuss what
happens if we take the obvious analog of the semantic definition, but
using only finite instances.
Let us say that
a $\deltazero$ specification $\Sigma(\vec I, O \vec A)$ \emph{implicitly
defines $O$ in terms of $\vec I$ over finite instances} if for any
\emph{finite} nested relations $\vec I, O, \vec A, O', \vec A'$, if
$\Sigma(\vec I, O, \vec A) \wedge \Sigma(\vec I, O', \vec A')$ holds,
then $O=O'$.
If this holds, then $\Sigma$ defines a query, and we call
such a query \emph{finitely implicitly definable}.

This class of queries is reasonably well understood, and 
we summarize what is known about it:

\begin{itemize}
\item \emph{ Can finitely implicitly definable queries always be defined in $\nrc$?}
The answer is a resounding ``no'': one can implicitly define the powerset query over
finite nested relations. Bootstrapping this, one can define iterated
powersets, and show that the expressiveness of implicit definitions
is the same as queries in $\nrc$ enhanced with powerset -- a query
language with non-elementary complexity. Even in the setting of relational
queries, considering only finite instances leads to a query class
that is not known to be in PTIME \cite{kolaitisimpdef}.

\item \emph{Can we generate explicit definitions from specifications $\Sigma$, given
a proof that $\Sigma$   implicitly defines $O$ in terms
of $\vec I$ over finite instances?} It depends on what you mean by ``a proof'', but
in some sense there is no way to make sense of the question:
there is no recursively enumerable complete proof system for such definitions. This follows
from the fact that the set of finitely
implicitly definable queries is not computably enumerable.

\item \emph{Is sticking to specifications $\Sigma$ that are implicit definitions
over all inputs -- as we do in this work -- too strong?} Here the answer can
not be definitive. But we know of no evidence that this is too restrictive
 in practice.
Implicit specifications suffice to specify any $\nrc$ query. And
the  answer to the first question above says that if we modified
the definition in the obvious way to get a larger class,  we would allow specification of
 queries that do not admit efficient evaluation. The answer to the second
question above says that we do not have a witness to membership in this
larger class.
\end{itemize}

%% file: app-fointerpbool.tex
\section{Proof of Proposition \ref{prop:verify}: obtaining\texorpdfstring{\\}{} $\nrc$ expressions that verify \texorpdfstring{$\deltazero$}{Δ₀} formulas}
\label{app:fointerpbool}
  
Recall that in the body of the paper, we claimed the following statement, concerning
the equivalence of $\nrc$ expressions of Boolean type and $\deltazero$ formulas:

\medskip

There is a polynomial time
function  taking an extended $\deltazero$ formula $\phi(\vec x)$ and producing
 an $\nrc$ expression $\verify_\phi(\vec x)$, where the expression  takes as input
$\vec x$ and returns true if and only if $\phi$ holds.

\medskip

We refer to this as the ``Verification Proposition'' later on in these supplementary
materials.

\begin{proof}
First, one should note that every term in the logic can be translated to
a suitable $\nrc$ expression of the same sort. For example, a variable in
the logic corresponds to a variable in $\nrc$. 

We prove the proposition by induction over the formula $\phi(\vec x)$; we
only treat the case for half of the connectives, as the dual connectives
can be then be handled similarly using De Morgan rules and the fact that
boolean negation is definable in $\nrc$:
\begin{itemize}
\item If $\phi(\vec x)$ is an equality $t =_T t'$, then one can translate
  that to the $\nrc$ expression $\bigcup \{ \{\langle\rangle\} \mid z \in \{t\} \cap \{t'\}\}$.
\item If $\phi(\vec x)$ is a membership $t \in_T t'$, then it can be translated
  to the $\nrc$ expression $\bigcup \{ \{\langle\rangle\} \mid z \in \{\{t\}\} \cap \{t'\}\}$.
\item If $\phi(\vec x)$ is a disjunction $\phi_1(\vec x) \vee \phi_2(\vec x)$,
we take $\verify_\phi(\vec x) = \verify_{\phi_1}(\vec x) \cup \verify_{\phi_2}(\vec x)$.
\item If $\phi(\vec x)$ begins with a bounded existential quantification $\exists z \in y \; \psi(\vec x, y, z)$,
we simply set $\verify_\phi(\vec x, y) = \bigcup \{ \verify_{\psi(\vec x, y, z)} \mid z \in y \}$.
\qedhere
\end{itemize}
\end{proof}

Note that the converse (without the polynomial time bound) also holds; this will follow
from the more general result on moving from $\nrc$ to interpretations
that is proven later in the supplementary materials.

%% file: app-reducemonadicnrc.tex
\section{First part of Proposition \ref{prop:reducemonadic}:\texorpdfstring{\\}{} Reduction to monadic schemas for $\nrc$}
\label{app:reducemonadicnrc}

In the body of the paper
we mentioned that it is possible to reduce questions about
definability within $\nrc$ to the case of monadic schemas.
We now give the details of this reduction.

Recall that  \emph{monadic type} is a type built only using
the atomic type $\ursort$ and the type constructor
$\sett$. Monadic types are in one-to-one correspondence with
natural numbers by setting
$\ursort_0 \eqdef \ursort$ and $\ursort_{n+1} \eqdef \sett(\ursort_n)$.
A monadic type is thus a $\ursort_n$ for some $n \in \bbN$.
A nested relational schema is monadic if it contains only monadic types,
and a $\deltazero$ formula is said to be monadic if it all of its variables have monadic types.

We start with a version of the reduction only for $\nrc$ expressions:

\begin{prop} \label{prop:reducemonadic-nrc}
For any nested relational schema $\aschema$, there is
a monadic nested relational schema $\aschema'$,
 an injection $\convert$ from instances of $\aschema$ to instances
of $\aschema'$ that is definable in $\nrc$, and an $\nrc[\nrcget]$ expression
$\convert^{-1}$
such that $\convert^{-1} \circ \convert$ is the identity transformation from
$\aschema \to \aschema$.

Furthermore, there is a $\deltazero$ formula $\image_\convert$ from $\aschema'$ to $\bool$
such that $\image_\convert(i')$ holds if and only if $i' = \convert(i)$
for some instance $i$ of $\aschema$.
\end{prop}

To prove this we  give an encoding of
general nested relational schemas into monadic nested relational schemas that
will allow us to reduce the equivalence between $\nrc$ expression,
interpretations, and implicit definitions to the case where
input and outputs are monadic.

Note that it  will turn out to be crucial to check that this encoding may be
defined \emph{either} through $\nrc$ expressions or interpretations, but
in this subsection we  will give the definitions in terms of $\nrc$ expressions.

The first step toward defining these encodings is actually to
emulate in a sound way the cartesian product structure for
types $\ursort_n$. Here ``sound'' means that we should give
terms for pairing and projections that satisfy the usual
equations associated with cartesian product structure.

\begin{prop}
\label{prop:nrc-monadicproduct}
For every $n_1, n_2 \in \bbN$, there are $\nrc$ expressions
$\widehat{\pair}(x,y) : \ursort_{n_1}, \ursort_{n_2} \to \ursort_{\max(n_1,n_2) + 2}$
and $\nrcwget$ expressions $\widehat\pi_i(x) : \ursort_{\max(n_1,n_2) + 2} \to \ursort_{n_i}$ for $i \in \{1,2\}$
such that the following equations hold
$$
\widehat{\pi}_1\left(\widehat{\pair}(a_1,a_2)\right) ~ = ~ a_1
\qquad \qquad
\widehat{\pi}_2\left(\widehat{\pair}(a_1,a_2)\right) ~ = ~ a_2
$$
Furthermore, there is a $\deltazero$ formula $\image_{\widehat{\pair}}(x)$ such that
$\image_{\widehat{\pair}}(a)$ holds if and only if there exists $a_1, a_2$ such
that $\widehat\pair(a_1,a_2) = a$. In such a case, the following also holds
$$\widehat{\pair}(\widehat\pi_1(a),\widehat\pi_2(a)) ~ = ~ a$$
\end{prop}

\begin{proof}
We adapt the Kuratowski encoding of pairs $(a,b) \mapsto \{ \{a\},\{a,b\}\}$.
The notable thing here is that, for this encoding
to make sense in the typed monadic setting, the types of $a$ and $b$ need to be the same.
This will not be an issue because we have $\nrc$-definable embeddings 
$$\uparrow_n^m : \ursort_n \to \ursort_m$$
for $n \le m$ defined as the $m - n$-fold composition of the singleton transformation $x \mapsto \{x\}$.
This will be sufficient to define the analogues of pairing for monadic types and thus
to define $\convert_T$ by induction over $T$.
On the other hand, $\convert^{-1}_T$ will require a suitable encoding of projections.
This means that to decode an encoding of a pair, we need to make use of a transformation inverse to
the  singleton construct $\uparrow$. But we have this  thanks to the $\nrcget$ construct. We let
$$\downarrow_n^m : \ursort_m \to \ursort_n$$ the transformation inverse to
 $\uparrow_n^m$, defined as the $m - n$-fold composition of $\nrcget$.

Firstly, we define the family of transformations $\widehat{\pair}_{n_1,n_2}(x_1,x_2)$,
where $x_i$ is an input of type $\ursort_{n_i}$ for $i \in \{1,2\}$ and
the output is of type $\ursort_{\max(n_1,n_2)+2}$, as follows
$$\widehat{\pair}_{n_1,n_2}(x_1,x_2) \; \eqdef \;
\{ \{ \uparrow x_1 \}, \{\uparrow x_1, \uparrow x_2\}\}$$

The associated projections
$\widehat\pi_i^{n_1,n_2}(x)$ where $x$ has type $\ursort_{\max(n_1,n_2)+2}$
and the output is of type $\ursort_{n_i}$
are a bit more challenging to construct.
The basic idea is that there is
first a case distinction to be made for encodings $\widehat{\pair}_{n_1,n_2}(x_1,x_2)$:
depending on whether  $\uparrow x_1 = \uparrow x_2$
or not. This can be actually tested by an $\nrc$ expression.
Once this case distinction is made, one may informally
compute the projections as follows:
\begin{itemize}
\item if $\uparrow x_1 = \uparrow x_2$, both projections can be computed
as a suitable downcasting $\downarrow$ (the depth of the downcasting is determined by the output
type, which is not necessarily the same for both projections).
\item otherwise, one needs to single out the singleton $\{\uparrow x_1\}$ and the
two-element set $\{ \uparrow x_1, \uparrow x_2 \}$ in $\nrc$. Then, one
may compute the first projection by downcasting the singleton, and the second projection
by first computing $\{\uparrow x_2\}$ as a set difference and then downcasting with $\downarrow$.
\end{itemize}

We now give the formal encoding for projections, making a similar case distinction.  
To this end, we first define a generic $\nrc$ expression
$$\allpairs_T(x) : \sett(T) \to \sett(T \times T)$$
computing all the pairs of distinct elements of its input $x$
$$\allpairs_T(x) = \bigcup \{ \bigcup \{ \{(y, z)\} \mid y \in x \setminus \{z\}\} \mid z \in x \}$$
Note in particular that $\allpairs(i) = \emptyset$ if and only if $i$ is a singleton or the empty set.
The projections can thus be defined as
$$
\begin{array}{l!~c!~l}
\widehat{\pi}_1(x) &\eqdef&
\case\left(\allpairs(x) = \emptyset,~ \downarrow x,~ \downarrow \bigcup\{ \pi_1(z) \cap \pi_2(z) \mid z \in \allpairs(x) \}\right) \\
\widehat{\pi}_2(x) &\eqdef& \case\left(\allpairs(x) = \emptyset, ~
\downarrow x, ~ \downarrow (x \setminus \uparrow \widehat{\pi}_1(x)))\right) \\
\end{array}
$$
These definitions crucially ensure that,
for every object $a_i$ with $i \in \{1,2\}$, we  have
$$\widehat{\pi}_i\left(\widehat{\pair}(a_1,a_2)\right) ~ = ~ a_i$$

Now all remains to be done is to define $\image_{\widehat\pair}$.
Before that, it is helpful to define a formula $\image_{\uparrow_n^m}(x)$
which holds if and only if $x$ is in the image of $\image_{\uparrow_n^m}$.

As a preliminary step, define generic $\deltazero$ formulas $\issing(x)$ and $\istwo(x)$
taking an object of type $\sett(T)$ and returning a Boolean indicating
whether the object is a singleton or a two-element set.
Defining $\image_{\uparrow_n^m}$ is straightforward using $\issing$ and Boolean connectives.
Then $\image_{\widehat{\pair}_{n,n}}(x)$ can be defined as follows for each $n \in \bbN$
$$
\begin{array}{lcl}
\image_{\widehat{\pair}_{n,n}}(x) &\eqdef&
\left(\issing(x) \wedge \image_{\widehat{\pair}_{n,n}}^{\issing}(x)\right) \vee 
\left(\istwo(x) \wedge \image_{\widehat{\pair}_{n,n}}^{\issing}(x)\right)
\\ 
\image_{\widehat{\pair}_{n,n}}^\issing(x) &\eqdef& \exists z \in x~~ \issing(z) \\
\image_{\widehat{\pair}_{n,n}}^\istwo(x) &\eqdef& \exists z \; z' \in x~~ (\istwo(z) \wedge \issing(z') \wedge \forall y \in z~~ y \in z')
\end{array}
$$
Then, the more general $\image_{\widehat{\pair}_{n_1,n_2}}$ can be defined using
$\image_{\uparrow_{n_i}^{m}}$ where $m = \max(n_1,n_2)$.
$$
\begin{array}{lll}
\image_{\widehat{\pair}_{n_1,n_2}}(x)
&\eqdef&
\image_{\widehat{\pair}_{m,m}}(x) ~\cap~
\image_{\uparrow_{n_1}^m}(\widehat\pi_1(x)) ~\cap~
\image_{\uparrow_{n_2}^m}(\widehat\pi_2(x))
\end{array}
$$

One can then easily check that $\image_{\widehat{\pair}}$ does have the advertised property:
if $\image_{\widehat{\pair}}(a)$ holds for some
object $a$, then there are $a_1$ and $a_2$ such that $\widehat\pair(a_1,a_2) = a$ and we have
\begin{align*}&\widehat{\pair}(\widehat\pi_1(a),\widehat\pi_2(a)) ~ = ~ a\qedhere\end{align*}
\end{proof}

We are now ready to give the proof of the proposition  given
at the beginning of this subsection.

\begin{proof}

$\convert_T$, $\convert^{-1}_T$ and $\image_{\convert_T}$ are defined
by induction over $T$. Beforehand, define the
map $d$ taking a type $T$ to a natural number $d(T)$ so
that $\convert$ maps instances of type $T$ to monadic types $\ursort_{d(T)}$.
$$
\begin{array}{l!~c!~ l !\qquad l!~c!~l}
d(\ursort) &=& 0 &
d(\sett(T)) &=& 1 + d(T) \\
d(T_1 \times T_2) &=& 2 + \max(d(T_1),d(T_2))
&
d(\unit) &=& 0 \\
\end{array}
$$

$\convert_T$, $\convert^{-1}_T$ and $\image_{\convert_T}$ are then defined by the following rules,
where we write $\map\left(z \mapsto E\right)(x)$ for the $\nrc$ expression $\bigcup\{\{E\} \mid z \in x \}$. 

$$
\begin{array}{lcl}
\convert_\ursort(x) &\eqdef& x \\
\convert_{\sett(T)}(x) &\eqdef& \map\left(z \mapsto \convert_{T}(z)\right)(x) \\
\convert_\unit(x) &\eqdef& c_0 \\
\convert_{T_1 \times T_2}(x) &\eqdef& \widehat{\pair}(\convert_{T_1}(\pi_1(x)), \convert_{T_2}(\pi_2(x))) \\
\\
\convert^{-1}_\ursort(x) &\eqdef& x \\
\convert^{-1}_{\sett(T)}(x) &\eqdef& \map\left(z \mapsto \convert^{-1}_{T}(z)\right)(x) \\
\convert_\unit(x) &\eqdef& \unitobj \\
\convert^{-1}_{T_1 \times T_2}(x) &\eqdef& \left\< \convert^{-1}_{T_1}(\widehat\pi_1(x)), \convert_{T_2}^{-1}(\widehat\pi_2(x)) \right\> \\
\\
\image_{\convert_\ursort}(x) &\eqdef& \true \\
\image_{\convert_{\sett(T)}}(x) &\eqdef& \forall z \in x ~~ \image_{\convert_{T}}(z) \\
\image_{\convert_{T_1 \times T_2}}(x) &\eqdef&
\image_{\pair_{d(T_1),d(T_2)}}(x) \wedge
\image_{\convert_{T_1}}(\widehat\pi_1(x)) \wedge \image_{\convert_{T_2}}(\widehat\pi_2(x))
\end{array}
$$
It is easy to check, by induction over $T$, that for every object $a$ of type $T$
$$\convert^{-1}(\convert(a)) = a$$
and that for every object $b$ of type $\ursort_{d(T)}$, if $\image_{\convert_T}(b) = \true$, then it lies in the image
of $\convert_T$ and $\convert(\convert^{-1}(b)) = b$.
\end{proof}

%% file: app-reducemonadicinterp.tex
\section{Second part of Proposition \ref{prop:reducemonadic}:\texorpdfstring{\\}{} Monadic reduction for interpretations}
\label{app:reducemonadicinterp}

We have  seen so far that it is possible to reduce questions about
definability within $\nrc$ to the case of monadic schema.
Now we turn to the analogous statement for interpretations, given by
the following proposition:

\begin{prop}
\label{prop:reducemonadic-interp}
For any object schema $\aschema$, there is
 a monadic nested relational schema $\aschema'$,
 a $\deltazero$ interpretation $\interp_\convert$ from instances of $\aschema$ to instances
of $\aschema'$, and another interpretation
$\interp_{\convert^{-1}}$ from instances of $\aschema$ to instances of $\aschema'$ compatible
with $\convert$ and $\convert^{-1}$ as defined in Proposition~\ref{prop:reducemonadic-interp} in the following sense:
for every instance $I$ of $\aschema$ and
for every instance $J$ of $\aschema'$ in the codomain of $\convert$, we have
$$\convert^{-1}(J) = \collapse(\interp_{\convert^{-1}}(J)) \qquad \qquad \convert(I) = \collapse(\interp_\convert(I))$$
\end{prop}

Before proving Proposition~\ref{prop:reducemonadic-interp}, it is helpful to check that
a number of basic $\nrc$ connectives may be defined at the level of interpretations.
To do so, we first present a technical result for more general interpretations.

\begin{prop}
\label{prop:mostowski-collapse-interp}
For any sort $T$, there is an interpretation
of $\aschema_T$ into $\aschema_T$ taking a models $M$ whose every sort is non-empty and $\booltype$
has at least two elements to a model $M$ of $\oneobjth(T)$.
Furthermore, we have that $M'$ is (up to isomorphism) the largest quotient of $M'$ satisfying $\oneobjth(T)$.
\end{prop}
\begin{proof}
This interpretation corresponds to a quotient of the input, that is definable at
every sort
\[
\begin{array}{lcl!\quad lcl}
\varphi_\equiv^{\sett(T)}(x,y) &=&
\forall z~(z \in x \equi z \in y)
& \varphi_\equiv^{T_1 \times T_2}(x,y) &=& \pi_1(x) = \pi_1(y) \wedge \pi_2(x) = \pi_2(y) \\
\varphi_\equiv^{\unit}(x,y) &=& \top
& \varphi_\equiv^{\ursort}(x,y) &=& x =_\ursort y \hfill\qedhere
\end{array}
\]
\end{proof}

\begin{prop}
\label{prop:interp-sing-cup-map-pair}
The following $\deltazero$-interpretations are definable:
\begin{itemize}
\item $\interp_\sing$ defining the transformation $x \mapsto \{x\}$.
\item $\interp_\cup$ defining the transformation $x,y \mapsto x \cup y$.
\end{itemize}
Furthermore, assuming that $\interp$ is a $\deltazero$-interpretation
defining a transformation $E$ and $\interp'$ is a $\deltazero$-interpretation
defining a transformation $R$, the following $\deltazero$-interpretations
are also definable:
\begin{itemize}
\item $\map(\interp)$ defining the transformation $x \mapsto \{ E(y) \mid x \in y\}$.
\item $\<\interp,\interp'\>$ defining the transformation $x,y \mapsto (E(x),F(y))$.
\end{itemize}
\end{prop}
\begin{proof}\hfill
\begin{itemize}
\item
For the singleton construction $\{e\}$ with $e$ of type $T$, we take the interpretation
$\interp_e$ for $e$, where $e$ itself is interpreted by a constant $c$ and we add an extra
level represented by an input constant $c'$. Then $\varphi_\domainof^{\sett(T)}(x)$ is set to $y = c'$ and
$\phi_\in^T(x,y)$ to $x = c \wedge y = c'$.
\item The empty set $\{\}$ at type $\sett(T)$ is given by the trivial interpretation
where $\varphi_\domainof^{\sett(T)}(x)$ is set to $x = c$ for some constant $c$
and $\varphi_\domainof^{T'}$ is set to false for $T'$ a component type of $T$, as well as all the $\varphi_\in^T$.
\item For the binary union $\cup : \sett(T), \sett(T) \to \sett(T)$, the interpretation is easy: $T$
is interpreted as itself. The difference between input and output
is that $\sett(T) \times \sett(T)$ is not an output sort and that $\sett(T)$ is interpreted as a single
element, the constant $\unitobj$ of $\unit$.
$$
\def\arraystretch{1.1}
\begin{array}{lcl}
\phi_\domainof^{\sett(T)}(x) &~~ \eqdef ~~& x = \unitobj \\
\phi_\in^T(z,x) &~~ \eqdef ~~& z \in \pi_1(\inobj) \vee z \in \pi_2(\inobj)
\end{array}
$$
\item We now discuss the $\map$ operator. Assume that we have an interpretation
$\interp$ defining a transformation $S \to T$ that we want to lift to an interpretation $\map(\interp) : \sett(S) \to \sett(T)$.
Let us write $\psi_\domainof^{T'}$, $\psi_\in^{T'}$ and $\psi_\equiv^{T'}$ for the formulas making
up $\interp$ and reserve the $\phi$ formulas for $\map(\interp)$.
At the level of sort, let us write $\interpsort^\interp$ and $\interpsort^{\map(\interp)}$ to distinguish the two.

For every $T' \subtype T$ such that $T'$ is not a cartesian product or a component type of $\booltype$,
we set $\interpsort^{\map(\interp)}(T') = S, \interpsort^{\interp}$. This means that objects of
sort $T'$ are interpreted as in $\interp$ with an additional tag of sort $S$.
We interpret the output object $\sett(T)$ as a singleton by setting $\interpsort^{\map(\interp)}(\sett(T)) = \unit$.

Assuming that $T \neq \ursort, \unit$, $\map(\interp)$ is determined by setting the following 
$$
\def\arraystretch{1.1}
\begin{array}{lcl}
\phi_\domainof^\ursort(a) &~~\eqdef~~& \exists s \in \inobj~~ \psi_\domainof(a)[s/\inobj] \\
\phi_\in^\ursort(a, s, \vec x) &~~\eqdef~~& \psi_\in^\ursort(a,\vec x)[s/\inobj] \\
\\
\phi_\domainof^{T'}(s, \vec x) &~~\eqdef~~& \psi_\domainof^{T'}(x)[s'/\inobj] \\
\phi_\in^{T'}(s, \vec x, s', \vec y) &~~\eqdef~~& \exists \vec{x'}~~ \psi_\in^{T'}(\vec{x'}, \vec y)[s'/\inobj] \wedge \phi^{T'}_\equiv(s,\vec x,s',\vec x') \\
\\
\phi_\domainof^{T}(s, \vec x) &~~\eqdef~~& s \in \inobj \\
\phi_\in^{T}(s, \vec x) &~~\eqdef~~& \phi_\domainof^T(s, \vec x) \\
\end{array}
$$
where $[x/\inobj]$ means that we replace occurrences of the constant
$\inobj$ by the variable $x$ and sorts $T'$ and $T' \times T''$ are component types of
$T$. Note that this definition is technically by induction over the type, as we use
$\phi^{T'}_\equiv$ to define $\phi_\in^{T'}$.
In case $T$ is $\ursort$ or $\unit$, the last two formulas $\phi_\domainof^{T}$ and $\phi_\in^{T}$ need to change.
If $T = \unit$, then we set
$$\phi_\domainof^{\unit}(c_0) \eqdef \phi_\in^\unit(c_0,c_0) \eqdef \exists s \in \inobj~~ \top$$
and if $T = \ursort$, we set
$$\phi_\domainof^{\ursort}(a) \eqdef \phi_\in^\ursort(a) \eqdef  \exists s \in \inobj~~ \psi_\domainof(a)[s/\inobj]$$
\item Finally we need to discuss the pairing of two interpretation-definable
transformations $\< \interp_1, \interp_2 \> : S \to T_1 \times T_2$.
Similarly as for map we reserve
$\phi_\domainof^T$, $\phi_\in^T$ and $\phi_\equiv^T$
formulas for the interpretation $\< \interp_1, \interp_2 \>$.
We write
$\psi_\domainof^T$, $\psi_\in^T$ and $\psi_\equiv^T$
for components of $\interp$ and
$\theta_\domainof^T$, $\theta_\in^T$ and $\theta_\equiv^T$
for components of $\interp'$.

Now, the basic idea is to interpret output sorts of $\< \interp_1, \interp_2 \>$ as tagged unions
of elements that either come from $\interp_1$ or $\interp_2$.
Here, we exploit the assumption that $\aschema_T$ contains the sort $\booltype$.
and that every sort is non-empty to interpret the tag of the union.
The union itself is then encoded as a concatenation of a tuple representing
a would-be element form $\interp_1$ with another tuple representing a would-be element from $\interp_2$,
the correct component being selected with the tag. For that second trick to work, note that
we exploit the fact that every sort has a non-empty denotation in the input structure.
Concretely, for every $T$ component type of either $T_1$ or $T_2$, we thus set
$$
\def\arraystretch{1.1}
\begin{array}{lcll}
\interpsort^{\<\interp_1,\interp_2\>}(T) &~~ \eqdef ~~& \booltype, \interpsort^{\interp_1}(T), \interpsort^{\interp_2}(T)
\\
\phi_\domainof^T(u, \vec x, \vec y) &~~ \eqdef ~~& (u = \booltt \wedge \psi_\domainof^T(\vec x)) \vee (u \neq \booltt \wedge \theta_\domainof^T(\vec y)) \\
\phi_\in^T(u, \vec x, \vec y, u', \vec x', \vec {y'}) &~~ \eqdef ~~& (u = u' = \booltt \wedge \psi_\in^T(\vec x, \vec x')) \vee (u = u' = \boolff \wedge \theta_\in^T(\vec {y}, \vec {y'})) \\
\phi_\equiv^T(u, \vec x, \vec y, u', \vec x', \vec {y'}) &~~ \eqdef ~~& (u = u' = \booltt \wedge \psi_\equiv^T(\vec x, \vec x')) \vee (u = u' = \boolff \wedge \theta_\equiv^T(\vec {y}, \vec {y'})) \\
\end{array}
$$
Note that this interpretation does not quite correspond to a pairing
because it is not a complex object interpretation: the interpretation of common subobjects of
$T_1$ and $T_2$ are not necessarily identified, so the output is not necessarily a model of $\oneobjth$.
This is fixed by postcomposing with
the interpretation of Proposition~\ref{prop:mostowski-collapse-interp} to obtain $\<I_1,I_2\>$.\qedhere
\end{itemize}
\end{proof}

\begin{proof}[Proof of Proposition~\ref{prop:reducemonadic-interp}]

Similarly as with Proposition~\ref{prop:reducemonadic-nrc}, we define auxiliary
interpretations $\interp_{\uparrow}$, $\interp_{\downarrow}$
$\interp_{\widehat\pair}$, $\interp_{\hat\pi_1}$ and $\interp_{\hat\pi_2}$
mimicking the relevant constructs of Proposition~\ref{prop:reducemonadic-nrc}.
Then we will dispense with giving the recursive definitions of $\interp_{\convert_T}$
and $\interp_{\convert^{-1}_T}$, as they will be obvious from inspecting the 
cases given in the proof of Proposition~\ref{prop:reducemonadic-nrc} and
replicating them using Proposition~\ref{prop:interp-sing-cup-map-pair} together with closure under composition
of interpretations. %

$\interp_{\uparrow}$, $\interp_{\downarrow}$
and $\interp_\pair$ are easy to define through Proposition~\ref{prop:interp-sing-cup-map-pair}, so
we focus on the projections $\interp_{\hat\pi^{n_1,n_2}_1}$ and $\interp_{\hat\pi^{n_1,n_2}_2}$,
defining transformations from $\ursort_m$ to $\ursort_{n_i}$ for $i \in \{1,2\}$ where $m \eqdef \max(n_1,n_2)$.
Note that in both cases, the output sort is part of the input sorts. Thus an output sort will be interpreted
by itself in the input, and the formulas will be trivial for every sort lying strictly below the output sort:
we take 
$$\phi_{\in_{\ursort_{k}}}(x,y) ~\eqdef~ x \in y \wedge \phi_\domainof^{\ursort_{k+1}}(y) \qquad
  \phi_{\equiv}^{\ursort_k}(x,y) ~\eqdef~ x = y \qquad
  \phi_\domainof^{\ursort_k}(x) ~\eqdef~ \top
$$
for every $k < n_i$ ($i$ according to which projection we are defining).
The only remaining important data that we need to provide are the 
formulas $\phi_\domainof^{\ursort_{n_i}}$,
which, of course, differ for both projections. We provide those below, calling $o_{in}$ the designated input
object. For both cases, we use an auxiliary predicate  $x \in^k y$ standing for $\exists y_1 \in y \ldots \exists y_{k-1}
\in y_{k-2}~~ x \in y_{k-1}$ for $k > 1$; for $k=0,1$, we take $x \in^1 y$ to be $x \in y$ and $x \in^0 y$ for $x = y$.
\begin{itemize}
\item For $\interp_{\hat\pi^{n_1,n_2}_1}$, we set
$$\phi_\domainof^{\ursort_{n_1}}(x) ~ \eqdef ~ \forall z \in o_{in}~\exists z' \in z ~~~ x \in^{m-n_1} z'$$
The basic idea is that the outermost $\forall\exists$ ensures that we compute the intersection of the two sets
contained in the encoding of the pair. 
\item For $\interp_{\hat\pi^{n_1,n_2}_2}$, first note that there are obvious $\deltazero$-predicates $\issing(x)$ and
$\istwo(x)$ classifying singletons and two element sets. This allows us to write the following $\deltazero$ formula
$$
\small
\phi_\domainof^{\ursort_{n_1}}(x) ~ \eqdef ~
\bigvee \left[ \begin{array}{l}
\issing(x) \wedge \forall z \in o_{in}~\exists z' \in z~~~ x \in^{m-n_2} z' \\
\istwo(x) \wedge \exists z \; z' \in o_{in}~ \exists y \in z'~ (y \notin z \wedge x \in^{m-n_2} z')
\end{array}\right.
$$
\end{itemize}
It is then easy to check that, regarded as transformations, those interpretation also implement the
projections for Kuratowski pairs.
\end{proof}

%% file: app-fointerpproof.tex
\section{Proof of Theorem \ref{thm:coddnrc}: converting\texorpdfstring{\\}{} between $\nrcwget$ expressions and interpretations}
\label{app:fointerpproof}

In the body of the paper we claimed that $\nrcwget$ expressions have
the same expressiveness as interpretations.
One direction of this expressive equivalence is given in the
following lemma:

\begin{lem} \label{lem:forward} There is an $\exptime$ computable function
taking an $\nrcwget$ expression $E$ to an equivalent FO interpretation $\interp_E$.
\end{lem}

As we mentioned in the body of the paper,  
very similar results occur in the prior literature, going
as far back as \cite{simulation}. 

\begin{proof}%
We can assume that the input and output schemas are monadic, using the
reductions to monadic schemas given previously.
Indeed, if we solve the problem
for expressions where input and output
schemas are monadic, we can reduce the problem of finding an interpretation for an
arbitrary $\nrcwget$ expression $E(x)$ as follows:
construct a $\deltazero$ interpretation $\interp$
for the  expression $\convert(E(\convert^{-1}(x)))$ --
where $\convert$ and $\convert^{-1}$
are taken as in Proposition~\ref{prop:reducemonadic-nrc} --
and then, using closure under composition of interpretations (see e.g.~\cite{xqueryinterp}),
one can then leverage Proposition~\ref{prop:reducemonadic-interp}
to produce the composition of $\interp_{\convert^{-1}}$, $\interp$ and $\interp_{\convert}$
which is equivalent to the original expression $E$.

The argument proceeds by induction on the structure of $E : \vec T \to S$ in
$\nrc$. Some atomic operators were treated in the prior section, like singleton
$\cup$, tupling, and projections.  Using
closure of interpretations under composition, we are thus able to translate
compositions of those operators.
We are only left with a few cases.
\begin{itemize}
\item
For the set difference, since interpretations are closed under
composition, it suffices to prove that we can code the transformation
$$(x, y) ~~\mapsto~~ x \setminus y$$
at every sort $\sett(\ursort_n)$. Each sort gets interpreted by itself.
We thus set
$$
\def\arraystretch{1.1}
\begin{array}{lcll}
\phi_\domainof^{\ursort_n}(z) &~~\eqdef~~&
z \in \pi_1(\inobj) \wedge z \notin \pi_1(\inobj) \\
\phi_\domainof^{\ursort_k}(z) &~~ \eqdef ~~&
\exists z' ~~(\phi_\domainof^{\ursort_n} \wedge z \in^{n-k} z') \\
\phi_\in^{\ursort_k}(z, z') &~~ \eqdef ~~& z \in z' \wedge \phi^{\ursort_k}_\domainof(z) \wedge \phi^{\ursort_{k+1}}_\domainof(z') \\
\end{array}
$$
\item To get $\nrcwget$ expressions, it suffices 
to create a
$\deltazero$ interpretation corresponding to  $\nrcget$
which follows
$$\phi_\domainof^\ursort(a) ~~ \eqdef~~ (\exists !~ z \in \inobj ~~ z = a) \vee (\neg (\exists !~ z \in \inobj) \wedge a = c_0)$$
\item For the binding operator
$$\bigcup \{ E_1 \mid x \in E_2 \}$$
we exploit the classical decomposition
$$\bigcup ~~\circ~~ \map(E_1) ~~\circ~~ E_2$$
As interpretations are closed under composition and the mapping operations
was handled in Proposition~\ref{prop:interp-sing-cup-map-pair},
it suffices to give an interpretation for the expression $\bigcup : \sett(\sett(T)) \to \sett(T)$
for every sort $T$. This is straightforward: each sort gets interpreted as itself,
except for $\sett(T)$ itself which gets interpreted as the singleton $\{c_0\}$.
The only non-trivial case is  the following
\begin{align*}
& \phi_\in^{T}(x,y) ~~\eqdef~~ \phi_\domainof^T ~~\eqdef~~ \exists y' \in \inobj ~~ x \in y'
\qedhere
\end{align*}
\end{itemize}
\end{proof}

\myparagraph{From interpretations to $\nrcwget$ expressions}
The other direction of the expressive equivalence is provided by the following
lemma:

\begin{lem} \label{lem:back} There is a polynomial time function
taking a $\deltazero$ interpretation to an equivalent $\nrcwget$ expression.
\end{lem}

This direction is not used directly in the conversion
from implicitly definable transformations to $\nrcwget$, 
but it is of interest
in showing that $\nrcwget$ and $\deltazero$ interpretations are equally expressive.

\begin{proof}[Proof of Lemma \ref{lem:back}]
Using the reductions to monadic schemas, it suffices to show this
for transformations that have monadic input schemas as input and output.

Fix a $\deltazero$ interpretation $\interp$ with input $\ursort_n$ and output $\ursort_m$.

Before we proceed, first note that for every $d \le m$,
there is an $\nrc$ expression
$$E_d : \ursort_n \to \sett(\ursort_d)$$
collecting all of the subobjects of its input of sort $\ursort_d$.
It is formally defined by the induction over $n - d$.
$$E_m(x) ~~ \eqdef~~ \{x\} \qquad \qquad E_d(x) = \bigcup E_{d-1}(x)$$
Write $E_{d_1, \ldots, d_k}(x)$ for $\< E_{d_1}, \ldots, E_{d_k} \>(x)$
for every tuple of integers $d_1 \ldots d_k$.

For $d \le m$, let $d_1, \ldots, d_k$ be the tuple such that the output
sort $\ursort_d$ is interpreted by the list of input sorts
$\ursort_{d_1}, \ldots, \ursort_{d_k}$.
By induction over $d$, we build  $\nrc$ expressions
$$E_d : \ursort_m, \ursort_{d_1}, \ldots, \ursort_{d_k} \to \ursort_d$$
such that, provided that $\phi_\domainof^{\ursort_d}(\vec a)$ and $\phi_\domainof^{\ursort_{d+1}}(\vec b)$ hold, we have
$$\phi_\in^{\ursort_d}(\vec a, \vec b) \qquad \qquad  \text{if and only if} \qquad \qquad  E_d(\vec a) \in E_{d+1}(\vec b)$$

For $E_0 : \ursort_m, \ursort \to \ursort$, we simply take the second projection.
Now assume that $E_d$ is defined and that we are
looking to define $E_{d+1}$.
We want to set
\begin{align*}
E_{d+1}(x_{in},\vec{y}) \eqdef \{E_d(x_{in}, \vec x) \mid \vec{x} \in E_{d_1, \ldots, d_k}(x_{in}, \vec y) \wedge \verify_{\phi_{\in}^i}(x_{in}, \vec x, y_{in}, \vec y) \}
\end{align*}
which is $\nrc$-definable as follows 
$$\bigcup \left\{
\case\left(\verify_{\phi_{\in}^i}(x_{in}, \vec x, y_{in}, \vec y), ~ \{E_d(x_{in})\}, ~ \{\}\right)
\mid \vec x \in E_{d_1,\ldots,d_k}(x_{in}) \right\}
$$
where $\verify$ is given as in the Verification Proposition proven earlier in the 
supplementary materials
and $\{E(\vec x, \vec y) \mid \vec x \in E'(\vec y)\}$ is a
notation for $\bigcup \{ \ldots \bigcup \{ E(\vec x, \vec y) \mid x_1 \in \pi_1(E'(\vec y)) \} \ldots \mid x_k \in \pi_k(E'(\vec y))\}$.
It is easy to check that the inductive invariant holds.

Now, consider the  transformation $E_m : \ursort_n, \ursort_{m_1}, \ldots, \ursort_{m_k} \to \ursort_m$.
The transformation
$$R ~~ \eqdef ~~\{ E_m(x_{in}, \vec y) \mid \vec y \in E_{m_1, \ldots, m_k}(x_{in}) \wedge \phi_\domainof^{\ursort_m}(\vec y)\}$$
is also $\nrc$-definable using $\verify$.
Since the inductive invariant holds at level $m$, $R$ returns the singleton containing the output of $\interp$. Therefore $\nrcwget(R) : \ursort_n \to \ursort_m$ is the desired $\nrcwget$ 
expression equivalent to the interpretation $\interp$.
\end{proof}

Note that  the argument can be easily modified to
 produce an $\nrcwget$ expression that is \emph{composition-free}:  in
union expressions $\bigcup\{E_1 \mid x \in E_2\}$, the range $E_2$ of
the variable $x$ is always another variable.
In composition-free expressions, we allow as a native  construct $\case(B,E_1,E_2)$ where $B$ is a Boolean combination
of atomic transformations with Boolean output, since we can not use composition to derive the conditional from the other operations.

Thus every $\nrcwget$ expression can be converted to one that is composition-free, and similarly
for $\nrcwget$.  The analogous
statements have been observed before for related languages like XQuery \cite{xqueryinterp}.

%% file: app-completeness.tex
\section{Completeness of proof systems} \label{sec:app-complete}
\label{app:completeness}

In the body of the paper we mentioned that the completeness
of the high-level proof system  of Figure \ref{fig:dzcalc-2sided}
can be argued using a standard method -- see, for example \cite{smullyanbook}.
We outline this for the reader who is unfamiliar with these arguments. In Appendix \ref{app:focused} we will prove
that the lower level system is complete, by reducing it to completeness of the higher level system.

We turn to the outline of the general method.
One has a sequent $\Theta;\; \Gamma \seq \Delta$ that is not provable.
We want to construct a countermodel: one that satisfies all the formulas in $\Theta$ and $\Gamma$
but none of the formulas in $\Delta$.
We construct a  tree with $\Theta;\; \Gamma \seq \Delta$ at the root by iteratively
applying applicable inference rules in reverse: in ``proof search mode'', generating
subgoals from goals. We apply the  rules whenever possible in a given order,
enforcing some fairness constraints: a rule that is active must be 
eventually applied along an infinite search,
and if a choice of terms must be made (as with the $\exists$-R rule),
all possible choices of terms are eventually made in an application of the rule.
For example, if we have a disjunction $\rho_1 \vee \rho_2$ on the right, we may
immediately ``apply $\vee$-R'': we generate a subgoal where on the
right hand side we add $\rho_1, \rho_2$.
Finite branches leading to
a sequent that does not match the conclusion of any rule or
axiom are artificially extended to infinite branches by
repeating the topmost sequent.

By assumption, this process does not produce a proof,
 and thus we have
 an infinite  branch $b$ of the tree.
 We create a model $M_b$ whose elements are the closure under tupling of the
 variables of types, other than product types, that appear on the branch.
 The model interprets these ``non-tuple'' variables by themselves, while
 variables with a product type are interpreted by -- possibly nested -- tuples
 of ``non-tuple'' variables.
 Specifically, $M_b$
 interprets terms $t$ by these domain elements as follows.
\begin{itemize}
\item If $t$ is a variable of type $\ursort$ or of a set type,
  then $t^{M_b}$ is $t$. 
\item If $t$ is a variable of
  type $\unit$ or is the term $\la\ra$, then $t^{M_b}$ is $\la\ra$.  
\item If $t$ is a variable of a product type, then $t^{M_b}$ is $\la
  x_1^{M_b}, x_2^{M_b}\ra$, where $t$ in the role of $x$ and $x_1,x_2$ are the
  respective parameters of an instance of rule $\rulepair$ on the branch. By
  our assumption that all applicable rules are applied in some fair way, there
  must be such an instance. Since it replaces all occurrences of $t$ with a
  pair of fresh variables, there is  exactly one such application,
  such that $x_1$ and $x_2$ are uniquely determined.
\item If $t$ is a term $\la t_1, t_2\ra$, then   $t^{M_b} = \la t_1^{M_b}, t_2^{M_b}\ra$.
\item If $t$ is a term $\pi_1(t')$, then $t^{M_b}$ is $x_1$, where $t'^{M_b} =
  \la x_1, x_2 \ra$.
\item If $t$ is a term $\pi_2(t')$, then $t^{M_b}$ is $x_2$, where $t'^{M_b} =
  \la x_1, x_2\ra$.
\end{itemize}
The memberships
correspond to the  membership atoms that appear on the left of any
sequent in $b$, and also the atoms that appear negated on the right hand
side of any sequent.

We claim that $M_b$ is the desired countermodel. 
It suffices to show
that for every sequent $\Theta; \; \Gamma, \seq \Delta$ in $b$,
$M_b$ is a counterexample to the sequent: it satisfies the conjunction of
formulas on the left and none of the formulas on the right.
We prove this by induction on the logical complexity of the formula.
Each inductive step will involve the assumptions about inference rules
not terminating proof search.
For example, suppose for some sequent $b_i$ in $b$ of the above form,
 $\Delta$ contains $\rho_1 \vee \rho_2$, we want
to show that $M_b$ satisfies $\neg (\rho_1 \vee \rho_2)$. 

But we know that in some successor (viewed root-first) $b_j$ of $b_i$, we
would have applied $\vee$-R, and thus have a descendant with $\rho'_1,
\rho'_2$ within the right, where $\rho'_1$ is identical to $\rho_1$ and
$\rho'_2$ to $\rho_2$, except possibly modified by applications of rules
$\rulepair$ and $\ruleproj$ in between $b_i$ and~$b_j$. By induction $M_b$
satisfies $\neg \rho_1'$ and $\neg \rho_2'$. The rules $\rulepair$ and
$\ruleproj$ just perform replacement of terms, where the replaced and the
replacing terms have the same denotation in $M_b$. 
Hence $M_b$ also satisfies $\neg \rho_1$ and
$\neg \rho_2$. Thus $M_b$ satisfies $\neg (\rho_1 \vee \rho_2)$ as desired.
The other connectives and quantifiers are handled similarly.

%% file: app-focused.tex
\section{Completeness -- translating to \normalized form} \label{app:focused}

We show the completeness of our \normalized calculus of
Figure~\ref{fig:dzcalc-1sided} by giving a translation from proofs in the
high-level calculus of Figure~\ref{fig:dzcalc-2sided}, whose completeness can be
established with standard techniques, as outlined in
Appendix~\ref{app:completeness}. We will give a precise account of the
translation that starts from the following \name{base calculus}.

\begin{defi}
  \label{def:dzcalc:base}
Define the \name{base calculus} for $\deltazero$ formulas as the following
modification of the \normalized calculus of Figure~\ref{fig:dzcalc-1sided}.
\begin{enumerate}
\item Remove all $\EL$ decorations, i.e. do not require any contexts or
  sequents to be $\EL$.
\item In the $\exists$ rule drop the requirement that $t$ is a tuple-term.
\end{enumerate}
\end{defi}

This base calculus may be seen as an intermediate between the high-level and
the \normalized system. The high-level calculus in essence just provides some
sugar on the base calculus, such that translation between both is
straightforward. %
On the
other hand, the base calculus is a relaxed version of the \normalized
calculus: if a base calculus proof satisfies the respective $\EL$ requirements
and all instances of the $\exists$ rule have as $t$ a tuple-term, the proof
can be considered as a proof in the \normalized calculus.
The objective of our translation is then to enforce these constraints. The core
technique for achieving the $\EL$ requirements is \emph{permutation of rules},
investigated for standard systems in \cite{kleene:permutability} and
\cite[Section~5.3]{bpt}, adapted here to our proof system and applied in
specific ways.
Of particular relevance for us is the interaction of rule
permutation with the following property of sequents in a proof.

\begin{defi}
  \label{def:rule:dominated}
  A top-level occurrence of a logic operator $\lor$, $\land$, or $\forall$ in
  a proof sequent is called \defname{\ruledominated} if in all branches
  through the sequent it is a
  descendant\footnote{In the non-strict sense: an occurrence is a descendant
  of itself.} of the principal formula of an instance of the respective rule
  that introduces the operator. A proof is called \defname{\ALruledominated}
  if the top-level occurrences of $\lor$, $\land$ and $\forall$ in all its
  sequents are rule dominated.
\end{defi}

Recall that a multiset of formulas is $\EL$ if its members are
positive atoms, negative atoms, existential quantifications and the
truth-value constant $\bot$. In other words, it does not contain formulas with
$\lor$, $\land$, $\forall$ or $\top$ as top-level operator.
The
following lemma is then easy to see.

\begin{lem}
  \label{lem:ruledom:el}
  In an \ALruledominated proof, any sequent that is not below\footnote{In the
  non-strict sense: the conclusion of a rule instance is below the rule
  instance.} an instance of one of the rules $\lor$, $\land$, or $\forall$ has
  the $\EL$ property.
\end{lem}

Correspondingly, our translation achieves the $\EL$ requirements by first
extending the upper part of the given proof to bring it into \ALruledominated
form and then permuting $\lor$, $\land$, or $\forall$ down over $\exists$,
$\neq$, $\rulepair$ and $\ruleproj$, where the \ALruledominated property is
preserved. In a final step, parts of the proof are rewritten with special
transformations to ensure that in all instances of the $\exists$ rule the term
$t$ is a tuple-term. Figure~\ref{fig:proc:normalize} specifies the steps of
this method in detail.

To make claims about the involved rule permutations precise, we define
\name{permutable} as follows.
\begin{defi}
  Rule~$R_{\alpha}$ (with one or two premises) is said to be
  \defname{permutable down over} rule~$R_{\beta}$ (with one premise) if
  for all instances $\alpha$ of $R_{\alpha}$ and $\beta$ of $R_{\beta}$ such
  that
  \begin{enumerate}
  \item $\alpha$ is immediately above $\beta$ where the conclusion of $\alpha$
    is the premise of $\beta$,
  \item the principal formulas of $\alpha$ are disjoint with the
    active formulas of $\beta$,
  \end{enumerate}
  $\alpha$ or $\beta$ is void, i.e. has a single premise that is identical to
  its conclusion, or there is a deduction from the premises of $\alpha$ to the
  conclusion of $\beta$ that consists of instances~$\beta'_i$ of $R_{\beta}$,
  one for each premise of $\alpha$, whose conclusions are the premises of an
  instance $\alpha'$ of $R_{\alpha}$, where $\alpha'$ has the same conclusion
  as $\beta$.
\end{defi}

Also more general versions of \name{permutability} are possible, where, e.g.,
permutation of a one-premise rule down over a multi-premise rule like $\land$
is considered \cite{bpt}, but not required for our purposes.

\begin{figure}

  \fbox{
\begin{minipage}{0.96\textwidth}
  \noindent\textsc{Input: } A proof in the base calculus for $\deltazero$
  formulas (Def.~\ref{def:dzcalc:base}).
 
  \smallskip
  \noindent\textsc{Method: }

   \begin{enumerate}
   \item \label{proc:normalize:step:pre} \textit{Bringing to \ALruledominated
     form by extending upwards at the top.} Bring the proof into
     \ALruledominated form by exhaustively extending its top nodes as follows:
     Select a top node (instance of $=$ or $\top$) whose sequent contains a
     formula occurrence with $\lor$, $\land$ or $\forall$ as top-level
     operator. Attach an instance of the rule for the respective operator such
     that the former top sequent becomes its conclusion, with the formula
     occurrence as principal formula.
   \item \label{proc:normalize:step:permute} \textit{Permuting rules to
     achieve the $\EL$ condition.} Permute all instances of $\lor$, $\land$
     and $\forall$ down over instances of $\exists$, $\neq$, $\rulepair$ and
     $\ruleproj$. Replace all derivations of sequents that have $\top$ as a
     member with applications of the $\top$ rule. Add the $\EL$ superscript to
     the contexts of the instances of $\exists$ and $\neq$ as well as to the
     sequents in instances of $\rulepair$ and $\ruleproj$, as required by the
     \normalized calculus.
   \item \label{proc:normalize:step:ex} \textit{Bringing terms $t$ in all
     instances of rule $\exists$ to tuple-term form}. As long as there is an
     occurrence of $\pi_1$ or $\pi_2$ in a term $t$ of an instance of
     $\exists$, rewrite the proof part rooted at that instance with
     proof transformation $\elimprojvar$ or $\elimprojpair$.
   \end{enumerate}

  \noindent\textsc{Output: } A proof in the \normalized calculus for
  $\deltazero$ formulas (Figure~\ref{fig:dzcalc-1sided}) with the same bottom
  sequent as the input proof.
  
\end{minipage}
}

\caption{\Normalization procedure.}
\label{fig:proc:normalize}
\end{figure}

Permutability properties of the rules of our base calculus, along with proof
transformations to achieve the tuple-term property required for instances of
the $\exists$ rule, justify the following claim.
\begin{thm}[Completeness of the \normalized calculus for 
    $\deltazero$ formulas (Figure~\ref{fig:dzcalc-1sided})]
  \label{thm:dzcalc:1sided:complete}
  A $\deltazero$ formula provable in the base calculus for $\deltazero$
  formulas (Def.~\ref{def:dzcalc:base}) is also provable in the \normalized
  calculus for $\deltazero$ formulas (Figure~\ref{fig:dzcalc-1sided}).
\end{thm}

\begin{proof}
  Figure~\ref{fig:proc:normalize} shows a procedure for translating a proof in
  the base calculus to a proof in the \normalized calculus. Both proofs have
  the same bottom sequent, that is, prove the same $\deltazero$ formula. For
  step~\ref{proc:normalize:step:pre} of the procedure it is evident that the
  result is still a legal proof with the same bottom sequent and is in
  \ALruledominated form.

  The rule permutations in step~\ref{proc:normalize:step:permute} are
  justified by Lemma ~\ref{lemma:permute:generic} and
  Lemma~\ref{lemma:permute:down:over:pairing:projection}, presented below. It
  is easy to see from the particular deduction conversions used to justify
  these lemmas that the permutations preserve the \ALruledominated property of
  the proof. From Lemma~\ref{lem:ruledom:el}, the $\EL$ property required by
  the \normalized calculus for instances of $\exists$, $\neq$, $\rulepair$ and
  $\ruleproj$ then follows.
  Step~\ref{proc:normalize:step:permute} finishes with straightforward
  conversions to let the proof literally match with the \normalized calculus:
  truncating redundant proof parts above sequents with $\top$ as a member and
  adding the $\EL$ decoration.
  
  Step~\ref{proc:normalize:step:ex} is justified by
  Lemma~\ref{lemma:exists:tuple-term}, presented below in
  Section~\ref{subsec:tuple:terms}. The overall bottom sequent with the proven
  formula is unaltered by all involved permutations and transformations. The
  proof is then a proof in the \normalized calculus with the same bottom
  sequent as the initially given proof.
\end{proof}

In the following Sects.~\ref{subsec:perm:generic}--\ref{subsec:tuple:terms} we
supplement the permutation lemmas and proof transformations referenced in
Fig~\ref{fig:proc:normalize} and in the proof of
Theorem~\ref{thm:dzcalc:1sided:complete}.

\subsection{Permutability justified by generic permutation schemas}
\label{subsec:perm:generic}

We begin our justification of the permuting of rules, used in the normalization algorithm.
Permutability among rules that correspond to logic operators can be shown with
straightforward generic schemas. The following lemmas express cases that were
used in the proof of Theorem~\ref{thm:dzcalc:1sided:complete}.

\begin{lem}
  \label{lemma:permute:generic}
  In the base calculus $\lor$, $\land$ and $\forall$ are permutable down
  over $\exists$ and $\neq$.
\end{lem}

Lemma~\ref{lemma:permute:generic} hold since all of the involved permutations
can be performed according to two generic schemas, one for permuting a
one-premise rule down over a one-premise rule and a second schema for
permuting the two-premise rule $\land$ down over a one-premise rule. The
considered one-premise rules $\neq, \lor, \forall, \exists$ are all of the
form
\begin{equation}
  \label{eq:form:one:premise:rule}
\begin{array}{c}
  \Theta, C \seq A, \Delta\\\midrule
  \Theta, C' \seq A', \Delta,
\end{array}
\end{equation}
where $C, C'$ are (possibly empty) multisets of membership atoms and $A, A'$
are multisets of formulas. For instances of the rule, $C$ and $A$ constitute
the \name{active formulas} and $C', A'$ the \name{principal formulas}.

Consider an instance $\alpha$ of a one-premise rule $R_\alpha$ followed by an
instance $\beta$ of a one-premise rule $R_\beta$, both rules of the
form~(\ref{eq:form:one:premise:rule}). Let the active formulas of the rule
instances be $C, A$ and $D, B$, respectively, and let their principal formulas
be $C', A'$ and $D', B'$, respectively. Assume that the multiset $A'$ of
principal formulas of $\alpha$ and the multiset $B$ of active formulas of
$\beta$ have no formula occurrences in common. ($C'$ and $D$ are not
constrained in this way.) The two rule applications then match the left side
of the following permutation schema. Except in the case where $R_\alpha =
\exists$ and $R_\beta = \forall$, the schema can be applied to permute the
instance of $R_\alpha$ down over that of $R_\beta$.
\begin{equation}
  \label{eq:perm:schema:one:one}
  \begin{array}{rc}
    \tabrulename{$R_{\alpha}$} 
    & \Theta, C, D \seq A, B, \Delta\\\cmidrule{2-2}
    \tabrulename{$R_{\beta}$} 
    & \Theta, C',D \seq A', B, \Delta\\\cmidrule{2-2}
    & \Theta, C',D' \seq A', B', \Delta
  \end{array}
  \;\;\text{ becomes }\;\;
  \begin{array}{rc}
    \tabrulename{$R_{\beta}$} 
    & \Theta, C, D \seq A, B, \Delta\\\cmidrule{2-2}
    \tabrulename{$R_{\alpha}$} 
    & \Theta, C, D' \seq A, B', \Delta\\\cmidrule{2-2}
    & \Theta, C',D' \seq A', B', \Delta
  \end{array}
\end{equation}

Permutation schema~(\ref{eq:perm:schema:one:one}) can for example be applied
to permute $\forall$ down over $\exists$:
\[\begin{array}{rc}
\tabrulename{$\forall$}
& \Theta, t \in c, y \in b \seq
  \phi[y/x],
  \psi[t/z],
  \exists z \in c \fdot \psi, \Delta\\\cmidrule{2-2}
\tabrulename{$\exists$}
& \Theta, t \in c\seq
  \forall x \in b \fdot \phi,
  \psi[t/z],
  \exists z \in c \fdot \psi,
  \Delta\\\cmidrule{2-2}
& \Theta, t \in c\seq
  \forall x \in b \fdot \phi,
  \exists z \in c \fdot \psi,
  \Delta
\end{array}
\]
\noindent becomes
\[\begin{array}{rc}
\tabrulename{$\exists$}
& \Theta, t \in c, y \in b \seq
  \phi[y/x],
  \psi[t/z],
  \exists z \in c \fdot \psi,
  \Delta\\\cmidrule{2-2}
\tabrulename{$\forall$}
& \Theta, t \in c, y \in b \seq
  \phi[y/x],
  \exists z \in c \fdot \psi,
  \Delta\\\cmidrule{2-2}
& \Theta, t \in c\seq
  \forall x \in b \fdot \phi,
  \exists z \in c \fdot \psi,
  \Delta
\end{array}
\]

The excluded case of permuting $\exists$ down over $\forall$ is not required
for our purposes. To illustrate how this case might fail, consider the
following instantiation, which shows two legal inferences matching the left
side of permutation schema~(\ref{eq:perm:schema:one:one}), but clearly not
permitting permutation.
\[\begin{array}{rc}
\tabrulename{$\exists$}
& \Theta, y \in b \seq
  \phi[y/x],
  \psi[y/z],
  \exists z \in b \fdot \psi,
  \Delta\\\cmidrule{2-2}
\tabrulename{$\forall$}
& \Theta, y \in b \seq
  \phi[y/x],
  \exists z \in b \fdot \psi,
  \Delta\\\cmidrule{2-2}
& \Theta\seq
  \forall x \in b \fdot \phi,
  \exists z \in b \fdot \psi,
  \Delta
\end{array}
\]

An instance of the two-premise rule $\land$ can be permuted below an instance
of a one-premise rule $R_{\beta}$ of the form~(\ref{eq:form:one:premise:rule})
with active formulas $D, B$ and principal formulas $D', B'$ according to the
following permutation schema.
\[
\begin{array}{rccrc}
  \tabrulename{$\land$} &
  \Theta, D \seq \phi_1, B, \Delta &&&
   \Theta, D \seq \phi_2, B, \Delta\\\cmidrule{2-5}
  \tabrulename{$R_{\beta}$} &
  \multicolumn{4}{c}{\Theta, D \seq (\phi_1 \land \phi_2), B, \Delta}\\\cmidrule{2-5}
  & \multicolumn{4}{c}{\Theta, D' \seq (\phi_1 \land \phi_2), B', \Delta}
\end{array}
\]
\begin{flalign}
  \label{eq:perm:schema:and:one}
\text{becomes} &&
\end{flalign}
\[
\begin{array}{rccrc}
  \tabrulename{$R_{\beta}$} &
  \Theta, D \seq \phi_1, B, \Delta &&
  \tabrulename{$R_{\beta}$} &
  \Theta, D \seq \phi_2, B, \Delta\\\cmidrule{2-2}\cmidrule{5-5}
  \tabrulename{$\land$} &
  \Theta, D' \seq \phi_1, B', \Delta &&&
  \Theta, D' \seq \phi_2, B', \Delta\\\cmidrule{2-5}
  & \multicolumn{4}{c}{\Theta, D' \seq (\phi_1 \land \phi_2), B', \Delta}
  \end{array}
\]

As an example for the permutation schema~(\ref{eq:perm:schema:and:one})
consider permuting $\land$ down over $\neq$:
\[\begin{array}{rcc}
\tabrulename{$\land$}
& \Theta \seq \phi_1,
  t \nequ u, \alpha[u/z], \alpha[t/z],
  \Delta\;\;
  & \;\;\Theta \seq \phi_2,
  t \nequ u, \alpha[u/z], \alpha[t/z],
  \Delta\\\cmidrule{2-3}
\tabrulename{$\neq$}
& \multicolumn{2}{c}{\Theta \seq \phi_1 \land \phi_2,
  t \nequ u, \alpha[u/z], \alpha[t/z],
  \Delta}\\\cmidrule{2-3}
& \multicolumn{2}{c}{\Theta \seq \phi_1 \land \phi_2,
  t \nequ u, \alpha[t/z],
  \Delta}
\end{array}  
\]
\noindent becomes
\[\begin{array}{rccrc}
\tabrulename{$\neq$}
& \Theta \seq \phi_1,
  t \nequ u, \alpha[u/z], \alpha[t/z],
  \Delta
  & &
\tabrulename{$\neq$}  
  & \Theta \seq \phi_2,
  t \nequ u, \alpha[u/z], \alpha[t/z],
  \Delta\\\cmidrule{2-2}\cmidrule{5-5}
  \tabrulename{$\land$}  
  & \Theta \seq \phi_1,
  t \nequ u, \alpha[t/z],
  \Delta\;\;
  &&& \;\;\Theta \seq \phi_2,
  t \nequ u, \alpha[t/z]
  \Delta\\\cmidrule{2-5}
  & \multicolumn{4}{c}{\Theta \seq \phi_1 \land \phi_2,
    t \nequ u, \alpha[t/z],
  \Delta}
\end{array}  
\]

\subsection{Permutability down over pairing and projection}
\label{subsec:perm:pair:proj}

We now continue the justification of the rule permutations used in proof normalization.
The rules $\rulepair$ and $\ruleproj$ do not fit the generic
shape~(\ref{eq:form:one:premise:rule}) underlying the generic permutation
schemas~(\ref{eq:perm:schema:one:one}) and (\ref{eq:perm:schema:and:one}).
Hence permuting a rule that corresponds to a logic operator down under
$\rulepair$ or $\ruleproj$ requires dedicated schemas. We consider
rule~$\ruleproj$ there just for $\pi_1$ (for $\pi_2$ the analogy is obvious).
The following lemma expresses the cases of permutability over $\rulepair$ and
$\ruleproj$ needed in the proof of Theorem~\ref{thm:dzcalc:1sided:complete}.

\begin{lem}
  \label{lemma:permute:down:over:pairing:projection}
  In the base calculus $\lor$, $\land$ and $\forall$ are permutable down
  over $\rulepair$ and $\ruleproj$. 
\end{lem}

The permutation schemas for $\lor$ and $\land$ down over $\rulepair$ and
$\ruleproj$ are straightforward. In their specification we use the following
shorthands for a formula or multiset of formulas $\Gamma$: $\Gamma^*$ stands
for $\Gamma[\la x_1,x_2\ra/x]$, $\Gamma'$ for $\Gamma[t_1/x]$ and $\Gamma''$
for $\Gamma[\pi_1(\la t_1,t_2\ra)/x]$. The respective permutation schemas then
are as follows.

\[\begin{array}{rc}
\tabrulename{$\lor$}
& \Theta^* \seq \phi_1^*,\phi_2^*,
  \Delta^*\\\cmidrule{2-2}
\tabrulename{$\rulepair$}
& \Theta^* \seq \phi_1^* \lor \phi_2^*,
   \Delta^*\\\cmidrule{2-2}
   & \Theta \seq
   \phi_1 \lor \phi_2,
   \Delta
\end{array}  
\;\; \text{ becomes }
\begin{array}{rc}
\tabrulename{$\rulepair$}
& \Theta^* \seq \phi_1^*,\phi_2^*,
  \Delta^*\\\cmidrule{2-2}
\tabrulename{$\lor$}
& \Theta \seq \phi_1,\phi_2,
  \Delta\\\cmidrule{2-2}
   & \Theta \seq
   \phi_1 \lor \phi_2,
   \Delta
\end{array}  
\]

\medskip

\[\begin{array}{rcc}
\tabrulename{$\land$}
& \Theta^* \seq \phi_1^*,
  \Delta^*\;\;
  & \;\;\Theta^* \seq \phi_2^*,
  \Delta^*\\\cmidrule{2-3}
\tabrulename{$\rulepair$}
& \multicolumn{2}{c}{\Theta^* \seq \phi_1^* \land \phi_2^*,
  \Delta^*}\\\cmidrule{2-3}
& \multicolumn{2}{c}{\Theta \seq \phi_1 \land \phi_2,
  \Delta}
\end{array}  
\;\;\text{ becomes }\;\;
\begin{array}{rccrc}
\tabrulename{$\rulepair$}
& \Theta^* \seq \phi_1^*,
  \Delta^*
  & &
\tabrulename{$\rulepair$}  
  & \Theta^* \seq \phi_2^*,
  \Delta^*\\\cmidrule{2-2}\cmidrule{5-5}
  \tabrulename{$\land$}  
  & \Theta \seq \phi_1,
  \Delta\;\;
  &&& \;\;\Theta \seq \phi_2,
  \Delta\\\cmidrule{2-5}
  & \multicolumn{4}{c}{\Theta \seq \phi_1 \land \phi_2,
  \Delta}
\end{array}  
\]

\medskip

\[\begin{array}{rc}
\tabrulename{$\lor$}
& \Theta' \seq \phi_1',\phi_2',
  \Delta'\\\cmidrule{2-2}
\tabrulename{$\ruleproj$}
& \Theta' \seq \phi_1' \lor \phi_2',
   \Delta'\\\cmidrule{2-2}
   & \Theta'' \seq
   \phi_1'' \lor \phi_2'',
   \Delta''
\end{array}  
\;\; \text{ becomes }\;\;
\begin{array}{rc}
\tabrulename{$\ruleproj$}
& \Theta' \seq \phi_1',\phi_2',
  \Delta'\\\cmidrule{2-2}
\tabrulename{$\lor$}
& \Theta'' \seq \phi_1'',
\phi_2'',
   \Delta''\\\cmidrule{2-2}
   & \Theta'' \seq
   \phi_1'' \lor \phi_2'',
   \Delta''
\end{array}  
\]

\medskip

\[\begin{array}{rcc}
\tabrulename{$\land$}
& \Theta' \seq \phi_1',
  \Delta'\;\;
  & \;\;\Theta' \seq \phi_2',
  \Delta'\\\cmidrule{2-3}
\tabrulename{$\ruleproj$}
& \multicolumn{2}{c}{\Theta' \seq \phi_1' \land \phi_2',
  \Delta'}\\\cmidrule{2-3}
& \multicolumn{2}{c}{\Theta'' \seq \phi_1'' \land \phi_2'',
  \Delta}
\end{array}  
\;\;\text{ becomes }\;\;
\begin{array}{rccrc}
\tabrulename{$\ruleproj$}
& \Theta' \seq \phi_1',
  \Delta'
  & &
\tabrulename{$\ruleproj$}  
  & \Theta' \seq \phi_2',
  \Delta'\\\cmidrule{2-2}\cmidrule{5-5}
  \tabrulename{$\land$}  
  & \Theta'' \seq \phi_1'',
  \Delta''\;\;
  &&& \;\;\Theta'' \seq \phi_2'',
  \Delta''\\\cmidrule{2-5}
  & \multicolumn{4}{c}{\Theta'' \seq \phi_1'' \land \phi_2'',
  \Delta''}
\end{array}  
\]

Permuting $\forall$ down over $\rulepair$ and $\ruleproj$ is more
intricate, because of the potential interaction of involved variables and
substitutions.
We make here explicit use of the assumption w.l.o.g. that bound and free
variables are kept distinct.

For permuting $\forall$ down
over $\rulepair$ the given deduction is, under the assumption $x \neq z$, as
follows.
\begin{equation}
  \label{eq:all:pair:given}
\begin{array}{rc}
\tabrulename{$\forall$}
& \Theta[\la z_1,z_2\ra/z], y \in b[\la z_1,z_2\ra/z] \seq
\phi[\la z_1,z_2\ra/z,y/x],
\Delta[\la z_1,z_2\ra/z]\\\cmidrule{2-2}
\tabrulename{$\rulepair$}
& \Theta[\la z_1,z_2\ra/z] \seq
\forall x \in b[\la z_1,z_2\ra/z] \fdot \phi[\la z_1,z_2\ra/z],
\Delta[\la z_1,z_2\ra/z]\\\cmidrule{2-2}
& \Theta \seq
\forall x \in b \fdot \phi,
\Delta
\end{array}
\end{equation}
We can exclude the cases $x = z_1$ or $x = z_2$, which would violate the
freshness requirement of $\rulepair$ because $x$ already appears in the bottom
sequent.

\noindent
Assuming $x \neq z$, $x \neq z_1$ and $x \neq z_2$ it holds that
\begin{gather*}
\phi[\la z_1,z_2\ra/z,y/x] = \phi[y/x, \la z_1,z_2\ra/z]
\end{gather*}
and thus (\ref{eq:all:pair:given}) becomes
\begin{equation}
  \label{eq:all:pair:alldifferent}
\begin{array}{rc}
\tabrulename{$\rulepair$}
& \Theta[\la z_1,z_2\ra/z], y \in b[\la z_1,z_2\ra/z] \seq
\phi[y/x,\la z_1,z_2\ra/z],
\Delta[\la z_1,z_2\ra/z]\\\cmidrule{2-2}
\tabrulename{$\forall$}
& \Theta, y \in b \seq
\phi[y/x],
\Delta\\\cmidrule{2-2}
& \Theta \seq
\forall x \in b \fdot \phi,
\Delta
\end{array}
\end{equation}

So far, for permuting $\forall$ down over $\rulepair$ we assumed $x \neq z$.
It remains to consider the case $x = z$.
Since $x$ occurs bound in the bottom sequent, it cannot occur in $\Theta$,
$\Delta$ and $b$. Since $\phi$ is in that sequent in the scope of a quantification
upon $x$, the $\rulepair$ inference cannot have any effect on $\phi$.
Hence, in the case $x = z$ the $\rulepair$ inference would
have no effect at all. Just deleting this void inference is then the
result of permuting. This concludes the specification of permuting
$\rulepair$ down over $\forall$.

We now turn to permuting $\forall$ down over $\ruleproj$. Assuming w.l.o.g.
that $z$ does not occur in any sequent of the proof and thus $x \neq z$, the given
deduction is as follows.
\begin{equation}
  \label{eq:all:proj:given}
\begin{array}{rc}
\tabrulename{$\forall$}
& \Theta[t_1/z], y \in b[t_1/z] \seq
\phi[t_1/z,y/x],
\Delta[t_1/z]\\\cmidrule{2-2}
\tabrulename{$\ruleproj$}
& \Theta[t_1/z] \seq
\forall x \in b[t_1/z] \fdot \phi[t_1/z],
\Delta[t_1/z]\\\cmidrule{2-2}
& \Theta[\pi_1(\la t_1,t_2\ra)/z] \seq
  \forall x \in b[\pi_1(\la t_1,t_2\ra)/z] \fdot \phi[\pi_1(\la t_1,t_2\ra)/z],
  \Delta[\pi_1(\la t_1,t_2\ra)/z]
\end{array}
\end{equation}
In the case where $x$ occurs neither in $t_1$ nor in $t_2$ it holds that
\begin{gather*}
  \phi[t_1/z,y/x] =  \phi[y/x,t_1/z]\\
  \phi[y/x,\pi_1(\la t_1, t_2\ra)/z] =
  \phi[\pi_1(\la t_1, t_2\ra)/z,y/x]
\end{gather*}
and thus (\ref{eq:all:proj:given}) becomes
\[\begin{array}{rc}
\tabrulename{$\ruleproj$}
& \Theta[t_1/z], y \in b[t_1/z] \seq
\phi[y/x,t_1/z],
\Delta[t_1/z]\\\cmidrule{2-2}
\tabrulename{$\forall$}
& \Theta[\pi_1(\la t_1,t_2\ra)/z], y \in b[\pi_1(\la t_1,t_2\ra)/z] \seq
\phi[\pi_1(\la t_1,t_2\ra)/z,y/x],
\Delta[t_1/z]\\\cmidrule{2-2}
& \Theta[\pi_1(\la t_1,t_2\ra)/z] \seq
  \forall x \in b[\pi_1(\la t_1,t_2\ra)/z] \fdot \phi[\pi_1(\la t_1,t_2\ra)/z],
  \Delta[\pi_1(\la t_1,t_2\ra)/z]
\end{array}
\]

We now consider the case where $x$ occurs in $t_1$ or in $t_2$.
Since $x$ occurs bound in the bottom sequent, there can be no free occurrences
of $x$ in that sequent.
Hence~$z$, mapped in that sequent to a pair in which~$x$ occurs,
cannot occur in $\Theta$, $\Delta$ and $b$. The given deduction
(\ref{eq:all:proj:given}) then specializes to
\begin{equation}
  \label{eq:proj:given:x:zi}
\begin{array}{rc}
\tabrulename{$\forall$}
& \Theta, y \in b \seq
\phi[t_1/z,y/x],
\Delta\\\cmidrule{2-2}
\tabrulename{$\ruleproj$}
& \Theta \seq
\forall x \in b \fdot \phi[t_1/z],
\Delta\\\cmidrule{2-2}
& \Theta \seq
  \forall x \in b \fdot \phi[\pi_1(\la t_1,t_2\ra)/z],
  \Delta
\end{array}
\end{equation}
Since $x \neq z$ it holds that
\begin{gather}
  \phi[t_1/z,y/x] = \phi[y/x,t_1[y/x]/z]\\
  \label{eq:proj:given:x:zi:equiv:2}
  \phi[y/x, \pi_1(\la t_1, t_2\ra)[y/x]/z] = \phi[\pi_1(\la t_1, t_2\ra)/z,y/x]
\end{gather}
and thus (\ref{eq:proj:given:x:zi}) becomes
\begin{equation}
  \label{eq:proj:perm:x:zi}
\begin{array}{rc}
\tabrulename{$\ruleproj$}
& \Theta, y \in b \seq
\phi[y/x,t_1[y/x]/z],
\Delta\\\cmidrule{2-2}
\tabrulename{$\forall$}
& \Theta, y \in b \seq
\phi[\pi_1(\la t_1,t_2\ra)/z,y/x],
\Delta\\\cmidrule{2-2}
& \Theta \seq
  \forall x \in b \fdot \phi[\pi_1(\la t_1,t_2\ra)/z],
  \Delta
\end{array}
\end{equation}
In~(\ref{eq:proj:perm:x:zi}) the $\ruleproj$ inference viewed top-down
replaces occurrences of $t_1[y/x]$, which, in case $t_1 = x$ is $y$, with
$\pi_1(\la t_1, t_2\ra)[y/x]$.
Justified by~(\ref{eq:proj:given:x:zi:equiv:2}), the middle sequent is shown
in~(\ref{eq:proj:perm:x:zi}) in the form suitable as premise for the
application of $\forall$.

This concludes the specification of permuting $\forall$ down over $\ruleproj$
and thus the specification of all permutation schemas required for
step~\ref{proc:normalize:step:permute} of the \normalization procedure
(Figure~\ref{fig:proc:normalize}).

\subsection{Conversion to tuple-terms}
\label{subsec:tuple:terms}

The objective of this conversion is to ensure that the term~$t$ in instances
of the $\exists$ rule is always a tuple-term, i.e., does not involve the
projection functions $\pi_1$ and $\pi_2$. This is achieved with proof
transformations where instances of $\rulepair$ and $\ruleproj$ are inserted
below instances of $\exists$ and the effects of these interspersed rules are
propagated upwards. The following lemma states the desired overall property of
this transformation.

\begin{lem}
  \label{lemma:exists:tuple-term}
  A $\deltazero$ formula provable in the base calculus with a proof where the
  $\EL$ requirements of the \normalized calculus of
  Figure~\ref{fig:dzcalc-1sided} are met (but not necessarily its tuple-term
  condition of the $\exists$ rule) is also provable in the \normalized
  calculus.
\end{lem}

Lemma~\ref{lemma:exists:tuple-term} presupposes a proof whose contexts and
sequents meet the $\EL$ requirements of the \normalized calculus. The claimed
tuple-term condition can be achieved by exhaustively applying the proof
transformations $\elimprojvar$ and $\elimprojpair$, specified below, which are
applicable whenever there is an occurrence of $\pi_1$ or $\pi_2$ in the term
$t$ of an instance of the $\exists$ rule. These transformations lead from a
legal proof in the base calculus to another one with the same bottom formula.
They require and preserve the $\EL$ requirements from the \normalized
calculus.

\subsubsection*{Proof transformation $\elimprojvar$}
This transformation applies to an instance of the $\exists$ rule whose term
$t$ has an occurrence of $\pi_1(z)$ or $\pi_2(z)$, where $z$ is a variable.
The instance of the $\exists$ rule is replaced with an instance of $\exists$
followed by an instance of $\rulepair$ and the subproof above is adjusted,
according to the following schema.
\[\small
  \begin{array}{rc}
    & P\\[-1ex]
    & \prooftreesymbol\\
        \tabrulename{$\exists$}
    & \Theta, t \in b \seq \phi[t/x], \exists x \in b\; \phi, \Delta_\EL\\\cmidrule{2-2}
    & \Theta, t \in b \seq \exists x \in b\, .\, \phi, \Delta_\EL
  \end{array}
  \;\textrm{ is replaced by }\;
  \begin{array}{rc}
    & P^*\\[-1ex]
    & \prooftreesymbol\\
    \tabrulename{$\exists$}
    & \Theta^*, t^* \in b^* \seq \phi^*[t^*/x], \exists x \in b^*\; \phi, \Delta^*_\EL\\\cmidrule{2-2}
    \tabrulename{$\rulepair$} 
    & \Theta^*, t^* \in b^* \seq \exists x \in b^*\, .\, \phi^*, \Delta^*_\EL\\\cmidrule{2-2}
    & \Theta, t \in b \seq \exists x \in b\, .\, \phi, \Delta_\EL
  \end{array}
\]
The starred components are specified there as follows: For a term, formula,
or multiset of formulas $\Gamma$, $\Gamma^*$ stands for $\Gamma[\la
  z_1,z_2\ra/z]$, where $z_1,z_2$ are fresh variables. Subproof~$P^*$ is
obtained from~$P$ by removing each instance of $\rulepair$ upon the variable
$z$ and some pair $\la z_1', z_2'\ra$, while replacing and $z_1'$ with $z_1$
and of $z_2'$ with $z_2$, and then replacing all occurrences of $z$ with $\la
z_1, z_2\ra$.

\subsubsection*{Proof transformation $\elimprojpair$.}
This transformation applies to an instance of the $\exists$ rule whose term
$t$ has an occurrence of $\pi_1(\la t_1, t_2\ra)$ (the case for $\pi_2(\la
t_1, t_2\ra)$ is analogous). The instance of the $\exists$ rule is replaced
with an instance of $\exists$ followed by an instance of $\ruleproj$ and the
subproof above is adjusted, according to the following schema.
\[\small
  \begin{array}{rc}
    & P\\[-1ex]
    & \prooftreesymbol\\
        \tabrulename{$\exists$}
    & \Theta, t \in b \seq \phi[t/x], \exists x \in b\; \phi, \Delta_\EL\\\cmidrule{2-2}
    & \Theta, t \in b \seq \exists x \in b\, .\, \phi, \Delta_\EL
  \end{array}
  \;\textrm{ is replaced by }\;
  \begin{array}{rc}
    & P'\\[-1ex]
    & \prooftreesymbol\\
    \tabrulename{$\exists$}
    & \Theta, t' \in b \seq \phi[t'/x], \exists x \in b\; \phi, \Delta_\EL\\\cmidrule{2-2}
    \tabrulename{$\ruleproj$} 
    & \Theta, t' \in b \seq \exists x \in b\, .\, \phi, \Delta_\EL\\\cmidrule{2-2}
    & \Theta, t \in b \seq \exists x \in b\, .\, \phi, \Delta_\EL
  \end{array}
  \]
The primed components are specified there as follows: Term $t'$ is $t$ with
the occurrence of $\pi_1(\la t_1, t_2\ra)$ replaced by $t_1$. Subproof $P'$ is
$P$ with that replacement propagated upward. It may be obtained by performing
the inference steps of $P$ in upward direction, but now starting from $\Theta,
t' \in b \seq \phi[t'/x], \exists x \in b\; \phi, \Delta_\EL$ instead of
$\Theta, t \in b \seq \phi[t/x], \exists x \in b\; \phi, \Delta_\EL$.
Instances of $\ruleproj$ and $\rulepair$ may get void there (identical premise
and conclusion) and can then be removed.

\subsection{Complexity considerations}
\label{sec:normalization:complexity}

In step~\ref{proc:normalize:step:pre} of the \normalization procedure
(Figure~\ref{fig:proc:normalize}), through the two-premise rule $\land$ a leaf
can get extended upwards by a number of nodes that is exponential in the
number of occurrences of $\land$ in its conclusion. Hence the time required
for step~\ref{proc:normalize:step:pre} is exponential in the number of
occurrences of $\land$ in a leaf conclusion.

Step \ref{proc:normalize:step:permute} where the proof is permuted, may,
through permuting $\land$ down over one-premise rules, increase the number of
nodes. However, none of the permutations increases the height~$h$ or  the
number of leaves~$l$ of the proof tree. Both $h$ and $l$ are less than or
equal to the number $n$ of nodes of the input proof, the output of
step~\ref{proc:normalize:step:pre}. Since a tree with height~$h$ and number of
leaves~$l$ can in general not have more than $h \times l$ nodes, our proof tree
can never have more than $n^2$ nodes when we perform the permutations. Since
for each node we can apply at most $h$ downward permutation steps, and $h \leq
n$, we can conclude that step~\ref{proc:normalize:step:permute}
involves not more that $n^3$ permutation steps.

Step~\ref{proc:normalize:step:ex}, where instances of $\rulepair$ and
$\ruleproj$ are inserted below instances of $\exists$ and
adjustments to these insertions are propagated upwards,
is polynomial. This step does not increase the
number of branches, but may increase the length of a branch by twice the number
of occurrences of~$\pi_i$ in terms $t$ of instances of $\exists$ on the
branch.

Thus, for \normalization we have exponential time complexity. However,
the source of the exponential complexity
can be associated with the first stage of the procedure, extending the proof
tree at its leaves, which may be seen as a blow up of \emph{irrelevant} parts
of the proof that accompany the axioms. The rule permutation stage is only a
polynomial-time operation.

%% file: app-deltazerointerpolation.tex
\section{\texorpdfstring{$\deltazero$}{Δ₀} interpolation: proof sketch of Theorem \ref{thm:interpolationdeltafocus}}
\label{app:deltazerointerpolation}

We recall the statement:

\medskip

Let $\Theta$ be an $\in$-context and $\Gamma, \Delta$ finite \msets of $\deltazero$ formulas.
Then from any proof of $\Theta; \; \Gamma \vdash \Delta$
we can compute in linear time a $\deltazero$ formula $\theta$ with $\FV(\theta) \subseteq \FV(\Theta, \Gamma) \cap \FV(\Delta)$,
such that $\Theta; \; \Gamma \vdash \theta$ and $\emptyset; \; \theta \vdash \Delta$

\medskip

Recall also that we claim this for both the higher-level $2$-sided system
and the $1$-sided system, where
the $2$-sided syntax is a ``macro'':  $\Theta; \; \Gamma \vdash \Delta$
is a shorthand for $\Theta \vdash \neg \Gamma, \Delta$, where $\neg \Gamma$
is itself a macro for dualizing connectives. 
Thus in the $1$-sided version, we are arbitrarily classifying some of the 
$\deltazero$ formulas as Left and the others as Right, and our interpolant
must be common according to that partition.

We stress that there are no new ideas needed in proving Theorem  \ref{thm:interpolationdeltafocus} --- unlike
for our main tool, the Parameter Collection Theorem, or our final result.
The construction for Theorem  \ref{thm:interpolationdeltafocus} proceeds exactly as in prior interpolation theorems 
for similar calculi \cite{takeuti:1987,bpt,smullyan}. Similar constructions
are utilized in works for query reformulation
in databases, so for a presentation geared towards
a database audience one can check   \cite{tomanweddell} or the later
\cite{ interpbook}.

We explain the argument for the higher-level $2$-sided system.
We prove a more general statement, where we partition the context and
the formulas 
on both sides of $\vdash$ into Left and Right.
So we have 
\[
\Theta_L \Theta_R;  \Gamma_L, \Gamma_R \vdash \Delta_L, \Delta_R
\]
And our inductive invariant is that
we will compute in linear time a $\theta$  such that:
\begin{align*}
\Theta_L; \Gamma_L \vdash  \theta, \Delta_L \\
\Theta_R; \Gamma_R, \theta \proves \Delta_R
\end{align*}
And we require that $FV(\theta) \subseteq FV(\Theta_L, \Gamma_L, \Delta_L) \cap
 FV(\Theta_R,\Gamma_R, \Delta_R)$.
This generalization 
is used to handle the negation rules, as we explain below.

We proceed by induction on the depth of the proof tree.

One of the base cases is where we have a trivial proof tree,
which uses rule (Ax) to derive:
\[
\Theta; \Gamma, \phi \proves  \phi, \Delta
\]
We do a case distinction on where the occurrences of $\phi$ sit in our partition.
Assume the occurrence on the left is in $\Gamma_L$ and the occurrence
on the right is in $\Delta_R$. 
Then we can take our interpolant $\theta$ to be $\phi$.
Suppose the occurrence on the left is $\Gamma_L$ and the occurrence on the right
is in $\Delta_L$. Then we can take $\theta$ to be $\bot$.
The other base cases are similar.

The inductive cases for forming the interpolant will work
``in reverse'' for each proof rule. That is,  if we used an inference rule
to derive sequent $S$ from sequents $S_1$ and $S_2$, we will partition
the sequents $S_1$ and $S_2$ based on the partition of $S$.
We will then apply induction to our partitioned
sequent for $S_1$ to get an interpolant $\theta_1$, and also
apply induction to our partitioned version
of $S_2$ to get an interpolant $\theta_2$.
We
then put them together to get the interpolant for the partitioned sequent
$S$. This ``putting together''
will usually reflect the semantics of the connective mentioned in the proof rule.

Consider the case where the last rule applied is the $\neg$-L rule: this is the case
that motivates the more general invariant involving partitions.
We have a partition of the final sequent $\Theta; \Gamma, \phi \proves  \Delta$.
We form a partition of the sequent $\Theta; \Gamma \proves \neg \phi, \Delta$
by placing  $\neg \phi$ on the same side (Left, Right) as $\phi$ was in the original
partition. We then get an interpolant $\theta$ by induction. We just
use $\theta$ for the final interpolant.

We  consider the inductive case for $\wedge$-R.  We have
two top sequents, one for each conjunct. We partition them in the obvious way:
each $\phi_i$ in the top is in the same partition that $\phi_1 \wedge \phi_2$ was
in the bottom. Inductively we take the interpolants
$\theta_1$ and $\theta_2$ for each sequent. We again do a case analysis based on
whether $\phi_1 \wedge \phi_2$ was in $\Delta_L$ or in $\Delta_R$. 

Suppose
$\phi_1 \wedge \phi_2$ was in $\Delta_R$, so $\Delta_R=\phi_1 \wedge \phi_2, \Delta'_R$.
Then we arranged that each $\phi_i$ was
in $\Delta_R$ in the corresponding top sequent. So we know 
that $\Theta_L; \Gamma_L \vdash  \theta_i, \Delta_L$
and $\Theta_R; \Gamma_R, \theta_i \proves \phi_i, \Delta'_R$ for $i=1,2$.
Now we can set the interpolant $\theta$ to be 
 $\theta_1 \wedge \theta_2$.

In the other case, $\phi_1 \wedge \phi_2$ was in $\Delta_L$, say $\Delta_L=\phi_1 \wedge \phi_2, \Delta'_L$.
Then we would arrange each $\phi_i$ to be ``Left'' in the corresponding top sequent, so we know that
$\Theta_L; \Gamma_L \vdash  \theta_i, \phi_i, \Delta'_L$
and $\Theta_R; \Gamma_R, \theta_i \proves \Delta_R$ for $i=1,2$.
We set $\theta=\theta_1 \vee \theta_2$ in this case.

With the $\exists$ rule, a term in the inductively-assumed
   $\theta'$ for the top sequent may become illegal for the
   $\theta$ for the bottom sequent, since it has a free variable that is not common. In this case, the term 
   in $\theta$ is replaced by a quantified variable, where the
   quantifier is existential or universal, depending on the
   partitioning, and bounded according to the requirements
   for $\deltazero$ formulas.
   The non-common variables are then replaced and variable constraints are
   added as described with the notation $\existsvars{x_1 \ldots x_n} t \in
   b\,.\, \phi$ and $\forall x_1 \ldots x_n| t \in b\,.\, \phi$ on
   p.~\pageref{page:existsvars}.

%% file: app-derivations.tex
\section{Proofs of polytime admissibility}
\label{app:derivations}

The goal of this section is to prove most claims of polytime
admissibility made in the body of the paper, crucially those of
Subsection~\ref{subsec:proof:ain}.
Recall that a rule
\[\dfrac{\Theta \vdash \Delta}{\Theta' \vdash \Delta'}\]
is \emph{polytime admissible} if we can compute in polynomial time 
a proof of the conclusion $\Theta' \vdash \Delta'$ from a proof of
the antecedent $\Theta \vdash \Delta$.

Throughout this section we deal with the {\focused} proof system of
Figure \ref{fig:dzcalc-1sided}.

\subsection{Standard rules}

Here we collect some useful standard sequent calculi rules, which are all polytime
admissible in our system. 
The arguments for these rules are straightforward.

\begin{lem}
  \label{lem:wkadm}
  The following weakening rule is polytime admissible:
\[
  \dfrac{\Theta \vdash \Delta}{\Theta' \vdash \Delta, \Delta'}
\]
\end{lem}

\begin{lem}
  \label{lem:wedgeprojadm}
  The following inference, witnessing the invertibility of the $\wedge$ rule, is polytime admissible for both $i \in \{1,2\}$:
\[
  \dfrac{\Theta \vdash \phi_1 \wedge \phi_2, \Delta}{\Theta \vdash \phi_i, \Delta}\]
\end{lem}

\begin{lem}
  \label{lem:allinvadm}
  The following, witnessing the invertibility of the $\forall$ rule, is polytime admissible:
  \[ \dfrac{\Theta \vdash \forall x \in t. \phi, \Delta}{\Theta, x \in t \vdash \phi, \Delta}\]
\end{lem}

\begin{lem}
  \label{lem:substitutivity}
  The following substitution rule is polytime admissible:
  \[\dfrac{\Theta \vdash \Delta}{\Theta[t/x] \vdash \Delta[t/x]}
  \]
\end{lem}

\subsection{Admissibility of generalized congruence}

Recall the admissibility claim concerning the rule related to congruence:

\begin{restatable}{lem}{gencongruence}
  \label{lem:gencongruence}
The following generalized congruence rule is polytime admissible:
\[ \dfrac{\Theta[t/x,t/y]\seq \Delta[t/x,t/y]}{\Theta[t/x,u/y] \seq \neg (t \equiv u), \Delta[t/x,u/y]}
  \]
\end{restatable}

Recall that in a two-sided reading  of this, the hypothesis
is 
$t \equiv u;\Theta[t/x,u/y] \seq \Delta[t/x,u/y]$.
So the rule says that if we $\Theta$ entails $\Delta$ where both
contain $t$, then if we assume $t \equiv u$ and substitute some occurrences of $t$ with
$u$ in $\Theta$, we can conclude the corresponding substitution of $\Delta$.

To prove Lemma \ref{lem:gencongruence} in the case where the terms $t$ and $u$ are of type
$\sett(T)$, we will need
a more general statement.
We are going to generalize the statement to treat tuples of terms and
use $\pospolarity$ formulas instead of $\neg (t \equiv u)$ to simplify the inductive
invariant.

Given two terms $t$ and $u$ of type $t$, define by induction the set of formulas
$\cE_{t,u}$:
\begin{itemize}
    \item If $T = \ur$, then $\cE_{t,u}$ is $t \neq_\ur u$ 
    \item If $T = T_1 \times T_2$, then
      $\cE_{t,u}$ is $\cE_{\pi_1(t),\pi_1(u)}, \cE_{\pi_2(t),\pi_2(u)}$
    \item If $T = \sett(T')$, $\cE_{t,u}$ is $\neg (t \subseteq_T u), \neg(u \subseteq_T t)$
\end{itemize}

The reader can check that $\cE_{t,u}$ is essentially $\neg (t \equiv u)$.

\begin{lem}
\label{lem:cEder}
The following rule is polytime admissible:
  \[
    \dfrac{\Theta \vdash \Delta, \cE_{t,u}}{\Theta \vdash \Delta, \neg (t \equiv u)}
  \]
\end{lem}
\begin{proof}
  Straightforward induction over $T$.
\end{proof}

Since we will deal with multiple equivalences, we will adopt vector notation
$\vv t = t_1, \ldots, t_n$ and $\vv x = x_1, \ldots, x_n$ for lists of terms and
variables. Call $\cE_{\vv t, \vv u}$ the union of the $\cE_{t_i,u_i}$.
We can now state our more general lemma:

\begin{lem}
  \label{lem:gencongruencenary}
  The following generalized $n$-ary congruence rule for set variables is polytime admissible:
  \[
    \dfrac{\Theta[\vv t/\vv x, \vv t /\vv y]\seq \Delta[\vv t/\vv x,\vv t/\vv y]}{
   \Theta[\vv t/\vv x, \vv u/\vv y]\seq
  \Delta[\vv t/\vv x, \vv u/\vv y], \cE_{\vv t, \vv u}
}\]
\end{lem}

\begin{proof}[Proof of Lemma~\ref{lem:gencongruencenary}]
We proceed by induction over the proof of $\Theta[\vv t/\vv x, \vv t/\vv y] \seq \Delta[\vv t/\vv x,\vv u/\vv y]$,
\begin{itemize}
\item If the last rule applied is the $=$ rule, i.e. we have
  \[ \dfrac{}{\Theta[\vv t/\vv x, \vv t/\vv y] \vdash
    a[\vv t/\vv x, \vv t/\vv y] =_\ur b[\vv t/\vv x, \vv t/\vv y],
    \Delta[\vv t/\vv x, \vv t/ \vv y]}\]
    and $a[\vv t/\vv x, \vv t/\vv y] = b[\vv t/\vv x, \vv t/\vv y]$, with $a,b$ variables.
    Now if $a$ and $b$ are equal, or if they belong both to either $\vv x$ or $\vv y$,
    it is easy to derive the desired conclusion with a single application of the $=$ rule.
    Otherwise, assume $a = x_i$ and $b = y_j$ (the symmetric case is handled similarly).
    In such a case, we have that $\cE_{\vv t, \vv u}$ contains $t_i \neq_\ur t_j$. So the desired
    proof
  \[ \dfrac{}{\Theta[\vv t/\vv x, \vv t/\vv y] \vdash
    t_i =_\ur u_j,
    \Delta[\vv t/\vv x, \vv t/ \vv y], \cE_{\vv t, \vv u}, t_i \neq_\ur t_j}\]
    follows from the polytime admissibility of the axiom rule.
\item Suppose the last rule applied is the $\top$ rule:
  \[ \dfrac{}{\Theta[\vv t/\vv x, \vv t/\vv y] \vdash \top, \Delta[\vv t/\vv x, \vv t/ \vv y]}\]
    Then we do not need to  apply the induction hypothesis. Instead we can immediately
apply  the $\top$ rule to obtain
    \[ \dfrac{}{\Theta[\vv t/\vv x, \vv u/\vv y] \vdash \top, \Delta[\vv t/\vv x, \vv u/ \vv y], \cE_{\vv t, \vv u}}\]

\item If the last rule applied is the $\wedge$ rule
\begin{small}
  \[ \dfrac{
    \Theta[\vv t/\vv x, \vv t/\vv y] \vdash \phi_1[\vv t/\vv x, \vv t/\vv y], \Delta[\vv t/\vv x, \vv t/ \vv y]
\qquad
    \Theta[\vv t/\vv x, \vv t/\vv y] \vdash \phi_2[\vv t/\vv x, \vv t/\vv y], \Delta[\vv t/\vv x, \vv t/ \vv y]
  }{
    \Theta[\vv t/\vv x, \vv t/\vv y] \vdash (\phi_1 \wedge \phi_2)[\vv t/\vv x, \vv t/\vv y], \Delta[\vv t/\vv x, \vv t/ \vv y]
}
    \]
\end{small}
    then the induction hypothesis gives us proofs of
    \[    \Theta[\vv t/\vv x, \vv u/\vv y] \vdash \phi_i[\vv t/\vv x, \vv u/\vv y], \Delta[\vv t/\vv x, \vv u/ \vv y], \cE_{\vv t, \vv u} \]
    for both $i \in \{1,2\}$. So we can apply the $\wedge$ rule to conclude that we have
    \[    \Theta[\vv t/\vv x, \vv u/\vv y] \vdash (\phi_1 \wedge \phi_2)[\vv t/\vv x, \vv u/\vv y], \Delta[\vv t/\vv x, \vv u/ \vv y], \cE_{\vv t, \vv u} \]
    as desired.
  \item The cases of the rules $\vee$,$\forall$ and $\times_\eta$ are equally straightforward and left to the reader.
  \item Now, let us handle the case of the $\exists$ rule. 
    \[ \dfrac{\Theta, t \in u \vdash \phi[t/x], \Delta^\pospolaritysuper}{\Theta, t \in u \vdash \exists x \in u. \; \phi, \Delta^\pospolaritysuper }\]
    So assume that $z$ is fresh wrt $\vv x, \vv y, \vv t, \vv u, a, b$ and
    that the last step of the proof is
\begin{tiny}
  \[ \dfrac{
    \Theta[\vv t/\vv x, \vv t/\vv y], a[\vv t/\vv x, \vv t/\vv y] \in b[\vv t/\vv x, \vv t/\vv y]
    \vdash \phi[a/z][\vv t/\vv x, \vv t/\vv y],
    \exists z \in c[\vv t/\vv x, \vv t/\vv y].\; \phi[\vv t/\vv x, \vv t/\vv y],
    \Delta[\vv t/\vv x, \vv t/ \vv y]
  }{
    \Theta[\vv t/\vv x, \vv t/\vv y], a[\vv t/\vv x, \vv t/\vv y] \in b[\vv t/\vv x, \vv t/\vv y]
    \vdash
    \exists z \in c[\vv t/\vv x, \vv t/\vv y].\; \phi[\vv t/\vv x, \vv t/\vv y],
    \Delta[\vv t/\vv x, \vv t/ \vv y]
    }
    \]
\end{tiny}
    with $b[\vv t/\vv x, \vv t/\vv y] = c[\vv t/\vv x, \vv t/\vv y]$.
    Set $a' = a[\vv t/\vv x, \vv u/\vv y]$, $b' = b[\vv t/\vv x, \vv u/\vv y]$, $c'=c[\vv t/\vv x, \vv u/\vv y]$.
    We have three subcases:
    \begin{itemize}
    \item If we have that $b' = c'$,
    using the induction hypothesis, we have a proof of
    \[    \Theta[\vv t/\vv x, \vv u/\vv y], a' \in b'
    \vdash \phi[a/z][\vv t/\vv x, \vv u/\vv y],
    \exists z \in c'.\; \phi[\vv u/\vv x, \vv u/\vv y],
    \Delta[\vv t/\vv x, \vv u/ \vv y], \cE_{\vv t, \vv u}\]
    we can simply apply an $\exists$ rule to that proof and we are done.
  \item Otherwise, if we have that $b' = t_i$ and $c' = u_j$ for some $i,j \le n$.
    In that case,
    extending the tuples $\vv t$ and $\vv u$ with $a'$ and a fresh variable
        $z'$ (the substitutions under consideration would be,
        we can apply the induction hypothesis to obtain a proof of
\begin{small}
\begin{align*}    \Theta[\vv t/\vv x, \vv u/\vv y, a'/z], a' \in t_i, z' \in u_j
    \vdash \\
\phi[\vv t/\vv x, \vv u/\vv y,z'/z],
    \exists z \in u_j.\; \phi[\vv u/\vv x, \vv u/\vv y],
        \Delta[\vv t/\vv x, \vv u/ \vv y], \cE_{\vv t, \vv u}, \cE_{a',z'}
\end{align*}
\end{small}
      Note that $\cE_{\vv t, \vv u}$
      contains an occurrence of $\neg (t_i \subseteq u_j)$, which expands to
        $\exists z \in t_i. \forall z' \in u_j. \neg (z \equiv z')$, so we can construct the
        partial derivation
\begin{tiny}
    \[
      \text{
        \AXC{$
    \Theta[\vv t/\vv x, \vv u/\vv y, a'/z], a' \in t_i, z' \in u_j
    \vdash \phi[\vv t/\vv x, \vv u/\vv y,z'/z'],
    \exists z \in u_j.\; \phi[\vv u/\vv x, \vv u/\vv y],
        \Delta[\vv t/\vv x, \vv u/ \vv y], \cE_{\vv t, \vv u}, \cE_{a',z'}$
        }
        \myLeftLabel{$\exists$}
        \UIC{$
    \Theta[\vv t/\vv x, \vv u/\vv y], a' \in t_i, z' \in u_j
    \vdash \exists z \in u_j. \phi[\vv t/\vv x, \vv u/\vv y],
        \Delta[\vv t/\vv x, \vv u/ \vv y], \cE_{\vv t, \vv u}, \cE_{a',z'}$
        }
        \myLeftLabel{Lemma~\ref{lem:cEder}}
        \UIC{$
    \Theta[\vv t/\vv x, \vv u/\vv y, a'/z], a' \in t_i, z' \in u_j
    \vdash 
    \exists z \in u_j.\; \phi[\vv u/\vv x, \vv u/\vv y],
        \Delta[\vv t/\vv x, \vv u/ \vv y], \cE_{\vv t, \vv u}, a' \equiv z'$}
        \myLeftLabel{$\forall$}
        \UIC{$
    \Theta[\vv t/\vv x, \vv u/\vv y, a'/z], a' \in t_i
    \vdash
    \exists z \in c'.\; \phi[\vv u/\vv x, \vv u/\vv y],
        \Delta[\vv t/\vv x, \vv u/ \vv y], \cE_{\vv t, \vv u}, \forall z' \in u_j. a' \equiv z'$}
        \myLeftLabel{$\exists$}
        \UIC{$
    \Theta[\vv t/\vv x, \vv u/\vv y, a'/z], a' \in t_i
    \vdash
    \exists z \in c'.\; \phi[\vv t/\vv x, \vv u/\vv y],
        \Delta[\vv t/\vv x, \vv u/ \vv y], \cE_{\vv t, \vv u}$}
      \DisplayProof
      }
    \]
\end{tiny}
    whose conclusion matches what we want.
    \item Otherwise, we are in a similar case where $b' = u_j$ and $c' = t_i$.
      We proceed similarly, except that we use the formula $\neg (u_j \subseteq t_i)$
        of $\cE_{t_i,u_j}$ instead of $\neg(t_i \subseteq u_j)$.
    \item For the rule $\times_\beta$, which has general shape
      \[\dfrac{\Theta[z_i/z] \vdash \Delta[z_i/z]}{\Theta[\pi_i(\tuple{z_1,z_2})/z] \vdash \Delta[\pi_i(\tuple{z_1,z_2})/z]}\]
      we can assume, without loss of generality, that $z$ occurs only once in $\Theta, \Delta$.
      Let us only sketch the case where $z$ occurs in a formula $\phi$ and the rule has shape
      \[\dfrac{\Theta \vdash \phi[z_i/z], \Delta
      }{\Theta \vdash \phi[\pi_i(\tuple{z_1,z_2})/z], \Delta}\]
      We have that $\phi[\pi_i(\tuple{z_1,z_2})/z)]$ is also of the shape $\psi[\vv t/\vv x, \vv u/\vv y]$ in our situation.
      We can also assume without loss of generality that each variable in $\vv x$ and $\vv y$ occur each a single time in $\Theta, \Delta$.
      Now if we have a couple of situations:
      \begin{itemize}
        \item If the occurrence of $z$ do not interfere with the substitution $[\vv t/\vv x, \vv u/\vv y]$, i.e., there is a formula $\theta$ such that
          \[\phi[\pi_i(\tuple{z_1,z_2})/z] = \psi[\vv t/\vv x, \vv u] = \theta[\vv t/\vv x, \vv u/\vv y, \tuple{z_1,z_2}/z]\]
            we can simply apply the induction hypothesis on the subproof and conclude with one application of $\times_\eta$.
          \item If we have that $z$ clashes with a variable of $\vv x, \vv y$, say $x_j$, but that $t_j = v[\pi_i(\tuple{z_1,z_2}/x_j)]$ for some term $v$.
            Then we can apply the induction hypothesis with the matching tuples of terms $\vv t, x_i$ and $\vv u, t_i$ and conclude by applying the $\beta$ rule.
          \item Otherwise, the occurrence of $z$ does interfere with the substitution in such a way that we have,
            say $t_j = \tuple{z_1, z_2}$. Then we can apply the induction hypothesis on the subproof with the matching tuples of terms $\vv x, z_i$ and $\vv y, \pi_i(u_j)$
            and conclude by applying the $\beta$ rule.\qedhere
      \end{itemize}
    \end{itemize}
\end{itemize}
\end{proof}

One easy consequence of the above is Lemma \ref{lem:gencongruence}:

\begin{proof}[Proof of Lemma~\ref{lem:gencongruence}]
Combine Lemma~\ref{lem:gencongruencenary} and Lemma~\ref{lem:cEder}.
\end{proof}

Another consequence is the following corollary, which will be used later in this section:

\begin{restatable}{cor}{memctxadm}
  \label{cor:memctxadm}
The following rule is polytime admissible:
\[
  \dfrac{\Theta, t \in u \seq \Delta}{\Theta \seq \neg t\memmac u, \Delta}
\]
\end{restatable}

\begin{proof}
Recall that $t \memmac u$ expands to $\exists x \in u. \; x \equiv t$,
so that $\neg t\memmac u$ is $\forall x \in u. \; \neg (x \equiv t)$.
So we have

\smallskip
{\centering
      \AXC{$\Theta, t \in u \vdash \Delta$}
      \myLeftLabel{Lemma~\ref{lem:gencongruence}}
      \UIC{$\Theta, x \in u \vdash \neg (x \equiv t), \Delta$}
      \myLeftLabel{$\forall$}
      \UIC{$\Theta \vdash \neg t \memmac u, \Delta$}
\DisplayProof\par\vspace{-10pt}\qedhere}
\end{proof}

\subsection{Proof of  Lemma \ref{lem:effnestedproof-down}}

We now recall the claim of admissibility concerning rules
for ``moving down in an equivalence''. These will make use of
the notation for quantifying on a path below an object, defined
in the body in  Definition \ref{def:subtypeocc}.

\effnestedproofdown*

\begin{proof}
 First, let us consider the simpler case where $p = \membersub$.
We proceed by induction over the input proof of
  $\Theta \vdash \Delta, \exists r' \inq o'. \; r \equiv_{\sett(T')} r'$
and make a case distinction
  according to which rule was applied last. All cases are straightforward, save for one: when a $\exists$ rule is applied on the formula  $\exists r'
  \inq o'. \; r \equiv_{\sett(T')} r'$.
Let us only detail that one.

In that case, the last step has shape
  \[
    \dfrac{
      \Theta, w \in o' \vdash \Delta^\pospolaritysuper, r \equiv_{\sett(T')} w, \exists r' \inq o'. \; r \equiv_{\sett(T')} r'
    }{
\Theta, w \in o' \vdash \Delta^\pospolaritysuper, \exists r' \inq o'. \; r \equiv_{\sett(T')} r'
    }\]

Now observe that $r \equiv_{\sett(T')} w$ is not $\pospolarity$, since if you unfold the macros it begins
with a universal. Thus the final rule that is applied in the proof witnessing the hypothesis sequent
 cannot be the $\exists$ rule.
We can thus infer that the  expanded 
proof tree ends with:
  \[\small
    \text{
      \AXC{\hspace{-6pt}
        $\Theta, w \in o', z \in r \vdash \Delta^\pospolaritysuper, z \memmac w, \exists r' \inq o'. \; r \equiv_{\sett(T')} r'$\hspace{-9pt}
      }
      \UIC{\hspace{-6pt}
        $\Theta, w \in o' \vdash \Delta^\pospolaritysuper, r \subseteq w, \exists r' \inq o'. \; r \equiv_{\sett(T')} r'$\hspace{-6pt}
      }
      \AXC{
        $\vdots$
      }
      \UIC{\hspace{-6pt}
        $\Theta, w \in o' \vdash \Delta^\pospolaritysuper, w \subseteq r, \exists r' \inq o'. \; r \equiv_{\sett(T')} r'$\hspace{-9pt}
      }
      \BIC{
        $\Theta, w \in o' \vdash \Delta^\pospolaritysuper, r \equiv_{\sett(T')} w, \exists r' \inq o'. \; r \equiv_{\sett(T')} r'$
      }
      \UIC{
$\Theta, w \in o' \vdash \Delta^\pospolaritysuper, \exists r' \inq o'. \; r \equiv_{\sett(T')} r'$
      }
      \DisplayProof
    }
\]
In particular, we have a strictly smaller subproof of 
\[\Theta, w \in o', z \in r \vdash \Delta^\pospolaritysuper, z \memmac w, \exists r' \inq o'. \; r \equiv_{\sett(T')} r'\]
Recall that $\memmac$ is a macro: see Definition \ref{def:macros}.
Applying the induction hypothesis, we get a proof of
  \[\Theta, w \in o', z \in r \vdash \Delta^\pospolaritysuper, z \memmac w, \exists z' \inq_{\membersub, \membersub} o'. \; z \equiv_{\sett(T')} z'\]
Applying the $\exists$-rule
  we can then obtain a proof with conclusion
  \[
    \Theta, w \in o', z \in r \vdash \Delta^\pospolaritysuper, \exists r \inq o'. \; z \memmac r, \exists z' \inq_{\membersub, \membersub} o'. \; z \equiv_{\sett(T')} z'\]
  which concludes our argument, since $\exists r \inq_p o'. \; z \memmac r$ and 
$\exists z' \inq_{\membersub, \membersub} o'. \; z \equiv_{\sett(T')} z'$ are syntactically the
same.

For the more complex case where $p$ is non-trivial, we can conduct a similar
argument, at the cost of making  the induction hypothesis more elaborate. We would
prove that the following rule is polytime admissible
\[
\dfrac{\Theta, \Theta_1, \ldots, \Theta_n \vdash \Delta, \phi_1, \ldots, \phi_n}{\Theta, \Theta_1, \ldots, \Theta_n \vdash \Delta, \exists z' \in_{\membersub, p} o'.
  \; z \equiv_{\sett(T')} z'}\]

Here $n$ is any natural number and for each $i \leq n$,  we have a decomposition of $p$ as a concatenation $p_i,p'_i$, where $p'_i$ is non-empty and multiple $p_i$ can be the same.
From  this decomposition we define $\phi_i = \exists z' \inq_{p'_i} o_i. \; z \equiv_{\sett(T')} z'$ (with $o_i = o'$ iff $p_i = \varepsilon$) and
$\Theta_i$ is a context witnessing that $o_i \inq_{p_i} o'$. Formally,
$\Theta_i = \Theta(p_i, o_i, o')$ with $\Theta(\varepsilon, x, y) = \cdot$, $\Theta(\membersub, x, y) = x \in y$,
$\Theta((p,\membersub), x, y) = \Theta(p, x, z), z \in y$ and $\Theta((p,j), x, y) = \Theta(p, x, \pi_j(y))$ when $j \in \{1,2\}$.
The induction over a proof of the premise can be carried out in an analogous way to show admissibility, and the only non-trivial case, when a rule interacts
with one of the $\phi_i$ that has $p'_i = \membersub$, is handled in the same way as the non-trivial case for $p = \membersub$.
\end{proof}

\subsection{Proof of  Lemma \ref{lem:equivsettoequivalence}}

\equivsettoequivalence*

\begin{proof}
Much like what happened in the proof of~\ref{lem:effnestedproof-down}, the proof
can be done by induction over the shape of the input derivation when $p = \membersub$.
When this is not the case, we can generalize our inductive hypothesis in a similar
way and use the same ideas. So for the sake of legibility, let us focus on
the case where $p = \membersub$ and go through the induction, making a case distinction according
to what was the last rule applied.

All cases are trivial, except for the case of the rule $\exists$, where
  $\exists r' \inq o'. r \equiv_{\sett(T')} r'$ is the principal formula. So let us focus on that one.

In that case, the proof necessarily has shape
  \[\scalebox{0.9}{\text{
\AXC{$\Theta, y \in w \vdash \Delta^\pospolaritysuper, y \memmac r, \exists r' \inq o'. \; r \equiv_{\sett(T')} r'$}
\myLeftLabel{$\forall$}
\UIC{$\Theta \vdash \Delta^\pospolaritysuper, w \subseteq_{T'} r, \exists r' \inq o'. \; r \equiv_{\sett(T')} r'$}
\AXC{$\Theta, x \in r \vdash \Delta^\pospolaritysuper, x \memmac w, \exists r' \inq o'. \; r \equiv_{\sett(T')} r'$}
\myLeftLabel{$\forall$}
\UIC{$\Theta \vdash \Delta^\pospolaritysuper, r \subseteq_{T'} w, \exists r' \inq o'. \; r \equiv_{\sett(T')} r'$}
  \myLeftLabel{$\wedge$}
\BIC{$\Theta \vdash \Delta^\pospolaritysuper, r \equiv_{\sett(T')} w, \exists r' \inq o'. \; r \equiv_{\sett(T')} r'$}
  \myLeftLabel{$\exists$}
\UIC{$\Theta \vdash \Delta^\pospolaritysuper, \exists r' \inq o'. \; r \equiv_{\sett(T')} r'$}
\DisplayProof
  }}\]
Applying the induction hypothesis to the leaves of that proof, we obtain two proofs
of
\begin{align*}
  \Theta, y \in w \vdash \Delta^\pospolaritysuper, y \memmac r, \exists r' \inq o'. \forall z \in a, \; z \memmac r \equi z \memmac r' 
\end{align*}

  and
\begin{align*}
  \Theta, x \in r \vdash \Delta^\pospolaritysuper, x \memmac w, \exists r' \inq o'. \forall z \in a, \; z \memmac r \equi z \memmac r'
\end{align*}
  so we can conclude using the admissibility of weakening and Corollary~\ref{cor:memctxadm} twice and replaying the $\wedge/\forall/\exists$ steps in the
  appropriate order (a branch of the proof derivation is elided below for lack of horizontal space):
\begin{align*}
  \text{
  \raisebox{52pt}{
  \begin{minipage}{13cm}
  \AXC{
    $\Theta, y \in w \vdash \Delta^\pospolaritysuper, y \memmac r, \exists r' \inq o'. \forall z \in a, \; z \memmac r \equi z \memmac r'$}
  \myLeftLabel{Corollary~\ref{cor:memctxadm}}
  \UIC{
    $\Theta \vdash \Delta^\pospolaritysuper, \neg (y \memmac w), y \memmac r, \exists r' \inq o'. \forall z \in a, \; z \memmac r \equi z \memmac r'$}
  \myLeftLabel{$\vee$}
  \UIC{$
    \Theta \vdash \Delta^\pospolaritysuper, {y \memmac w} \imp {y \memmac r}, \exists r' \inq o'. \forall z \in a, \; z \memmac r \equi z \memmac r'$}
  \AXC{ $\vdots$ }
   \BIC{$
    \Theta \vdash \Delta^\pospolaritysuper, {y \memmac r} \equi {y \memmac w}, \exists r' \inq o'. \forall z \in a, \; z \memmac r \equi z \memmac r'$} 
  \myLeftLabel{Lemma~\ref{lem:wkadm}}
  \UIC{$
    \Theta, y \in a \vdash \Delta^\pospolaritysuper, {y \memmac r} \equi {y \memmac w}, \exists r' \inq o'. \forall z \in a, \; z \memmac r \equi z \memmac r'$} 
  \myLeftLabel{$\forall$}
  \UIC{$
    \Theta \vdash \Delta^\pospolaritysuper, \forall z \in a. \; {z \memmac r} \equi {z \memmac w}, \exists r' \inq o'. \forall z \in a, \; z \memmac r \equi z \memmac r'$} 
  \myLeftLabel{$\exists$}
  \UIC{$
    \Theta \vdash \Delta^\pospolaritysuper, \exists r' \inq o'. \forall z \in a, \; z \memmac r \equi z \memmac r'$} 
\DisplayProof
\end{minipage}}
  }\tag*{\qedhere}
\end{align*}
\end{proof}

%% file: app-mainthm.tex
\section{Proof of the main theorem for non-set types}
\label{app:mainthm}
  
Recall again our main result:

\thmmainset*

 In the body of the paper, we gave a proof for the case where the type of the defined object is $\sett(T)$ for any $T$.
We now discuss the remaining cases: the base case and the inductive
case for product types.

So assume we are given an implicit
  definition $\varphi(\vv i, \vv a, o)$ and a {\focused} witness, and proceed by induction over the type $o$:
  \begin{itemize}
    \item If $o$ has type $\unit$, then, since there is only one inhabitant in
      type $\unit$, then we can take our explicit definition to be the corresponding $\nrc$ expression $\unitobj$.
    \item If $o$ has type $\ur$, then using interpolation on the entailment $\varphi(\vv i, \vv a, o) \imp (\varphi(\vv i, \vv a', o') \imp o =_\ur o')$,
      we obtain $\theta(\vv i, o)$ with $\varphi(\vv i, \vv a, o) \imp \theta(\vv i, \vv a, o)$ and $\theta(\vv i, o) \wedge \varphi(\vv i, \vv a, o') \imp
      o =_\ur o'$. 
      But then we know that $\phi$ implies
$o$ is a subobject of $\vv i$: otherwise we could find a model that contradicts
the entailment. 
There is an $\nrc$ definition $A(\vv i)$ that collects all of the $\ur$-elements 
lying beneath $\vv i$. 
      We can then take $E(\vv i) = \nrcget_T(\{ x \in A(\vv i) \mid \theta(\vv i, x)\})$ as our
      $\nrcwget$ definition of $o$.  The correctness of $E$ follows
from the properties of $\theta$ above.
    \item If $o$ has type $T_1 \times T_2$, recalling the definition of $\equiv_{T_1 \times T_2}$, we have a derivation of
      \[ \varphi(\vv i, \vv a, o) \wedge \varphi(\vv i, \vv a, o') \vdash
      \pi_1(o) \equiv_{T_1} \pi_1(o') \wedge
      \pi_2(o) \equiv_{T_2} \pi_2(o')\] 
      By Lemma~\ref{lem:wedgeprojadm}, we have proofs of
      \[
         \varphi(\vv i, \vv a, o) \wedge \varphi(\vv i, \vv a, o') \vdash
        \pi_i(o) \equiv_{T_i} \pi_i(o') \]
      for $i \in \{1,2\}$.
      Take $o_1$ and $o_2$ to be fresh variables of types $T_1$ and $T_2$. Take $\tilde\phi(\vv i, \vv a, o_1, o_2)$ to be $\phi(\vv i, \vv a, \tuple{o_1,
      o_2})$. By substitutivity (the admissible rule given
by Lemma~\ref{lem:substitutivity}) and applying the $\times_\beta$ rule, we have {\focused} proofs of
      \[ \varphi(\vv i, \vv a, \tuple{o_1,o_2}) \wedge \varphi(\vv i, \vv a, \tuple{o_1', o_2'}) \vdash
      o_i \equiv_{T_i} o'_i 
      \]
      We can apply our inductive hypothesis to obtain a definition $E\IH_i(\vv i)$ for both $i \in \{1,2\}$.
      We can then take our explicit definition to be $\tuple{E\IH_1(\vv i),E\IH_2(\vv i)}$.
  \end{itemize}

%% file: app-paramfo.tex
\section{Variant of Parameter Collection Theorem,\texorpdfstring{\\}{} 
Theorem \ref{thm:mainflatcasebounded}, for parameterized definability
in first-order logic}
\label{app:paramfo}

Our paper has focused on the setting of nested relations, phrasing our results
in terms of the language $\nrc$. We indicated in the conclusion of the paper that
there is a variant of the $\nrc$ parameter collection theorem, Theorem
\ref{thm:mainflatcasebounded}, for the broader context of first-order logic.
In fact, this  first-order version of the result provided the intuition
for the theorem.
In this section we present this variant.

\input{app-foproofsystem}

We now discuss our generalization of the $\nrc$ Parameter Collection Theorem from the body of the paper to this 
first-order setting.
The concept of explicit definition can
be generalized to \name{definition up to parameters and disjunction}: A family
of formulas $\chi_i(\vv z, \vv y, \vv r)$, $1 \leq i \leq n$, provides an
\name{explicit definition up to parameters and disjunction} of a formula
$\lambda(\vv z, \vv l)$ relative to a formula $\phi$ if
\begin{equation}
  \label{eq-def-pd}
\phi \entails \bigvee_{i=1}^n
\exists \vv y \forall \vv z\, (\lambda(\vv z, \vv l) \equi \chi_i(\vv z, \vv
y, \vv r)).\tag{$\star$}
\end{equation}
The entailment (\ref{eq-def-pd}) is considered with restrictions on the
predicates and variables permitted to occur in the $\chi_i$. In the simplest
case, $\lambda(\vv z,\vv l)$ is a positive literal $p(\vv z)$ with a predicate
that is permitted in $\phi$ but not in the $\chi_i$. The predicate~$p$ is then
said to be \name{explicitly definable up to parameters and disjunction} with
respect to $\phi$ \cite{ck}. 

The disjunction over
a finite family of formulas $\chi_i$ can be consolidated into a single
quantified biconditional as long as the domain 
has size at least $2$ in every model of $\phi$. 
Notice that if $\phi$ has only finite models, then by the
compactness theorem of first-order logic, the size of models must be bounded.
In such cases every formula $\lambda$ is definable 
with sufficiently many parameters.

We can now state our analog of the Parameter Collection Theorem,  Theorem~\ref{thm:mainflatcasebounded}.
\pagebreak
\begin{restatable}{thm}{mainflatcase} \label{thm:mainflatcase}
  Let $\phi$, $\psi$, $\lambda(\vv z, \vv l)$, and $\rho(\vv
  z, \vv y, \vv r)$ be first-order formulas such that
  \[\phi \land \psi
  \entails \exists \vv y \forall \vv z\, (\lambda(\vv z, \vv l) \equi \rho(\vv
  z, \vv y, \vv r)).\] Then there exist first-order formulas $\chi_i(\vv z, {\vv v}_i,
  {\vv c}_i)$, $1 \leq i \leq n$, such that
  \begin{enumerate}
  \item $\phi \land \psi \entails \bigvee_{i=1}^{n} \exists {\vv v}_i \forall
    \vv z (\lambda(\vv z, \vv l) \equi \chi_i(\vv z, {\vv v}_i, {\vv c}_i))$,
  \item ${\vv c}_i \subseteq (\FV(\phi) \cup \vv l) \cap (\FV(\psi) \cup \vv
    r)$,
  \item $\pred(\chi_i) \subseteq (\pred(\phi) \cup \pred(\lambda)) \cap
    (\pred(\psi) \cup \pred(\rho))$.
  \end{enumerate}

  \noindent
  Moreover, given a \folfocused proof of the precondition with the system of
  Figure~\ref{fig:rhflk}, a family of formulas $\chi_i$, $1 \leq i \leq n$, with
  the claimed properties can be computed in polynomial time in the size of the
  proof tree.
\end{restatable}

In the theorem statement, 
the free variables of $\lambda$ are
$\vv z$ and $\vv l$, and the free variables of $\rho$ are $\vv
z$, $\vv y$, and $\vv r$.
The precondition supposes an explicit
definition~$\rho$ of $\lambda$ up to parameters with respect to a conjunction
$\phi \land \psi$. The conclusion then claims that one can effectively compute
another definition of $\lambda$ with respect to $\phi \land \psi$ that is up
to parameters and disjunction and has a constrained signature: free variables
and predicates must occur in at least one of the ``left side'' formulas $\phi$
and $\lambda$ and also in at least one of the ``right side'' formulas $\psi$
and $\rho$. In other words, %
the theorem states that if
$\sig_L$ and $\sig_R$ are ``left'' and ``right'' signatures such that $\phi$
and $\lambda$ are over $\sig_L$, $\psi$ and $\rho$ are over $\sig_R$, and
$\rho$ provides an explicit definition of $\lambda$ up to parameters with
respect to $\phi \land \psi$, then one can effectively compute another
definition of $\lambda$ with respect to $\phi \land \psi$ that is up to
parameters and disjunction and is just over the intersection of the signatures
$\sig_L$ and $\sig_R$.

\input{app-proofparamfo}

%% file: app-foproofsystem.tex
We consider first-order logic with equality and without function symbols,
which also excludes nullary function symbols, that is, individual constants,
whose role is just taken by free individual variables. Specifically, we consider
first-order formulas with the following syntax
\[ \varphi, \psi ~~\bnfeq ~~ P(\vv x) \bnfalt \neg P(\vv x) \bnfalt x = y \bnfalt x \neq y
\bnfalt \top \bnfalt \bot
\bnfalt \varphi \wedge \psi \bnfalt \varphi \vee \psi
\bnfalt \forall x\, \varphi \bnfalt \exists x\, \varphi.
\]
Amongst the formulas given in the grammar above, those of the shape
$P(\vv x), \neg P(\vv x), x = y$ and $x \neq y$ are called \emph{literals}.
Literals come with a polarity: those of the shape $P(\vv x)$ and $x = y$ are
positive while the other, of the shape $\neg P(\vv x)$ or $x \neq y$ are
negative.

On top of this, we give some ``syntactic sugar''. We define $\neg \varphi$ by
induction over $\varphi$, dualizing every connective, including the
quantifiers, and removing doubled negation over literals. We define implication $\varphi
\imp \psi$ as an abbreviation of $\neg \varphi \vee \psi$ and bi-implication
$\varphi \equi \psi$ as an abbreviation of $(\varphi \imp \psi) \wedge (\psi
\imp \varphi)$. The set of free variables occurring in a formula~$\varphi$ is
denoted by $\FV(\varphi)$ and the set of predicates occurring in~$\varphi$ by
$\pred(\varphi)$.

Figure~\ref{fig:rhflk} shows our proof system for first-order logic. 
It is identical to a system from the prior literature.\footnote{
 G3c+\textsc{Ref}+\textsc{Repl} \cite{negri:plato:structural:2001,bpt},
in the one-sided form of GS3, discussed in Chapter~3 of \cite{bpt}, which
reduces the number of rules.} %
Like the \focused proof system we used in the body of the paper for $\deltazero$ formulas,
it is a $1$-sided calculus.
The formulas other than $\Gamma$ in the premise are the \name{active
  formulas} of the rule, while the \name{principal formulas} are the other
formulas in its conclusion. The complementary principal formulas in
\textsc{Ax} have to be literals. The replacement of symbols induced by
equality with \textsc{Repl} is only performed on negative literals.
\begin{figure}[h]
\
\input{rhflk_nosorts.tex}
\caption{One-sided sequent calculus for first-order logic with equality.}
\label{fig:rhflk}
\end{figure}

As in the body of the paper,   a proof tree or derivation is a tree whose nodes are
labelled with sequents, such that the labels of the children of a given node
and that of the node itself are the premises and conclusion, resp., of an
instance of a rule from Figure~\ref{fig:rhflk}. The \emph{conclusion} of a
proof tree is the sequent that labels its root. The proof system is closed
under cut, weakening and contraction. Closure under contraction in particular
makes it suited as basis for ``root-first'' proof search. Read in this
``bottom-up'' way, the $\exists$ rule states that a disjunction with an
existentially quantified formula can be proven if the extension of the
\pagebreak
disjunction by a copy of the formula where the formerly quantified
variable~$x$ is instantiated with an arbitrary variable~$t$ can be proven. The
existentially quantified formula is retained in the premise and may be used to
add further instances by applying~$\exists$ again in the course of the proof.

Soundness of the rules is straightforward. For example the $\exists$ rule
could be read as stating that if we deduce a disjunction in which one disjunct
is a formula $\phi$ with $t$ in it, then we can deduce the same disjunction
but with some occurrences of $t$ replaced in that disjunct with an
existentially quantified variable. Completeness of the proof system can also
be proven by a standard Henkin-style construction: indeed, since this
is really ordinary first-order logic, there are proofs in the literature
for systems that are very similar to this one~\cite{bpt,negri:plato:equality:1998}. 

Like our higher-level system in the body of the paper
(Figure~\ref{fig:dzcalc-2sided}) the first-order system in Figure~\ref{fig:rhflk}
does not impose any special constraints on the shape of the proof. We adapt
the extra conditions on contexts of the \focused system for $\deltazero$
formulas (Figure~\ref{fig:dzcalc-1sided}) to our first-order system. We
characterize a proof in the system of Figure~\ref{fig:rhflk} as
\name{\folfocused} if no application of \rulename{Ax}, $\top$, $\exists$,
\rulename{Ref}, \rulename{Repl} contains in its conclusion a formula whose
top-level connective is $\lor$, $\land$ or $\forall$.

The \folfocused property may be either incorporated directly into a
``root-first'' proof procedure by constraining rule applications or it may be
ensured by converting an arbitrary given proof tree to a \folfocused proof
tree with the same ultimate consequence.
This conversion is achieved by a straightforward adaption of the method for
the $\deltazero$ calculus from Appendix~\ref{app:focused},
Figure~\ref{fig:proc:normalize}. Only steps~(1) and~(2) of the method are
relevant here, since step~(3) is specifically for the pair terms of
$\deltazero$ formulas. In step~(2), instances of $\lor$, $\land$ or $\forall$
are permuted down over $\exists$, \rulename{Ref} and \rulename{Repl} according
to the generic permutation schemas of Lemma~\ref{lemma:permute:generic}. The
left sides of $\vdash$ in these schemas can be ignored since they represent
the $\in$-contexts for $\deltazero$ formulas. Finally, derivations of sequents
containing $\top$ are replaced by applications of rule~$\top$ and derivations
of sequents containing complementary literals $\phi, \lnot \phi$ by
applications of rule~$\rulename{Ax}$. The conversion, however, may increase
the proof size exponentially as discussed in
Section~\ref{sec:normalization:complexity}.

%% file: rhflk_nosorts.tex
\begin{mathpar}
\small

\inferrule*
[left={Ax}, right={\quad $\phi$ \textit{a positive literal}}]
{ }{\seq \Gamma, \phi, \neg \phi}
\and
\inferrule*
[left={$\top$}]
{ }{\seq \Gamma, \top}
\and
\inferrule*
[left={$\wedge$}]
{
\seq \Gamma, \phi_1
\and
\seq \Gamma, \phi_2
}{\seq \Gamma, \phi_1 \wedge \phi_2}
\and
\inferrule*
[left={$\vee$}]
{
\seq \Gamma, \phi_1, \phi_2
}{\seq \Gamma, \phi_1 \vee \phi_2}\\
\inferrule*
[left={$\forall$}, right={\quad $y \notin \FV(\Gamma,\forall x\;\phi)$}]
{\seq \Gamma, \phi[y/x]}{\seq \Gamma, \forall x\;\phi}
\and
\inferrule*
[left={$\exists$}]
{\seq \Gamma, \phi[t/x], \exists x\;\phi}{\seq \Gamma, \exists x\;\phi}\\
\inferrule*
[left={Ref}]
{\seq t \neq t, \Gamma}
{\seq \Gamma}
\and
\inferrule*
[left={Repl}, right={\quad $\phi$ \textit{a negative literal}}]
{\seq t\neq u, \phi[u/x], \phi[t/x], \Gamma}
{\seq t\neq u, \phi[t/x], \Gamma}
\end{mathpar}

%% file: app-proofparamfo.tex
We now prove Theorem~\ref{thm:mainflatcase} by induction on the depth of the
proof tree, generalizing the constructive proof method for Craig interpolation
often called Maehara's method \cite{takeuti:1987,bpt,smullyan}. To simplify
the presentation we assume that the tuples $\vv z$ and $\vv y$ in the theorem
statement each consist of a single variable $z$ and $y$, respectively. The
generalization of our argument to tuples of variables is straightforward.

To specify conveniently the construction steps of the family of formulas
$\chi_i$ we introduce the following concept: A \defname{pre-defining
  equivalence up to parameters and disjunction} (briefly $\pdepd$) for a
formula~$\lambda(z, \vv l)$ is a formula~$\delta$ built up from formulas of
the form $\forall z\, (\lambda(z,\vv l) \equi \chi(z, \vec p))$ (where the
left side is always the same formula $\lambda(z,\vv l)$ but the right sides
$\chi(z, \vec p)$ may differ) and a finite number of applications of
disjunction and existential quantification upon variables from the
vectors~$\vec p$ of the right sides. The empty disjunction $\bot$ is allowed
as a special case of a $\pdepd$. By rewriting with the equivalence $\exists
v\, (\delta_1 \lor \delta_2) \equiv \exists v\, \delta_1 \lor \exists v\,
\delta_2$, any $\pdepd$ for $\lambda$ can be efficiently transformed into the
form $\bigvee_{i=1}^n \exists \vs_i \forall z\, (\lambda(z, \vv l) \equi
\chi_i(z,\vs_i,\vec r_i))$ for some natural number $n \geq 0$. That is,
although a $\pdepd$ has in general not the syntactic form of the disjunction
of quantified biconditionals in the theorem statement (thus ``pre-''), it
corresponds to such a disjunction. The more generous syntax will be convenient
in the induction.
The sets of additional variables $\vv l$ and $\vec p$ in the biconditionals
$\lambda(z,\vv l) \equi \chi(z, \vec p)$ can overlap, but the overlap will be
top-level variables that never get quantified. Although we have defined the
notion of $\pdepd$ for a general $\lambda$, in the proof we just consider
$\pdepd$s for the formula $\lambda$ from the theorem statement.

For $\pdepd$s $\delta$ we provide analogs to $\FV$ and $\pred$ that only yield
free variables and predicates of $\delta$ that occur in a right side of its
biconditionals, which helps to express the restrictions by definability
properties that constrains the signature of exactly those right sides. Recall
that we refer of these right sides as subformulas $\chi(z,\vv p)$. For
$\pdepd$~$\delta$, define $\PREDD(\delta)$ as the set of the predicate symbols
that occur in a subformula~$\chi(z,\vv p)$ of $\delta$ and define
$\ADDLVARS(\delta)$ as the set of all variables that occur in a subformula
$\chi(z,\vv p)$ of $\delta$ and are free in $\delta$. In other words,
$\ADDLVARS(\delta)$ is the set of all variables~$p$ in the vectors $\vv p$ of
the subformulas $\chi(z,\vv p)$ that have an occurrence in $\delta$ which is
not in the scope of a quantifier $\exists p$. If, for example $\delta =
\bigvee_{i=1}^n \exists \vs_i \forall z\, (\lambda(z, \vv l) \equi
\chi_i(z,\vs_i,\vec r_i))$, then $\ADDLVARS(\delta)$ is the set of all
variables in the vectors $r_i$, for $1 \leq i \leq n$. 

To build up $\pdepd$s we provide an operation that only affects the right
sides of the biconditionals, a restricted form of existential quantification.
It is for use in interpolant construction to convert a free variable in right
sides that became illegal into an existentially quantified parameter. For
variables $p,y$ define $\delta\substdef{p}{y}$ as $\delta$ after substituting
all occurrences of $p$ that are within a right side $\chi(z,\vec p)$ and are
free in $\delta$ with $y$. Define $\exdef p\, \delta$ as shorthand for
$\exists y\, \delta\substdef{p}{y}$, where $y$ is a fresh variable. Clearly
$\delta \entails \exdef p\, \delta$ and $p$ has no free occurrences in $\exdef
p\, \delta$ that are in any of the right side formulas $\chi(z,\vec p)$, i.e.,
$p \notin \ADDLVARS(\exdef p\, \delta)$. Occurrences of $p$ in the left sides
$\lambda(z,\vv l)$, if $p$ is a member of $\vv l$, are untouched in $\exdef
p\, \delta$. If $p$ is not in $\vv l$, then $\exdef p\, \delta$ reduces to
ordinary existential quantification $\exists y\, \delta[y/p]$.

We introduce the following symbolic shorthand for the parametric definition on
the right side in the theorem's precondition.
  \[
  \begin{array}{lcl@{\hspace{1em}}lcl}
    \cG & \eqdef & \exists y \forall z\, (\lambda(z,\vv l) \equi \rho(z,y, \vv r)).\\
  \end{array}
  \]
  Note that our proof rules are such that if we have a proof (\folfocused or
  not) that our original top-level ``global'' parametric definition is implied
  by some formula, then every one-sided sequent in the proof must include that
  parametric definition in it. This is because the rules that eliminate a
  formula when read ``bottom-up'' cannot apply to that parametric definition,
  whose outermost logic operator is the existential quantifier. Thus, in our
  inductive argument, we can assume that $\cG$ is always present.

  We write
  \[\seq \SL\sep \SR\sep \cG \ipol{\theta}{\D},\]
  where $\seq \SL, \SR, \cG$ is a sequent, partitioned into three components,
  multisets $\SL$ and $\SR$ of formulas and the formula~$\cG$ from the
  theorem's hypothesis, $\theta$ is a formula and $\D$ is a $\pdepd$, to
  express that the following properties hold:
  \begin{enumerate}[{I}1.]
  \item \label{prop-r-theta} $\valid \SR \lor \theta$.
  \item \label{prop-l-theta} $\valid \lnot \theta \lor \SL \lor \D$.
  \item \label{prop-sig-theta} $\pred(\theta) \subseteq (\pred(\SL) \cup \ADDPREDSLAM)
    \cap (\pred(\SR) \cup \ADDPREDSRHO)$.
  \item \label{prop-fv-theta} $\FV(\theta) \subseteq
    (\FV(\SL) \cup \ADDVARSLAM ) \cap (\FV(\SR) \cup \ADDVARSRHO)$.
  \item \label{prop-form-d} $\D$ is a $\pdepd$ for $\lambda(z,\vv l)$.
  \item \label{prop-sig-d} $\PREDD(\D) \subseteq (\pred(\SL) \cup
    \ADDPREDSLAM) \cap (\pred(\SR) \cup \ADDPREDSRHO)$.
  \item \label{prop-fv-d} $\ADDLVARS(\D) \subseteq (\FV(\SL) \cup \ADDVARSLAM) \cap
    (\FV(\SR) \cup \ADDVARSRHO)$.
  \end{enumerate}

  For a given proof with conclusion $\seq \lnot \phi, \lnot \psi, \cG$,
  corresponding to the hypothesis $\phi \land \psi \entails \cG$ of the
  theorem, we show the construction of a formula $\theta$ and $\pdepd$ $\D$
  such that
  \[\seq \lnot \phi\sep \lnot \psi \sep \cG \ipol{\theta}{\D}.\]
  From properties~\ref{prop-r-theta}--\ref{prop-l-theta} and
  \ref{prop-form-d}-- \ref{prop-fv-d} it is then straightforward to read off
  that the formula $\bigvee_{i = 1}^{n} \exists \vv {v_i} \forall \vv z
  (\lambda(\vv z, \vv l) \equi \chi_i(\vv z,{\vv v}_i,{\vv r}_i))$ obtained
  from $\D$ by propagating existential quantifiers inwards is as claimed in
  the theorem's conclusion.

  Formula $\theta$ plays an auxiliary role in the induction. For the overall
  conclusion of the proof it a side result that is like a Craig interpolant
  of~$\psi$ and $\phi \imp D$, but slightly weaker syntactically constrained
  by taking $\lambda$ and $\rho$ into account: $\FV(\theta) \subseteq
  ((\FV(\phi) \cup \vv l) \cap (\FV(\psi) \cup \vv r))$ and $\pred(\theta)
  \subseteq (\pred(\phi) \cup \pred(\lambda)) \cap (\pred(\psi) \cup
  \pred(\rho))$.

  As basis of the induction, we have to show constructions of $\theta$ and
  $\D$ such that $\seq \SL\sep \SR\sep \cG \ipol{\theta}{\D}$ holds for
  \rulename{Ax} and \rulename{$\top$}, considering each possibility in which
  the principal formula(s) can be in $\SL$ or $\SR$. For the induction step,
  there are a number of subcases, according to which rule is last applied and
  which of the partitions $\SL$, $\SR$ or $\cG$ contain the principal
  formula(s). We first discuss the most interesting case, the induction step
  where $\cG$ is the principal formula. This case is similar to the most
interesting case in the $\nrc$ Parameter Collection Theorem, covered in the body
of the paper.

  \paragraph*{Case where the principal formula is $\cG$.}

   We now give more detail on the most complex case. If the principal formula
   of a conclusion $\seq \SL\sep\SR\sep\cG$ is $\cG$, then the rule that is
   applied rule must be~$\exists$. From the \folfocused property of the proof
   it follows that the derivation tree ending in $\seq \SL\sep\SR\sep\cG$ must
   have the following shape, for some $u \notin \FV(\SL,\SR,\cG)$ and $w \neq
   u$. Note that $u$ could be either a top-level variable from $\lambda$,
   i.e., a member of $\vv l$, or one introduced during the proof.
    \[
    \begin{array}{rccrc}
      \tabrulename{$\lor$} & \seq \SL, \lnot \lambda(u,\vv l), \SR,
      \rho(u,w,\vv r), \cG
      & \hspace*{2em} &
      \tabrulename{$\lor$}
      & \seq \SL, \lambda(u, \vv l), \SR, \lnot \rho(u,w,\vv r) , \cG\\\cmidrule{2-2}\cmidrule{5-5}
      \tabrulename{$\land$} & \seq \SL, \SR, \lambda(u, \vv l) \imp
      \rho(u,w,\vv r), \cG & &
      & \seq \SL, \lambda(u, \vv l), \SR, \rho(u,w,\vv r) \imp \lambda(u, \vv l), \cG\\\cmidrule{2-5}
      \multicolumn{5}{c}{
        \begin{array}{rc}
          \tabrulename{$\forall$} & \seq \SL, \SR, \lambda(u, \vv l) \equi
          \rho(u,w,\vv r), \cG\\\cmidrule{2-2}
          \tabrulename{$\exists$} & \seq \SL, \SR,
          \forall z\, (\lambda(z, \vv l) \equi \rho(z,w,\vv r)), \cG\\\cmidrule{2-2}
          & \seq \SL,\SR,\cG
      \end{array}}
    \end{array}
    \]
The important point is that the two ``leaves'' of the above tree are both sequents 
where we can apply our induction hypothesis.
    Taking into account the partitioning of the sequents at the bottom
    conclusion and the top premises in this figure, we can express the
    induction step in the form of a ``macro'' rule that specifies the
    how we constructed the required $\theta$ and $\D$ for the conclusion,  making use of
the $\theta_1, \D_1$ and $\theta_2, \D_2$ that we get by applying  the induction hypothesis to
each of the two premises.
    \[
    \begin{array}{c@{\hspace{1em}}cl}
      \seq \SL, \lnot \lambda(u, \vv l);\, \SR, \rho(u,w,\vv r);\, \cG \ipolnarrow{\theta_1}{\D_1} &
      \seq \SL, \lambda(u, \vv l);\, \SR, \lnot \rho(u,w,\vv r);\, \cG \ipolnarrow{\theta_2}{\D_2}
      \\\cmidrule{1-2}
      \multicolumn{2}{c}{\seq \SL \sep  \SR\sep \cG \ipol{\theta}{\D},}
    \end{array}
    \]
        where $u$ is as above. The values of $\theta$ and $\D$ -- the new formula and definition
that we are building --  will depend on occurrences of~$w$, and
we give their construction in cases below:
    
    \begin{enumerate}[(i)]
    \item
      \label{eq-case-g-u}
       If $w \notin \FV(\SL) \cup \ADDVARSLAM$ or $w \in \FV(\SR) \cup \ADDVARSRHO$, then
      \[\begin{array}{lcl}
      \theta & \eqdef & \forall u\, (\theta_1 \lor \theta_2).\\
         \D & \eqdef & \exdef u\, \D_1 \lor \exdef u\, \D_2 \lor \forall z\,
         (\lambda(z, \vv l) \equi \theta_2[z/u]).
       \end{array}\]
     \item \label{eq-case-g-wu}
       Else it holds that $w \in \FV(\SL) \cup \ADDVARSLAM$ and $w \notin \FV(\SR) \cup
       \ADDVARSRHO$. Then
       \[\begin{array}{lcl}
       \theta & \eqdef & \forall w \forall u\, (\theta_1 \lor \theta_2).\\
       \D & \eqdef & \exdef w\, (\exdef u\, \D_1 \lor \exdef u\, \D_2 \lor \forall z\,
       (\lambda(z, \vv l) \equi \theta_2[z/u])).
       \end{array}\]
    \end{enumerate}

    We now verify that $\seq \SL \sep \SR\sep \cG \ipol{\theta'}{\D'}$, that
    is, properties~\ref{prop-r-theta}--\ref{prop-fv-d}, hold. The proofs for
    the individual properties are presented in tabular form, with explanations
    annotated in the side column, where \name{IH} stands for \name{induction
      hypothesis}. We concentrate on the case~\ref{eq-case-g-u} and indicate
    the modifications of the proofs for case~\ref{eq-case-g-wu} in remarks,
    where we refer to the values of $\theta$ and $\D$ for that case in terms
    of the values for the case~\ref{eq-case-g-u} as $\forall w\, \theta$ and
    $\exdef w\, \D$. In the proofs of the semantic
    properties~\ref{prop-r-theta} and~\ref{prop-l-theta} we let sequents
    stand for the disjunction of their members.
    
    \smallskip
\noindent
Property~\ref{prop-r-theta}:
\prlReset{r-theta}

\smallskip\noindent
$
\begin{arrayprf}
  \prl{r-theta-1} & \valid \SR \lor \rho(u,w,\vv r) \lor \theta_1.
  & \textrm{IH}\\
  \prl{r-theta-2} & \valid \SR \lor \lnot \rho(u,w,\vv r) \lor \theta_2.
  & \textrm{IH}\\
  \prl{r-theta-3} & \valid \SR \lor \theta_1 \lor \theta_2.
  & \textrm{by~\pref{r-theta-1}, \pref{r-theta-2}}\\
 \prl{r-theta-4} & \valid \SR \lor \forall u\, (\theta_1 \lor \theta_1)
 & \textrm{by~\pref{r-theta-3} since } u \notin \FV(\SR)\\
 \prl{r-theta-5} & \valid \SR \lor \theta. &
 \textrm{by~\pref{r-theta-4} and def. of } \theta\\
\end{arrayprf}
$

\medskip
For case~\ref{eq-case-g-wu}, it follows from the precondition $w \notin
\FV(\SR)$ and step~\pref{r-theta-5} that $\valid \SR \lor \forall w\, \theta$.

\medskip

\noindent
Property~\ref{prop-l-theta}:
\prlReset{l-theta}

\smallskip\noindent
$
\begin{arrayprf}
  \prl{l-theta-i1} & \valid \lnot \theta_1 \lor \SL \lor \lnot \lambda(u, \vv l) \lor \D_1.
  & \textrm{IH}\\
  \prl{l-theta-i2} & \valid \lnot \theta_2 \lor \SL \lor \lambda(u, \vv l) \lor \D_2.
  & \textrm{IH}\\
  \prl{l-theta-3} &
  \valid \lnot (\theta_1 \lor \theta_2) \lor \SL \lor \lnot \lambda(u, \vv l) \lor \D_1 \lor
  \theta_2. & \textrm{by \pref{l-theta-i1}}\\
  \prl{l-theta-4} &
  \valid \lnot (\theta_1 \lor \theta_2) \lor \SL \lor \lnot \lambda(u, \vv l) \lor \D_1 \lor
  \D_2 \lor \theta_2. & \textrm{by \pref{l-theta-3}}\\
  \prl{l-theta-5} &
  \valid \lnot \forall u\, (\theta_1 \lor \theta_2) \lor \SL \lor \lnot \lambda(u, \vv l) \lor \D_1 \lor \D_2 \lor \theta_2. & \textrm{by \pref{l-theta-4}}\\
  \prl{l-theta-6} &
\valid \lnot \forall u\, (\theta_1 \lor \theta_2) \lor \lnot \theta_2 \lor \SL \lor \lambda(u, \vv l) \lor \D_1 \lor \D_2. & \textrm{by \pref{l-theta-i2}}\\
  \prl{l-theta-7} &
  \valid \lnot \forall u\, (\theta_1 \lor \theta_2) \lor \SL \lor \D_1 \lor \D_2
  \;\lor\\ & \hspace{2em} (\lambda(u, \vv l) \equi \theta_2).
  & \textrm{by \pref{l-theta-6}, \pref{l-theta-5}}\\
  \prl{l-theta-8} &
  \valid \lnot \forall u\, (\theta_1 \lor \theta_2) \lor \SL \lor \exdef u\,
  \D_1 \lor \exdef u\, \D_2
  \;\lor\\ & \hspace{2em} (\lambda(u, \vv l) \equi \theta_2).
  & \textrm{by \pref{l-theta-7}}\\
  \prl{l-theta-9} &
  \valid \lnot \forall u\, (\theta_1 \lor \theta_2) \lor \SL \lor \exdef u\,
  \D_1 \lor \exdef u\, \D_2
  \;\lor\\ & \hspace{2em} \forall z\, (\lambda(z, \vv l) \equi \theta_2[z/u]). &
  \textrm{by \pref{l-theta-8}}
   \textrm{ since } u \notin \FV(\SL,\lambda(z, \vv l))\\  
  \prl{l-theta-10} &
  \valid \lnot \theta \lor \SL \lor \D. 
  & \textrm{by \pref{l-theta-9}, defs. } \theta, \D\\
\end{arrayprf}
$

\smallskip
That $u \notin \FV(\lambda(z, \vv l))$ follows from the precondition $u \notin
\FV(\cG)$. It is used in step~\pref{l-theta-9} to justify that the substitution
$[z/u]$ has only to be applied to $\theta_2$ and not to $\lambda(z, \vv l)$ and, in
addition, to justify that $u \notin \FV(\exdef u\, \D_1)$ and $u \notin
\FV(\exdef u\, \D_2)$, which follow from $u \notin \FV(\lambda(z, \vv l))$ and the
induction hypotheses that property~\ref{prop-form-d} applies to~$\D_1$
and~$\D_2$.

For case~\ref{eq-case-g-wu}, it follows from step~\pref{l-theta-10} that
$\valid \lnot \forall w\, \theta \lor \SL \lor \exdef w\, \D$.

\medskip

\noindent
Property~\ref{prop-sig-theta}:
\prlReset{sig-theta}

\smallskip\noindent
$
\begin{arrayprf}
  \prl{sig-theta-1} & \pred(\theta_1) \subseteq
  (\pred(\SL,\lnot \lambda(u, \vv l)) \cup \ADDPREDSLAM) \cap
  (\pred(\SR,\rho(u,w,\vv r)) \cup \ADDPREDSRHO). & \textrm{IH}\\
  \prl{sig-theta-2} &
  \pred(\theta_2) \subseteq (\pred(\SL,\lambda(u, \vv l)) \cup \ADDPREDSLAM) \cap
  (\pred(\SR,\lnot \rho(u,w,\vv r)) \cup \ADDPREDSRHO). & \textrm{IH}\\
  \prl{sig-theta-3} &
  \pred(\forall u\, (\theta_1 \lor \theta_1))
  & \hspace{-0.5em} \textrm{by \pref{sig-theta-1}, \pref{sig-theta-2}}\\
  \prl{sig-theta-4} & 
  \pred(\theta)
    \subseteq (\pred(\SL) \cup \ADDPREDSLAM) \cap (\pred(\SR) \cup \ADDPREDSRHO)
    & \hspace{-0.5em} \textrm{by \pref{sig-theta-3}, def. }\!\theta
\end{arrayprf}
$

\medskip
For case~\ref{eq-case-g-wu} the property follows since $\pred(\theta) =
\pred(\forall w\, \theta)$.

\pagebreak
\noindent
Property~\ref{prop-fv-theta}:
\prlReset{fv-theta}

\smallskip\noindent
$
\begin{arrayprf}
  \prl{fv-theta-1} & \FV(\theta_1) \subseteq (\FV(\SL, \lnot \lambda(u, \vv l))
  \cup \ADDVARSLAM)
  \cap (\FV(\SR, \rho(u,w,\vv r)) \cup \ADDVARSRHO). & \textrm{IH}\\
  \prl{fv-theta-2}  & \FV(\theta_2) \subseteq
  (\FV(\SL, \lambda(u, \vv l)) \cup \ADDVARSLAM) \cap
  (\FV(\SR, \lnot \rho(u,w,\vv r)) \cup \ADDVARSRHO).& \textrm{IH}\\
  \prl{fv-theta-x1} & \FV(\theta_1) \subseteq (\FV(\SL) \cup \ADDVARSLAM \cup \{u\})
  \cap (\FV(\SR) \cup \ADDVARSRHO \cup \{u,w\}).
  & \textrm{by \pref{fv-theta-1}}\\
  \prl{fv-theta-x2} & \FV(\theta_1) \subseteq (\FV(\SL) \cup \ADDVARSLAM \cup \{u\})
  \cap (\FV(\SR) \cup \ADDVARSRHO \cup \{u,w\}).
  &\textrm{by \pref{fv-theta-2}}\\
  \prl{fv-theta-3} & \FV(\forall u\, (\theta_1 \lor \theta_2))
  \subseteq (\FV(\SL) \cup \ADDVARSLAM) \cap
  (\FV(\SR) \cup \ADDVARSRHO).
  & \textrm{by \pref{fv-theta-x1}, \pref{fv-theta-x2}},  \\
      &  \multicolumn{2}{r}{\textrm{and 
      $w \notin \FV(\SL) \cup \ADDVARSLAM$ or $w \in \FV(\SR) \cup \ADDVARSRHO$}}\\
  \prl{fv-theta-4} &
  \FV(\theta)
    \subseteq (\FV(\SL) \cup \ADDVARSLAM) \cap
    (\FV(\SR) \cup \ADDVARSRHO).
    & \textrm{by \pref{fv-theta-3}, def. } \theta\\
\end{arrayprf}
$

\medskip

For case~\ref{eq-case-g-wu} step~\pref{fv-theta-3} has to be replaced by
\[\FV(\forall w \forall u\, (\theta_1 \lor \theta_2)) \subseteq (\FV(\SL) \cup \ADDVARSLAM) \cap
(\FV(\SR) \cup \ADDVARSRHO),\] which follows just from \pref{fv-theta-x1} and
\pref{fv-theta-x2}. Instead of step~\pref{fv-theta-4} we then have
$\FV(\forall w\, \theta) \subseteq (\FV(\SL) \cup \ADDVARSLAM) \cap (\FV(\SR) \cup
\ADDVARSRHO)$.

\medskip
\noindent
Property~\ref{prop-form-d}: Immediate from the induction hypothesis and the
definition of $\D$.

\medskip\noindent
Property~\ref{prop-sig-d}:

\prlReset{sig-d}

\smallskip\noindent
$
\begin{arrayprf}
  \prl{sig-d-1} &
  \PREDD(\D_1) \subseteq (\pred(\SL, \lnot \lambda(u, \vv l)) \cup
  \ADDPREDSLAM) \cap (\pred(\SR, \rho(u,w,\vv r)) \cup \ADDPREDSRHO). &  \textrm{IH}\\
  \prl{sig-d-2} &
  \PREDD(\D_2) \subseteq (\pred(\SL, \lambda(u, \vv l)) \cup
  \ADDPREDSLAM) \cap (\pred(\SR, \lnot \rho(u,w,\vv r)) \cup \ADDPREDSRHO). &  \textrm{IH}\\
  \prl{sig-d-3} &
  \pred(\theta_2) \subseteq (\pred(\SL,\lambda(u, \vv l)) \cup \ADDPREDSLAM) \cap
  (\pred(\SR,\lnot \rho(u,w,\vv r)) \cup \ADDPREDSRHO). & \textrm{IH}\\
  \prl{sig-d-4} &
  \PREDD(\exdef u\, \D_1 \lor \exdef u\, \D_2 \lor
  \forall z\, (\lambda(z, \vv l) \equi
  \theta_2[z/u]))\; \subseteq\\
  & (\pred(\SL) \cup \ADDPREDSLAM) \cap (\pred(\SR) \cup \ADDPREDSRHO)
  & \hspace{-3.0em} \textrm{by \pref{sig-d-1}--\pref{sig-d-3}}\\
  \prl{sig-d-5} &
  \PREDD(\D) \subseteq
  (\pred(\SL) \cup \ADDPREDSLAM) \cap (\pred(\SR) \cup \ADDPREDSRHO)
  & \hspace{-3.0em} \textrm{by \pref{sig-d-4}, def. } \D
\end{arrayprf}
$

\medskip
For case~\ref{eq-case-g-wu} the property follows since $\PREDD(\D) =
\PREDD(\exdef w\, \D)$.

\medskip

\medskip\noindent
Property~\ref{prop-fv-d}:

\prlReset{fv-d}

\smallskip\noindent
$
\begin{arrayprf}
  \prl{fv-d-1} &
  \ADDLVARS(\D_1) \subseteq (\FV(\SL, \lnot \lambda(u, \vv l)) \cup
  \ADDVARSLAM) \cap (\FV(\SR, \rho(u,w,\vv r)) \cup \ADDVARSRHO). %
  & \textrm{IH}\\
  \prl{fv-d-2} &
  \ADDLVARS(\D_2) \subseteq (\FV(\SL, \lambda(u, \vv l)) \cup
  \ADDVARSLAM) \cap (\FV(\SR, \lnot \rho(u,w,\vv r)) \cup \ADDVARSRHO).
  & \textrm{IH}\\
  \prl{fv-d-3} &
  \FV(\theta_2) \subseteq
  (\FV(\SL, \lambda(u, \vv l)) \cup \ADDVARSLAM) \cap
  (\FV(\SR, \lnot \rho(u,w,\vv r)) \cup \ADDVARSRHO).
  & \textrm{IH}\\
  \prl{fv-d-1x} &
  \ADDLVARS(\D_1) \subseteq (\FV(\SL) \cup
  \ADDVARSLAM \cup \{u\}) \cap (\FV(\SR) \cup \ADDVARSRHO \cup \{u,w\}).
  & \textrm{by \pref{fv-d-1}}\\
  \prl{fv-d-2x} &
  \ADDLVARS(\D_2) \subseteq (\FV(\SL) \cup
  \ADDVARSLAM \cup \{u\}) \cap (\FV(\SR) \cup \ADDVARSRHO \cup \{u,w\}).
  & \textrm{by \pref{fv-d-2}}\\
  \prl{fv-d-3x} &
  \FV(\theta_2) \subseteq (\FV(\SL) \cup
  \ADDVARSLAM \cup \{u\}) \cap (\FV(\SR) \cup \ADDVARSRHO \cup \{u,w\}).
  & \textrm{by \pref{fv-d-3}}\\
  \prl{fv-d-4} &
  \ADDLVARS(\exdef u\, \D_1 \lor \exdef u\, \D_2 \lor
  \forall z\, (\lambda(z, \vv l) \equi
  \theta_2[z/u]))\; \subseteq\\
  & (\FV(\SL) \cup \ADDVARSLAM) \cap (\FV(\SR) \cup \ADDVARSRHO).
  & \textrm{by \pref{fv-d-1x}, \pref{fv-d-2x}, \pref{fv-d-3x},}\\
  & \multicolumn{2}{r}{\textrm{and $w \notin \FV(\SL) \cup \ADDVARSLAM$
  or $w \in \FV(\SR) \cup \ADDVARSRHO$}}\\
  \prl{fv-d-5} &
  \ADDLVARS(\D) \subseteq (\FV(\SL) \cup \ADDVARSLAM) \cap (\FV(\SR) \cup \ADDVARSRHO).
  & \textrm{by \pref{fv-d-4}, def. } \D
\end{arrayprf}
$

\medskip

For case~\ref{eq-case-g-wu} step~\pref{fv-d-4} has to be replaced by
\[\begin{array}{l}
\ADDLVARS(\exdef w (\exdef u\, \D_1 \lor \exdef u\, \D_2 \lor
\forall z\, (\lambda(z, \vv l) \equi \theta_2[z/u])))\; \subseteq\\
(\FV(\SL) \cup \ADDVARSLAM) \cap (\FV(\SR) \cup \ADDVARSRHO),
  \end{array}
  \] which follows just from
  \pref{fv-d-1x}, \pref{fv-d-2x}, \pref{fv-d-3x}. Instead of
  step~\pref{fv-d-5} we then have $\ADDLVARS(\exdef w\, \D) \subseteq (\FV(\SL)
  \cup \ADDVARSLAM) \cap (\FV(\SR) \cup \ADDVARSRHO)$.

  \medskip
  
This completes the verification of correctness, and thus ends our discussion of this
case.

We now turn to the base of the induction along with the other inductive cases.

\paragraph*{Cases where the principal formulas are in the
  $\SL$ or $\SR$ partition.}

The inductive cases where the principal formulas are in the $\SL$ or $\SR$
partition can be conveniently specified in the form of rules that lead from
induction hypotheses of the form $\seq \SL\sep \SR\sep \cG \ipol{\theta}{\D}$
as premises to an induction conclusion of the same form. Base cases can 
be taken   as rules without premises. The axioms and rules shown below
correspond to  those of the calculus, but replicated for each possible way
in which the partitions $\SL$, $\SR$ or $\cG$ of the conclusion can contain the
principal formula(s). To verify that
properties~\ref{prop-r-theta}--\ref{prop-fv-d} are preserved by each of the
constructions is straightforward, and therefore  we only 
point out a few subtleties.

\begin{enumerate}[(1)]
\item
  $\begin{array}{rcl}
    \tabrulename{\rulename{Ax}} && \hspace{2em} \tabrulecond{\phi \textit{ a positive literal}}\\
    \cmidrule{2-2} & \seq \SL, \phi, \lnot \phi\sep  \SR\sep \cG
    \ipol{\top}{\bot}\\
  \end{array}$
  \smallskip

\item
  $\begin{array}{rcl}
    \tabrulename{\rulename{Ax}} && \hspace{2em} \tabrulecond{\phi \textit{ a positive literal}}\\
  \cmidrule{2-2} & \seq \SL\sep  \SR, \phi, \lnot
  \phi\sep \cG \ipol{\bot}{\bot}
\end{array}$
  \smallskip
  
\item
  $\begin{array}{rcl}
    \tabrulename{\rulename{Ax}} && \hspace{2em} \tabrulecond{\phi \textit{ a literal}}\\
   \cmidrule{2-2} & \seq \SL, \phi\sep  \SR, \lnot \phi\sep \cG
   \ipol{\phi}{\bot}
\end{array}$
  \smallskip

\item
  $\begin{array}{rc}
    \tabrulename{$\top$}\\    
    \cmidrule{2-2} & \seq \SL, \top\sep  \SR\sep \cG
    \ipol{\top}{\bot}
\end{array}$
  \smallskip

\item  
  $\begin{array}{rc}
    \tabrulename{$\top$}\\    
    \cmidrule{2-2} & \seq \SL\sep  \SR, \top\sep \cG
    \ipol{\bot}{\bot}
\end{array}$
  \bigskip

\item
  $\begin{array}[t]{rc}
    \tabrulename{$\lor$} &
    \seq \SL, \phi_1, \phi_2\sep  \SR\sep \cG \ipol{\theta}{\D}\\
    \cmidrule{2-2} & \seq \SL, \phi_1 \lor \phi_2\sep  \SR\sep \cG
  \ipol{\theta}{\D}
\end{array}$
  \bigskip

\item
  $\begin{array}[t]{rc}
    \tabrulename{$\lor$} &
    \seq \SL\sep  \SR, \phi_1, \phi_2 \sep \cG \ipol{\theta}{\D}\\
    \cmidrule{2-2} & \seq \SL\sep  \SR, \phi_1 \lor \phi_2\sep \cG
    \ipol{\theta}{\D}
\end{array}$
  \bigskip

\item
  \label{rule-and-1}
  $\begin{array}[t]{rc@{\hspace{2em}}c}
    \tabrulename{$\land$} &
    \seq \SL, \phi_1 \sep  \SR\sep \cG \ipol{\theta_1}{\D_1}
  & \seq \SL, \phi_2\sep  \SR\sep \cG \ipol{\theta_2}{\D_2}\\\cmidrule{2-3} &
  \multicolumn{2}{c}{\seq \SL, \phi_1 \land \phi_2\sep  \SR\sep \cG
    \ipol{\theta_1 \land \theta_2}{\D_1 \lor \D_2}}
\end{array}$
  \bigskip

\item
  $\begin{array}[t]{rc@{\hspace{2em}}c}
    \tabrulename{$\land$} &    
    \seq \SL \sep  \SR, \phi_1\sep \cG \ipol{\theta_1}{\D_1}
  & \seq \SL \sep \SR, \phi_2\sep \cG \ipol{\theta_2}{\D_2}\\\cmidrule{2-3} &
  \multicolumn{2}{c}{\seq \SL\sep  \SR, \phi_1 \land \phi_2 \sep \cG
 \ipol{\theta_1 \lor \theta_2}{\D_1 \lor \D_2}}
\end{array}$
  \bigskip

\item  \label{rule-exists-1}
  $\begin{array}[t]{rc}
    \tabrulename{$\exists$} &    
    \seq \SL, \phi\subst{x}{t}, \exists x\, \phi\sep  \SR\sep \cG
  \ipol{\theta}{\D}\\\cmidrule{2-2} &
  \seq \SL, \exists x\, \phi\sep  \SR\sep \cG \ipol{\theta'}{\D'},
\end{array}$ \par
  where the values of $\theta'$ and $\D'$ depend on occurrences of~$t$:
  \begin{itemize}
  \item If $t \in \FV(\SL,\exists x\, \phi) \cup \ADDVARSLAM$, then $\theta' \eqdef
    \theta$ and $\D' \eqdef \D$.
  \item Else it holds that $t \notin \FV(\SL,\exists x\, \phi) \cup \ADDVARSLAM$.
    Then $\theta' \eqdef \exists t\, \theta$ and $\D' \eqdef \exdef t\, \D$.
  \end{itemize}
  \bigskip

\item \label{eq-case-ex-r}
  $\begin{array}[t]{rc}
    \tabrulename{$\exists$} &    
    \seq \SL\sep  \SR, \phi\subst{x}{t}, \exists x\, \phi\sep \cG \ipol{\theta}{\D}\\
  \cmidrule{2-2} & \seq \SL\sep \SR, \exists x\, \phi \sep \cG
  \ipol{\theta'}{\D'},
\end{array}$ \par
  where the values of $\theta'$ and $\D'$ depend on occurrences of~$t$:
  \begin{itemize}
  \item If $t \in \FV(\SR,\exists x\, \phi) \cup \ADDVARSRHO$, then
    $\theta' \eqdef \theta$ and $\D' \eqdef \D$.
  \item Else it holds that $t \notin \FV(\SR,\exists x\, \phi) \cup \ADDVARSRHO$. Then
    $\theta' \eqdef \forall t\, \theta$ and $\D' \eqdef \exdef t\, \D$.
  \end{itemize}
  \bigskip

\item
  $\begin{array}[t]{rcl}
    \tabrulename{$\forall$} &    
    \seq \SL, \phi[y/x]\sep  \SR\sep \cG \ipol{\theta}{\D}
  & \hspace{2em} \tabrulecond{y \notin \FV(\SL,\forall x\, \phi,\SR,\cG)}\\
    \cmidrule{2-2} & \seq \SL, \forall x\, \phi\sep  \SR\sep \cG
    \ipol{\theta}{\D}
\end{array}$
  \bigskip
  
\item
  $\begin{array}[t]{rcl}
    \tabrulename{$\forall$} &    
    \seq \SL\sep  \SR,\phi[y/x]\sep \cG \ipol{\theta}{\D}
  & \hspace{2em} \tabrulecond{y \notin \FV(\SL,\SR,\forall x\, \phi,\cG)}\\
  \cmidrule{2-2} & \seq \SL\sep  \SR, \forall x\, \phi\sep \cG
 \ipol{\theta}{\D}
\end{array}$
  \bigskip

\item \label{eq-rule-ref-l}
  $\begin{array}[t]{rcl}
    \tabrulename{\rulename{Ref}} &    
    \seq t \neq t, \SL\sep \SR\sep \cG \ipol{\theta}{\D}
    & \hspace{2em} \tabrulecond{t \in \FV(\SL) \cup \ADDVARSLAM}\\
    \cmidrule{2-2} &
    \seq \SL\sep \SR\sep \cG \ipol{\theta}{\D}
\end{array}$
  \bigskip
  
\item \label{eq-rule-ref-r}
  $\begin{array}[t]{rcl}
    \tabrulename{\rulename{Ref}} &    
    \seq \SL\sep t \neq t, \SR\sep \cG \ipol{\theta}{\D}
    & \hspace{2em} \tabrulecond{t \notin \FV(\SL) \cup \ADDVARSLAM}\\    
    \cmidrule{2-2} &
    \seq \SL\sep \SR\sep \cG \ipol{\theta}{\D}
  \end{array}$
  \bigskip
  
\item
  $\begin{array}[t]{rcl}
    \tabrulename{\rulename{Repl}} &    
    \seq t \neq u, \phi[u/x], \phi[t/x], \SL\sep \SR\sep \cG \ipol{\theta}{\D}
    & \hspace{2em} \tabrulecond{\phi \textit{ a negative literal}}
    \\\cmidrule{2-2} &
    \seq t \neq u, \phi[t/x], \SL\sep \SR\sep \cG \ipol{\theta}{\D}
\end{array}$
  \bigskip

\item
  $\begin{array}[t]{rcl}
  \tabrulename{\rulename{Repl}} &    
  \seq t \neq u, \SL\sep \phi[u/x], \phi[t/x], \SR\sep \cG \ipol{\theta}{\D}
  & \hspace{2em} \tabrulecond{\phi \textit{ a negative literal}}
  \\\cmidrule{2-2} &
  \seq t \neq u, \SL\sep \phi[t/x], \SR\sep \cG \ipol{\theta'}{\D'},
\end{array}$ \par
  where the values of $\theta'$ and $\D'$ depend on occurrences of $t$ and $u$:
  \begin{itemize}
  \item If $t \notin \FV(\phi[t/x], \SR) \cup \ADDVARSRHO$, then $\theta' \eqdef
    \theta$ and $\D' \eqdef \D$. In this subcase the precondition $t
    \notin \FV(\phi[t/x])$ implies that $x \notin \FV(\phi)$ and thus
    $\phi[u/x] = \phi[t/x]$.
  \item If $t,u \in \FV(\phi[t/x], \SR) \cup \ADDVARSRHO$, then $\theta' \eqdef
    \theta \lor t \neq u$ and $\D' \eqdef \D$.
  \item Else it holds that $t \in \FV(\phi[t/x], \SR) \cup \ADDVARSRHO$ and $u
    \notin \FV(\phi[t/x], \SR) \cup \ADDVARSRHO$. Then $\theta' \eqdef \theta[t/u]$
    and $\D' \eqdef \D\substdef{u}{t}$. For this subcase, to derive
    property~\ref{prop-r-theta} it is used that the precondition $u \notin
    \phi[t/x]$ implies that $\phi[u/x][t/u] = \phi[t/x]$.
  \end{itemize}
  \medskip
  
\item
  $\begin{array}[t]{rcl}
    \tabrulename{\rulename{Repl}} &    
    \seq \SL\sep t \neq u, \phi[u/x], \phi[t/x], \SR\sep \cG \ipol{\theta}{\D}
    & \hspace{2em} \tabrulecond{\phi \textit{ a negative literal}}
    \\\cmidrule{2-2} &
    \seq \SL\sep t \neq u, \phi[t/x], \SR\sep \cG \ipol{\theta}{\D}
  \end{array}$
  \bigskip

\item
  $\begin{array}[t]{rcl}
    \tabrulename{\rulename{Repl}} &    
    \seq \phi[u/x], \phi[t/x], \SL\sep t \neq u; \SR\sep \cG \ipol{\theta}{\D}
    & \hspace{2em} \tabrulecond{\phi \textit{ a negative literal}}
    \\\cmidrule{2-2} &
    \seq \phi[t/x], \SL\sep t \neq u, \SR\sep \cG \ipol{\theta'}{\D'},
  \end{array}$ \par
  where the values of $\theta'$ and $\D'$ depend on occurrences of $t$ and $u$:
  \begin{itemize}
  \item If $t \notin \FV(\phi[t/x], \SL) \cup \ADDVARSLAM$, then $\theta' \eqdef
    \theta$ and $\D' \eqdef \D$. In this subcase the precondition $t
    \notin \FV(\phi[t/x])$ implies that $x \notin \FV(\phi)$ and thus
    $\phi[u/x] = \phi[t/x]$.
  \item If $t,u \in \FV(\phi[t/x], \SL) \cup \ADDVARSLAM$, then $\theta' \eqdef
    \theta \land t = u$ and $\D' \eqdef \D$.
  \item Else it holds that $t \in \FV(\phi[t/x], \SL) \cup \ADDVARSLAM$ and $u
    \notin \FV(\phi[t/x], \SL) \cup \ADDVARSLAM$. Then $\theta' \eqdef \theta[t/u]$
    and $\D' \eqdef \D\substdef{u}{t}$. To derive property~\ref{prop-l-theta}
    for this subcase, that is, $\valid \lnot \theta[t/u] \lor \phi[t/x], \SL
    \lor \D\substdef{u}{t}$, it is required that $u \notin
    \FV(\D\substdef{u}{t})$, which follows from the precondition $u \notin
    \ADDVARSLAM$ of the subcase.
  \end{itemize}
\end{enumerate}

This completes the proof of Theorem~\ref{thm:mainflatcase}.

%% file: bethjournal.bbl
\newcommand{\etalchar}[1]{$^{#1}$}
\begin{thebibliography}{GHHW18}

\bibitem[AMN08]{madarasz}
H.~Andr\'eka, J.~X. Madar\'asz, and I.~N\'emeti.
\newblock Definability of new universes in many-sorted logic, 2008.
\newblock manuscript available at
  \url{old.renyi.hu/pub/algebraic-logic/kurzus10/amn-defi.pdf}.

\bibitem[BBtC18]{gnfjsl}
Vince B{\'{a}}r{\'{a}}ny, Michael Benedikt, and Balder ten Cate.
\newblock Some model theory of guarded negation.
\newblock {\em J. Symb. Log.}, 83(4):1307--1344, 2018.

\bibitem[BBV19]{lmcsusinterpolfixedpoint}
Michael Benedikt, Pierre Bourhis, and Michael {Vanden Boom}.
\newblock Definability and interpolation within decidable fixpoint logics.
\newblock {\em Log. Methods Comput. Sci.}, 15(3):29:1–29:53, 2019.

\bibitem[BCLT16]{interpbook}
Michael Benedikt, Balden~Ten Cate, Julien Leblay, and Efthymia Tsamoura.
\newblock {\em Generating Plans from Proofs: The Interpolation-Based Approach
  to Query Reformulation}.
\newblock Morgan Claypool, San Rafael, CA, 2016.

\bibitem[BDK18]{interpmikolaj}
Mikolaj Bojanczyk, Laure Daviaud, and Shankara~Narayanan Krishna.
\newblock Regular and first-order list functions.
\newblock In {\em LICS}, 2018.

\bibitem[Bet53]{beth}
E.~W. Beth.
\newblock On {Padoa's} method in the theory of definitions.
\newblock {\em Indag. Mathematicae}, 15:330 -- 339, 1953.

\bibitem[BK09]{xqueryinterp}
Michael Benedikt and Christoph Koch.
\newblock From {XQuery} to {Relational} {Logics}.
\newblock {\em ACM TODS}, 34(4):25:1--25:48, 2009.

\bibitem[BKMT17]{usijcai17}
Michael Benedikt, Egor~V. Kostylev, Fabio Mogavero, and Efthymia Tsamoura.
\newblock Reformulating queries: Theory and practice.
\newblock In {\em IJCAI}, 2017.

\bibitem[BNTW95]{natj}
Peter Buneman, Shamim~A. Naqvi, Val Tannen, and Limsoon Wong.
\newblock Principles of programming with complex objects and collection types.
\newblock {\em Theor. Comput. Sci.}, 149(1):3--48, 1995.

\bibitem[BP21]{benediktpradicpopl}
Michael Benedikt and Cécilia Pradic.
\newblock Generating collection transformations from proofs.
\newblock In {\em POPL}, 2021.

\bibitem[BPW23]{bethpods}
Michael Benedikt, C\'{e}cilia Pradic, and Christoph Wernhard.
\newblock Synthesizing nested relational queries from implicit specifications.
\newblock In {\em PODS}, pages 33--45, 2023.

\bibitem[BtCV16]{guardedinterpj}
Michael Benedikt, Balder ten Cate, and Michael {Vanden Boom}.
\newblock Effective interpolation and preservation in guarded logics.
\newblock {\em {ACM} TOCL}, 17(2):8:1--8:46, 2016.

\bibitem[CH99]{modeltheorydiff}
Zo\'e Chatzadakis and Ehud Hrushovski.
\newblock Model theory of difference fields.
\newblock {\em Transactions of the American Mathematical Society},
  351:2997--3071, 1999.

\bibitem[Cha64]{changbeth}
C.~C. Chang.
\newblock Some new results in definability.
\newblock {\em Bull. of the AMS}, 70(6):808 -- 813, 1964.

\bibitem[CK92]{ck}
C.~C. Chang and H.~Jerome Keisler.
\newblock {\em Model Theory}.
\newblock North-Holland, 1992.

\bibitem[CL07]{interpcolcombet}
Thomas Colcombet and Christof L{\"{o}}ding.
\newblock Transforming structures by set interpretations.
\newblock {\em Logical Methods in Computer Science}, 3(2), 2007.

\bibitem[Cra57a]{craig57interp}
William Craig.
\newblock Linear reasoning. a new form of the {Herbrand-Gentzen} theorem.
\newblock {\em J. Symb. Log.}, 22(03):250--268, 1957.

\bibitem[Cra57b]{craig57beth}
William Craig.
\newblock Three uses of the {Herbrand-Gentzen} theorem in relating model theory
  and proof theory.
\newblock {\em J. Symb. Log.}, 22(3):269--285, 1957.

\bibitem[DH00]{interpolationmucalc}
Giovanna D'Agostino and Marco Hollenberg.
\newblock {Logical Questions Concerning The mu-Calculus: Interpolation, Lyndon
  and Los-Tarski}.
\newblock {\em J. Symb. Log.}, 65(1):310--332, 2000.

\bibitem[Fit96]{fittingbook}
Melvin Fitting.
\newblock {\em First-order Logic and Automated Theorem Proving}.
\newblock Springer, second edition, 1996.

\bibitem[FKN13]{franconisafe}
Enrico Franconi, Volha Kerhet, and Nhung Ngo.
\newblock Exact query reformulation over databases with first-order and
  description logics ontologies.
\newblock {\em J. Artif. Int. Res.}, 48:885--922, 2013.

\bibitem[Gai74]{gaifman74}
Haim Gaifman.
\newblock Operations on relational structures, functors and classes {I}.
\newblock In {\em Proc. of the {T}arski {S}ymposium}, volume~25 of {\em Proc.
  of Symposia in Pure Mathematics}, pages 20--40, 1974.

\bibitem[GHHW18]{gibbonshenglein}
Jeremy Gibbons, Fritz Henglein, Ralf Hinze, and Nicolas Wu.
\newblock Relational algebra by way of adjunctions.
\newblock {\em {PACMPL}}, 2({ICFP}), 2018.

\bibitem[Gib16]{ringads}
Jeremy Gibbons.
\newblock {Comprehending Ringads} - for {Phil Wadler}, on the occasion of his
  60th birthday.
\newblock In {\em A List of Successes That Can Change the World - Essays
  Dedicated to {Philip Wadler} on the Occasion of His 60th Birthday}, 2016.

\bibitem[HHM90]{hodgesdugald}
Wilfrid Hodges, I.M. Hodkinson, and Dugald Macpherson.
\newblock Omega-categoricity, relative categoricity and coordinatisation.
\newblock {\em Annals of Pure and Applied Logic}, 46(2):169 -- 199, 1990.

\bibitem[HMO99]{HMO}
Eva Hoogland, Maarten Marx, and Martin Otto.
\newblock Beth definability for the guarded fragment.
\newblock In {\em LPAR}, 1999.

\bibitem[Hod75]{hodgesnormal}
Wilfrid Hodges.
\newblock A normal form for algebraic constructions {II}.
\newblock {\em Logique et Analyse}, 18(71/72):429--487, 1975.

\bibitem[Hod93]{hodgesbook}
Wilfrid Hodges.
\newblock {\em Model Theory}.
\newblock Cambridge University Press, 1993.

\bibitem[Hru14]{groupoids}
Ehud Hrushovski.
\newblock Groupoids, imaginaries and internal covers.
\newblock {\em Turkish Journal of Mathematics}, 36:173 -- 198, 2014.

\bibitem[Hua95]{huang}
Guoxiang Huang.
\newblock Constructing {Craig} interpolation formulas.
\newblock In {\em Computing and Combinatorics}. 1995.

\bibitem[Jec03]{jech}
Thomas Jech.
\newblock {\em Set Theory}.
\newblock Springer, 2003.

\bibitem[Kle52]{kleene:permutability}
S.~C. Kleene.
\newblock Permutability of inferences in {G}entzen's calculi lk and lj.
\newblock {\em Memoirs of the American Mathematical Society}, 10:1--26, 1952.

\bibitem[Koc06]{koch}
Christoph Koch.
\newblock {On the Complexity of Non-recursive XQuery and Functional Query
  Languages on Complex Values}.
\newblock {\em ACM TODS}, 31(4):1215--1256, 2006.

\bibitem[Kol90]{kolaitisimpdef}
Phokion~G. Kolaitis.
\newblock Implicit definability on finite structures and unambiguous
  computations.
\newblock In {\em LICS}, 1990.

\bibitem[Kue71]{kueker}
David Kueker.
\newblock Generalized interpolation and definability.
\newblock {\em Annals of Mathematical Logic}, 1(4):423--468, 1971.

\bibitem[LE65]{lopezescobar}
E.~G.~K. Lopez-Escobar.
\newblock An interpolation theorem for denumerably long sentences.
\newblock {\em Fundamenta Mathametica}, 57:253--272, 1965.

\bibitem[Lyn59]{Lyndon59}
Roger~C. Lyndon.
\newblock An interpolation theorem in the predicate calculus.
\newblock {\em Pacific J. Math.}, 9:129--142, 1959.

\bibitem[Mak64]{makkai}
Michael Makkai.
\newblock On a generalization of a theorem of {E. W. Beth}.
\newblock {\em {Acta Math. Ac. Sci. Hung.}}, 15:227--235, 1964.

\bibitem[MBB06]{linq}
Erik Meijer, Brian Beckman, and Gavin Bierman.
\newblock {LINQ}: Reconciling object, relations and {XML} in the {.NET}
  framework.
\newblock In {\em SIGMOD}, 2006.

\bibitem[MGL{\etalchar{+}}10]{dremel}
Sergey Melnik, Andrey Gubarev, Jing~Jing Long, Geoffrey Romer, Shiva
  Shivakumar, Matt Tolton, and Theo Vassilakis.
\newblock {Dremel: Interactive Analysis of Web-Scale Datasets}.
\newblock {\em {PVLDB}}, 3(1-2):330--339, 2010.

\bibitem[Mos49]{mostowski}
Andrzej Mostowski.
\newblock An undecidable arithmetical statement.
\newblock {\em Fundamenta Mathematicae}, 36(1):143–164, 1949.

\bibitem[NSV10]{NSV}
Alan Nash, Luc Segoufin, and Victor Vianu.
\newblock Views and queries: Determinacy and rewriting.
\newblock {\em ACM TODS}, 35(3):1--41, 2010.

\bibitem[NvP98]{negri:plato:equality:1998}
Sara Negri and Jan von Plato.
\newblock Cut elimination in the presence of axioms.
\newblock {\em Bull. Symb. Log.}, 4(4):418--435, 1998.

\bibitem[NvP01]{negri:plato:structural:2001}
Sara Negri and Jan von Plato.
\newblock {\em Structural Proof Theory}.
\newblock Cambridge University Press, 2001.

\bibitem[Ott00]{otto}
Martin Otto.
\newblock An interpolation theorem.
\newblock {\em Bull. Symb. Log.}, 6(4):447--462, 2000.

\bibitem[Smu68a]{smullyanbook}
Raymond Smullyan.
\newblock {\em First Order Logic}.
\newblock Springer, 1968.

\bibitem[Smu68b]{smullyan}
Raymond~M. Smullyan.
\newblock {\em Craig's Interpolation Lemma and Beth's Definability Theorem.
  \textnormal{In: } First-Order Logic}, pages 127--133.
\newblock Springer, 1968.

\bibitem[Suc95]{suciuthesis}
Dan Suciu.
\newblock {\em Parallel Programming Languages for Collections}.
\newblock PhD thesis, Univ. Pennsylvania, 1995.

\bibitem[SV05]{svconf}
Luc Segoufin and Victor Vianu.
\newblock Views and queries: Determinacy and rewriting.
\newblock In {\em PODS}, 2005.

\bibitem[Tak87]{takeuti:1987}
Gaisi Takeuti.
\newblock {\em Proof Theory}.
\newblock North-Holland, second edition, 1987.

\bibitem[tCFS13]{balderbethdl}
Balder ten Cate, Enrico Franconi, and Inan\c{c} Seylan.
\newblock Beth definability in expressive description logics.
\newblock {\em J. Artif. Int. Res.}, 48(1):347--414, 2013.

\bibitem[TS00]{bpt}
Arne~S. Troelstra and Helmut Schwichtenberg.
\newblock {\em Basic Proof Theory}.
\newblock Cambridge University Press, 2000.

\bibitem[TW11]{tomanweddell}
David Toman and Grant Weddell.
\newblock {\em Fundamentals of Physical Design and Query Compilation}.
\newblock Morgan Claypool, 2011.

\bibitem[{Van}01]{simulation}
Jan {Van den Bussche}.
\newblock {Simulation of the Nested Relational Algebra by the Flat Relational
  Algebra, with an Application to the Complexity of Evaluating Powerset Algebra
  Expressions}.
\newblock {\em Theor. Comput. Sci.}, 254(1--2):363--377, 2001.

\bibitem[Won94]{limsoonthesis}
Limsoon Wong.
\newblock {\em Querying Nested Collections}.
\newblock PhD thesis, Univ. Pennsylvania, 1994.

\end{thebibliography}
